\documentclass[11pt,a4paper]{article}
\usepackage[a4paper]{geometry}
\usepackage{amssymb,latexsym,mathtools,amsfonts,amsthm}
\usepackage{graphicx}
\usepackage{amsbsy}
\usepackage{hyperref}
\hypersetup{
    colorlinks=true,    
    citecolor=brown,      
	linkcolor=blue,      
}
\usepackage{fullpage}
\usepackage{enumitem}
\usepackage{color}
\usepackage{tikz}
\usepackage{overpic}
\usepackage{comment}
\usepackage{cancel}
\usepackage{subfigure}
\usepackage{booktabs,longtable}
\usepackage[alphabetic,initials]{amsrefs}

\newcommand{\C}{\mathbb{C}}
\newcommand{\realR}{\mathbb{R}}

\newcommand{\E}{\mathbb{E}}

\newcommand{\Gvar}{\mathcal{G}}

\newcommand{\CH}{\mathrm{CH}}

\newcommand{\vs}{\mathrm{s}}
\newcommand{\vr}{\mathrm{r}}
\newcommand{\ii}{\mathrm{i}}
\newcommand{\Boh}{\mathcal{O}}
\newcommand{\tac}{\mathrm{tac}}
\newcommand*{\dif}{\mathop{}\!\mathrm{d}}

\DeclareMathOperator{\diag}{diag}

\DeclareMathOperator*{\res}{Res}
\DeclareMathOperator*{\im}{Im}
\DeclareMathOperator*{\re}{Re}

\newtheorem{theorem}{Theorem}[section]

\newtheorem{proposition}[theorem]{Proposition}

\newtheorem{corollary}[theorem]{Corollary}

\newtheorem{dbar-RHP}[theorem]{$\bar{\partial}$-RH problem}

\theoremstyle{remark}
\newtheorem{remark}[theorem]{Remark}
\theoremstyle{definition}
\numberwithin{equation}{section}
\allowdisplaybreaks[2]
\renewcommand{\thefootnote}{\fnsymbol{footnote}}

\begin{document}

\title{Moment generating function of the tacnode process}

\author{Taiyang Xu\footnotemark[1]}

\renewcommand{\thefootnote}{\fnsymbol{footnote}}
\footnotetext[1]{School of Mathematical Sciences, Fudan University, Shanghai 200433, China.
E-mail: \texttt{tyxu19@fudan.edu.cn.}}
\date{\today}
\maketitle

\begin{abstract}
The tacnode process is a universal determinantal point process arising in non-intersecting particle systems and random tiling models. 
In this paper, we study the generating function for the counting functions of the tacnode process on a union of $m$ intervals, $m\in\mathbb{N}^{+}$. 
Our first result provides an integral representation for the $m$-point generating function in terms of the Hamiltonian governing a 
system of $8m+4$ coupled differential equations. Combined with several differential identities for this Hamiltonian, the representation yields the large gap asymptotics, 
up to and including the constant term. As further applications, we obtain asymptotic formulae for the expectations, variances, and covariances of the counting functions, and establish a central limit theorem for their joint fluctuations.
These results extend the previously known $1$-point theory for the tacnode process to the multi-interval setting with multiple discontinuities.
\end{abstract}

\setcounter{tocdepth}{2} \tableofcontents

\section{Introduction}
Determinantal point processes (DPP) have attracted a lot of interest over the past 
few decades due to their rich mathematical structures, and their frequent appearances 
in various ensembles of nonintersecting particles, growth and tiling processes, random matrix models, 
and so on; see \cites{Sosh, Johnbook06} for an overview.

One of the most well-known examples for DPP is the non-intersecting Brownian motion model, 
which originates from the seminal work of Dyson on the dynamics of eigenvalues for Gaussian unitary ensemble (GUE) \cite{Dyson76},
and was extensively studied in the literatures due to its connections with a variety of physical, combinatorial and probabilistic models; 
see, for instance, \cites{Fisher,For11,GOV,John05,John02,KT07,KT04,WFS}.

In a non-intersecting Brownian motion model, one considers $n$ Brownian particles on the real line, which start from some initial positions at time $t=0$ and end at some final positions at time $t=1$, and are conditioned to never intersect with each other in the time interval $(0,1)$.
As time varies in an interval for $t\in[0, 1]$, the particles remain confined in a specific region and for every time $t\in[0, 1]$ the positions of the
Brownian paths form a DPP. Moreover, as the number of particles $n$ tends to infinity, the particles fill out such a region 
in the space-time plane with a deterministic limit shape. By fine-tuning the parameters, we
may create a situation with two groups of Brownian motions; see Figure \ref{fig:3cases} for an illustration where two starting points and two ending points are specified. 
It is well-known that the local statistics of the particles is governed 
by the sine process in the bulk of the limit shape \cites{ABK05,BK07,DKV,Joh01,LSY19,LY17}, 
by the Airy process at the edge of the limit shape \cites{ABK05,BK07,DKV,Huang,LY17b}, and by the Pearcey process
at the cusp singularity \cites{AOV10,AV07,BK07,BH98a,BH98b,TW06}; see also \cites{ADV11,ADV10,CNV20,NV22} for relevant studies.

\begin{figure}[htbp]
\centering
\begin{overpic}[scale=.24]{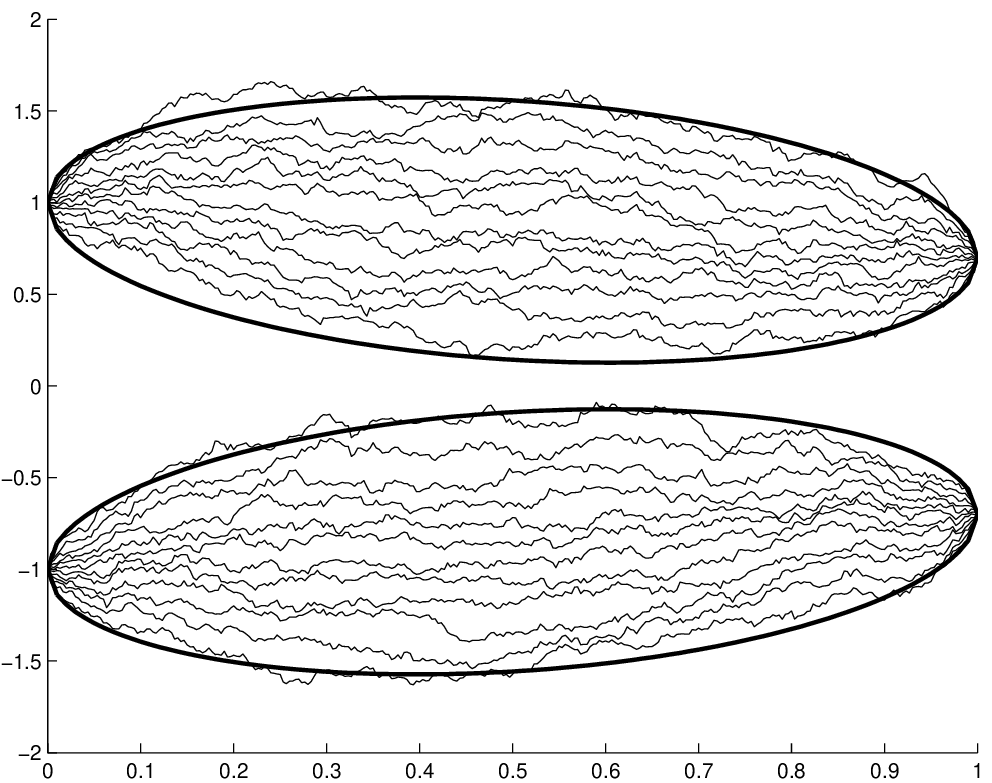}
\small{\put(50,-6){(a)}}
\end{overpic}
\hfill
\begin{overpic}[scale=.24]{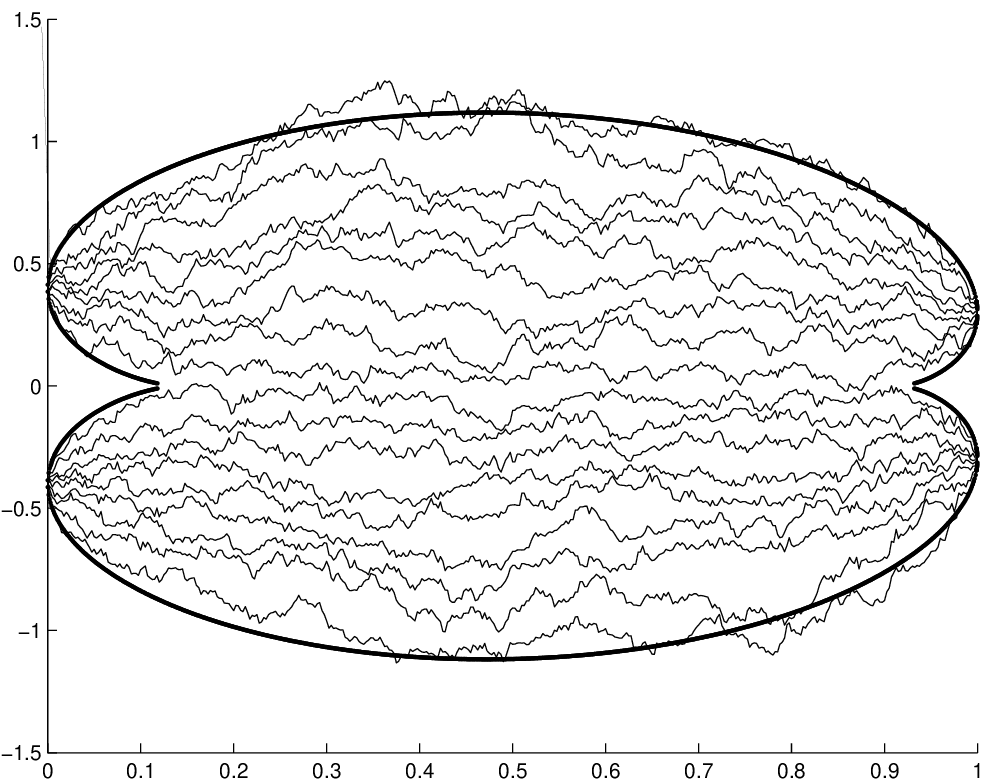}
\small{\put(50,-6){(b)}}
\end{overpic}
\hfill
\begin{overpic}[scale=.24]{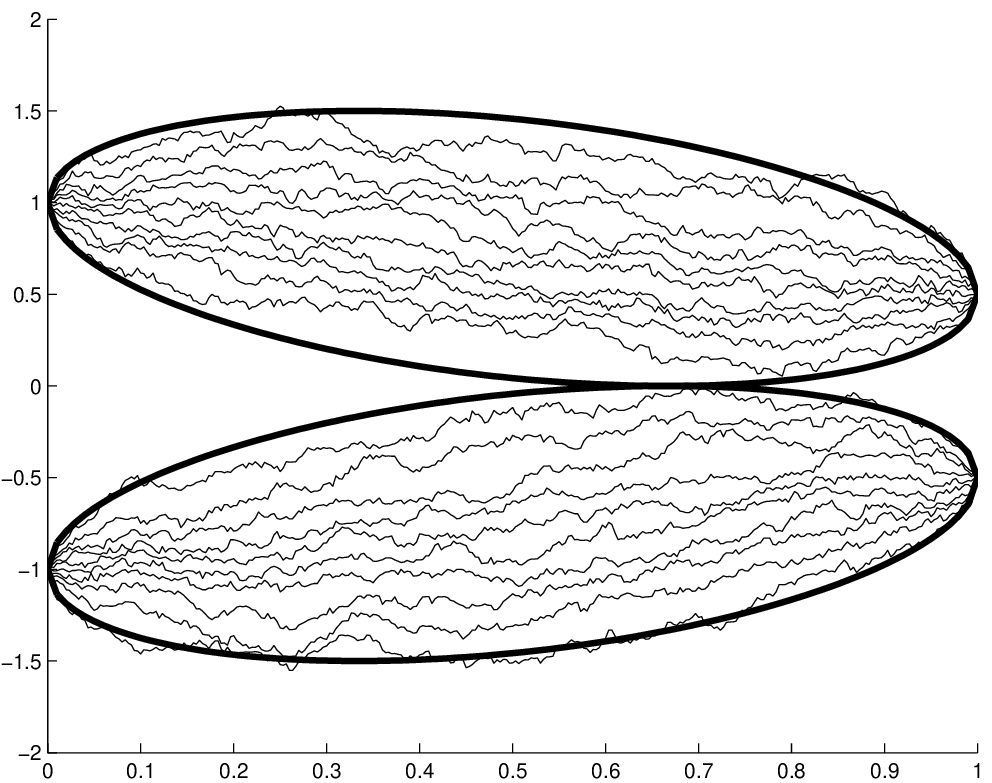}
\small{\put(50,-6){(c)}}
\end{overpic}
\caption{\small Non-intersecting Brownian
motions with two starting points $\pm P_{s}$ and two ending
positions $\pm P_{e}$ in case of (a) large, (b) small, and (c)
critical separation between the endpoints. Here the horizontal axis
denotes time, $t\in[0,1]$, and the positions of the $n$
non-intersecting Brownian motions at each fixed time $t$ are shown
on the corresponding vertical line. Note that for
$n\to\infty$ the positions of the Brownian motions fill a prescribed
region in the time-space plane, which is bounded by the boldface
lines in the figures. Here we have chosen $n=20$ in each of the
figures, and (a) $P_s=1$, $P_e=0.7$, (b) $P_s=0.4$,
$P_e=0.3$, and (c) $P_s=1$, $P_e=0.5$. Taken from \cite{GZ}.} \label{fig:3cases}
\end{figure}

When two groups of Brownian motions become tangent at a single point; see Figure \ref{fig:3cases}(c), the \textit{tacnode process} arises at the critical transition where two components of the limiting shape touch without merging, and describes the microscopic correlations of the particles near the tangency point.
As a DPP, the tacnode process is characterized by a two-point correlation kernel denoted by $K^{\mathrm{tac}}(x,y)$ called the tacnode kernel, also depending on additional parameters.
The tacnode process first emerged in the work of Adler, Ferrari, and van Moerbeke \cite{AFV13} in their analysis of continuous-time random walks on $\mathbb{Z}$.
They showed that the tacnode kernel can be expressed in terms of the resolvents and Fredholm determinants of the Airy integral operator 
acting on a semi-infinite interval. Later, Johansson \cite{John13} derived an integral representation of the tacnode kernel, 
and Ferrari and Vet\H{o} \cite{FV12} extended this representation to the non-symmetric setting in which the two touching groups of Brownian motions may have different sizes.
The equivalence of \cite{AFV13} and \cite{John13} was shown in \cite{AJV14} by computing a scaling limit of the 
statistics in the domino tiling of the double Aztec diamond.
From a different viewpoint, Delvaux, Kuijlaars and Zhang \cite{DKZ11} expressed the kernel via a $4\times4$ matrix-valued 
Riemann-Hilbert (RH) problem, which reveals a remarkable connection with the Hastings-McLeod \cite{HM} solution of the Painlev\'{e} II equation
\begin{equation}\label{eq:PII}
	q''(x)=2q(x)^3+xq(x).	
\end{equation}
We also refer to \cite{Del} for the equivalence between the RH formulation in \cite{DKZ11} and the Airy resolvent formula in \cite{John13}.
Like the sine, Airy, and Pearcey processes, the tacnode process and its variants form a universality class across probability and mathematical physics. Concrete examples include 
non-intersecting Brownian motions on the unit circle \cites{BL19,LW16} and various random tiling models \cites{AJV22,AJV14,AV23}, among others.
Transitions from the tacnode process to the Airy and Pearcey processes have also been studied. At the level of correlation kernels, the transition to the Pearcey kernel was analyzed in \cite{GZ}; at the level of gap probabilities, related transitions were studied in \cites{BC13,Gir14}. 

This work concerns the moment generating function of the tacnode process. 
Let us introduce the parameters
\begin{equation}\label{def: vecu and vecx}
	m\in\mathbb{N}^{+}, \quad \vec{u}=(u_1, \ldots, u_m)\in\mathbb{R}^m, \quad \vec{x}:=(x_1, \ldots, x_m)\in\mathbb{R}^{+,m}_{<},
\end{equation}
where $\mathbb{N}^{+}:=\{1,2,\ldots\}$ and $\mathbb{R}^{+,m}_{<}:=\{(x_1, \ldots, x_m): 0<x_1<\cdots<x_m\}$.
Let $\tilde{\chi}$ be a locally finite random point configuration distributed according to the tacnode process, 
and let 
\begin{equation}
	N(x):=\#\{\xi\in\tilde{\chi}: \ \xi\in(-x, x)\}
\end{equation} 
be the associated counting function. 
More precisely, we consider the $m$-point generating function of $N(x)$, defined by 
\begin{equation}\label{eq:generating_function}
F(\vec{x}, \vec{u}):=\E\left(\prod_{j=1}^m e^{u_jN(x_j)}\right),
\end{equation}
where the expectation is taken with respect to the tacnode process.
Furthermore, define
\begin{equation}\label{def: interval A_j}
A_1:=(-x_1,x_1), \qquad A_j:=(-x_j,-x_{j-1})\cup(x_{j-1},x_j), \quad j=2,\ldots,m,
\end{equation}
the generating function can be rewritten as
\[
	F(\vec{x},\vec{u})=\E\left(\prod_{j=1}^m s_j^{N(A_j)}\right), \quad  {\rm with} \quad s_j:=e^{u_j+\cdots+u_m}\in (0, +\infty), \quad j=1,\ldots,m,
\]
and $N(A_j):=\#(\tilde{\chi}\cap A_j)$.
As a consequence, the general theory of DPP \cites{Sosh2} implies that 
\begin{equation}\label{eq:Fredholm_determinant}
	F(\vec{x},\vec{u})=\det\left(1-\sum_{j=1}^m (1-s_j) \mathcal K^{\tac}\chi_{A_j}\right).
\end{equation}
Here $\mathcal{K}^{\rm tac}$ denotes the trace class operator acting on $L^2(\mathbb{R})$ whose
kernel is $K^{\tac}(x,y)$. This Fredholm determinant representation allows us to investigate the integrable structure of the $m$-point generating function,
and to analyze the large gap asymptotics of $F(r\vec{x},\vec{u})$ as the scaling parameter $r\to+\infty$.
\vspace{0.5em}

\noindent {\bf Integrable structure.}
The connection between Fredholm determinants and differential equations,
especially Painlev\'e equations, is one of the central themes in modern integrable probability. 
Its theoretical basis lies in the theory of isomonodromic deformations developed 
by Jimbo, Miwa, \textit{et al.} \cite{JMU81}, in which tau functions arise naturally as 
fundamental objects associated with isomonodromic deformation problems. This framework suggests that 
Fredholm determinants may inherit an integrable structure when they can be identified with suitable tau functions.
A major breakthrough in this direction was made by Tracy and Widom \cite{TW94} in their seminal work on the Airy process, 
where they showed that the Airy kernel determinant is related to the Hastings-McLeod solution of the 
Painlev\'e II equation \eqref{eq:PII}. Since then, related integrable representations have been obtained for several other point processes, including the sine, Bessel, and Pearcey processes; see \cites{JMMS,TW94b,DXZ22}. 
We also refer to \cites{CD18,CH19,CM23} for representations
of the $m$-point Airy, Bessel, and Pearcey processes in terms of systems of coupled differential equations.

For the tacnode process, a Hamiltonian representation involving a system of $12$ coupled differential equations was obtained in \cite{YZ2024} for $m=1$. 
In Section \ref{sec: coupled system}, we extend this picture to the $m$-point setting and prove that the $m$-point generating 
function can be expressed in terms of the Hamiltonian associated with a system of $8m+4$ coupled differential equations, with new cross terms reflecting the interaction between different $x_j$'s.
\vspace{0.5em}

\noindent {\bf Thinning and large gap asymptotics.} 
The operation of thinning consists in removing particles from the original point process according to some independent Bernoulli random variables, and the resulting thinned process is still a DPP with a modified correlation kernel; 
see \cite{IllianBook}. It has been firstly studied in random matrix theory by Bohigas and Pato \cites{BCP,BP04}, and first rigorously investigated in \cite{BDIK15} for sine process.
A constant thinning consists in removing each particle with the same probability $s\in (0,1)$. In the case $m=1$ and $s_1=e^{u_1}\in (0,1)$, the $1$-point generating function $F(x_1, u_1)$ defined in \eqref{eq:generating_function}
can be interpreted as the probability that there are no particles (a.k.a. \textit{gap probability}) in the interval $(-x_1,x_1)$ for the thinned tacnode process. The corresponding \textit{large gap asymptotics}, including the notoriously difficult constant term, were obtained in \cite{YZ2024}: 
\begin{multline*}
	F(x_1, u_1) = \exp\left( \frac{4 \beta_1 \ii (\vr_1 + \vr_2)}{3} x_1^{\frac 32} - 4 \beta_1 \ii (\vs_1 + \vs_2) x_1^{\frac 12} - 3 \beta_1^2 \log{x_1} \right.\\ \left.+ 2 \log{(G(1+\beta_1)G(1-\beta_1))}-\beta_1^2 \log{(64 \vr_1 \vr_2)}+ \Boh(x_1^{-\frac 12})\right), \quad x_1\to+\infty,
\end{multline*}
where $\beta_1=u_1/ 2\pi \ii $, the parameters $\vr_1$, $\vr_2$, $\vs_1$, $\vs_2$ come from the tacnode kernel, and $G$ denotes the Barnes' $G$-function \cite[Chapter 5.17]{DLMF}.
At the level of $1$-point generating functions, large gap asymptotics have been studied for several classical thinned point processes; see \cites{BW83,BCI,BB18,BIP,DXZ22} for some examples.

By contrast, a piecewise thinning consists in removing particles with different probabilities $s_j$ in different regions. 
In the present work, we construct the large gap asymptotics ($r\to+\infty$) of the piecewise thinned tacnode process for arbitrary $m\in\mathbb{N}^{+}$ and $s_j\in (0,+\infty)$, as stated in Theorem \ref{thm: large_gap} below. 
We also mention that, while the interpretation of $F(r\vec{x},\vec{u})$ as gap probability only makes sense for $s_j\in(0,1)$, we still use 
the terminology of ``large gap asymptotics'' for the case $s_j\in (0,+\infty)$.
In particular, for $m\geq 2$, the asymptotic formula contains linear terms together with a quadratic form in $u_1,\ldots,u_m$. As we show later, this structure leads to the large $r$ asymptotics of the expectations, variances and covariances of the counting functions (see Corollary \ref{cor: large r of statistics}) and to the analysis of their joint fluctuations (see Corollary \ref{cor: large r of fluctuation}).
Such exponential moments are also closely tied to the rigidity for tacnode point process according to a general result established in \cite{CC21}.
For the large gap asymptotics of other thinned point processes with $m\in\mathbb{N}^{+}$, we refer to, for instance, \cite{Charlier21a} for the sine process, \cite{CC20} for the Airy process, \cite{Charlier21c} for the Bessel process, 
\cites{Charlier21b,CM23} for the Pearcey process.

\paragraph{Notations.} 
Throughout this paper, the following notations will be used.
\begin{itemize}
  \item If $A$ is a matrix, then $(A)_{ij}$ stands for its $(i,j)$-th entry and $A^{\rm T}$ stands for its transpose. 

  \item As usual, the three Pauli matrices $\{\sigma_j\}_{j=1}^3$ are defined by
\begin{equation}\label{def:Pauli}
\sigma_1=\begin{pmatrix}
           0 & 1 \\
           1 & 0
        \end{pmatrix},
        \qquad
        \sigma_2=\begin{pmatrix}
        0 & -\ii \\
        \ii & 0
        \end{pmatrix},
        \qquad
        \sigma_3=
        \begin{pmatrix}
        1 & 0 \\
         0 & -1
         \end{pmatrix}.
\end{equation}
We use $I$ to denote an identity matrix, and the size might differ in different contexts. To emphasize a $k\times k$ identity matrix, we also use the notation $I_k$. 
  
 \item It is convenient to denote by $E_{jk}$ the $4\times 4$ elementary matrix
	whose entries are all $0$, except for the $(j,k)$-entry, which is $1$, that is,
	\begin{equation}\label{def:Eij}
	E_{jk}=\left( \delta_{lj}\delta_{km} \right)_{l,m=1}^4,
	\end{equation}
	where $\delta_{jk}$ is the Kronecker delta.

\item Given a function $f(z)$ and an oriented curve, $f_+(z)$ and $f_-(z)$ denote the non-tangential boundary values of $f(z)$ taken from the left ($+$) and right ($-$) sides of the curve, respectively, relative to its orientation. 

\item We denote by $D_{z_0}$ the open disc centred at $z_0$ with radius $\delta > 0$, i.e.,
  \begin{equation}\label{def:dz0r}
   D_{z_0} := \{ z\in \mathbb{C} \mid |z-z_0|<\delta \},
   \end{equation}
  and by $\partial D_{z_0}$ its boundary. The orientation of $\partial D_{z_0}$ is usually taken in a clockwise manner if there are no other specifications.

\item From time to time, we will encounter some functions that depend on the real parameters $\vr_1, \vr_2, \vs_1, \vs_2$ and $\tau$. If $f(\cdot; \vr_1, \vr_2, \vs_1, \vs_2, \tau)$ is such a function, we set
    \begin{align}
\widetilde f(\cdot; \vr_1, \vr_2, \vs_1, \vs_2, \tau) & = f(\cdot; \vr_2, \vr_1, \vs_2, \vs_1, \tau), \label{def:tildeX}
\\
\dot f(\cdot; \vr_1, \vr_2, \vs_1, \vs_2, \tau) &= f(\cdot; \vr_1, \vr_2, \vs_1, \vs_2, -\tau). \label{def:dotX}
\end{align}
It is seen that $f=\widetilde f$ if $\vr_1=\vr_2$ and $\vs_1=\vs_2$, and $f=\dot f$ if $\tau=0$.

\item 
To avoid confusion, we emphasize that the upright symbols $\vr_1,\vr_2,\vs_1,\vs_2$ denote the fixed parameters in the tacnode kernel, and should be distinguished from italic variables such as $s_j$, $\mathfrak{s}_j$ and $r$ that appear elsewhere in the paper.
\end{itemize}

\section{Statement of main results}
\subsection{An RH characterization of the tacnode kernel}\label{sec:tacnode-RH}
We now recall a formula from \cites{DKZ11, DG13} which expresses $K^{\rm tac}$ in terms of the solution $M$ to a $4\times 4$ RH problem.
\begin{paragraph}{RH problem for $M$}
	\begin{itemize}
\item[\rm (a)]  $M(z)=M(z; \vr_1, \vr_2, \vs_1, \vs_2,\tau)$ is analytic for $ z \in \C \setminus \Gamma_M $, where the parameters $\vr_1,\vr_2,\vs_1,\vs_2,\tau$ are real with $\vr_i>0$, $i=1,2$, and
    \begin{equation}
    \Gamma_M:=\cup_{k=0}^5\Gamma_k \cup \{0\}
    \end{equation}
    with
    \begin{equation}\label{phi}
    \begin{aligned}
     &\Gamma_0= (0,+\infty),&& \Gamma_1=e^{\varphi \ii}(0,+\infty), &&\Gamma_2=e^{-\varphi \ii}(-\infty,0),\\
     &\Gamma_3= (-\infty,0),&& \Gamma_4=e^{\varphi \ii}(-\infty,0), &&\Gamma_5=e^{-\varphi \ii}(0,+\infty), \quad 0<\varphi<\frac{\pi}{3};
    \end{aligned}
    \end{equation}
    see Figure \ref{fig:tacnode} for an illustration of the contour $\Gamma_M$.
\item[\rm (b)]  For $z\in\Gamma_k$, $k=0,1,\ldots,5$, the boundary values $M_+(z)$ and $M_-(z)$ exist and satisfy the jump relation
\begin{equation}\label{jumps:M}
M_{+}(z) = M_{-}(z)J_k(z),\qquad k=0,\ldots,5,
\end{equation}
where the jump matrix $J_k(z)$ for each ray $\Gamma_k$ is shown in Figure \ref{fig:tacnode}.
\item[\rm (c)] As $z \to \infty$ with $z \in \C \setminus \Gamma_M$, we have
\begin{align}\label{eq:asy:M}
M(z)&=\left( I+\frac{M_1}{z}+ \Boh(z^{-2}) \right) \diag \left((-z)^{-\frac14},z^{-\frac14},(-z)^{\frac14},z^{\frac14} \right)
\nonumber  \\
& \quad \times A \diag \left(e^{-\theta_1(z)+\tau z}, e^{-\theta_2(z)- \tau z}, e^{\theta_1(z)+\tau z},e^{\theta_2(z)- \tau z} \right),
\end{align}
where the matrix $M_1=M_1(\vr_1, \vr_2, \vs_1, \vs_2, \tau)$ is independent of $z$ but depends on the parameters $\vr_1, \vr_2, \vs_1, \vs_2, \tau$, and
\begin{align} \label{def:A}
&A:=\frac{1}{\sqrt 2} \begin{pmatrix} 1 & 0 & -\ii & 0 \\ 0 & 1& 0& \ii \\
-\ii & 0& 1& 0
\\
0 & \ii & 0 & 1 \end{pmatrix},
\\
&\theta_1(z)= \frac23 \vr_1(-z)^{\frac 32} +2 \vs_1 (-z)^{\frac 12}, \qquad z\in \mathbb{C}\setminus [0,\infty), \label{def:theta1}
\\
&\theta_2(z)=\frac23\vr_2z^{\frac 32} +2 \vs_2 z^{\frac 12}, \qquad z\in \mathbb{C}\setminus (-\infty,0]. \label{def:theta2}
\end{align}
\item[\rm (d)] $M(z)$ is bounded near $z=0$.
\end{itemize}
\end{paragraph}

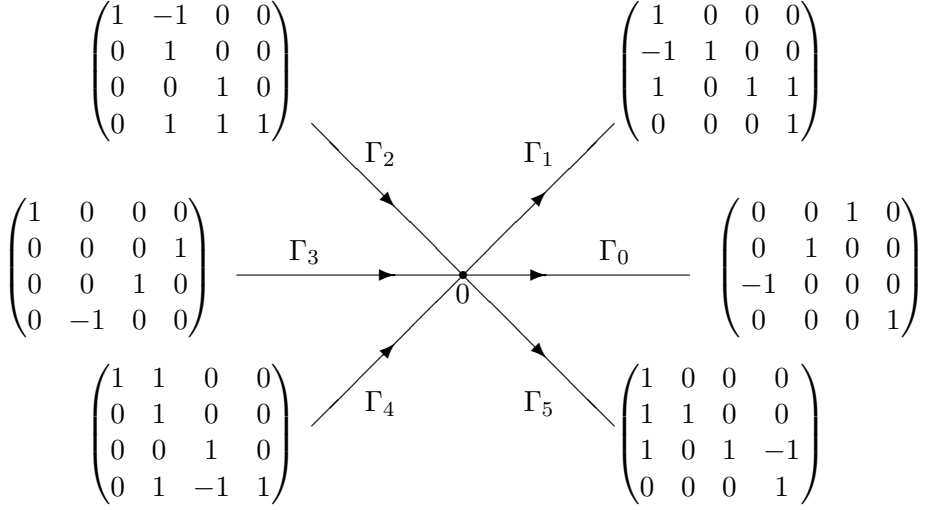
\begin{figure}[t]
\begin{center}
   \setlength{\unitlength}{1truemm}
   \begin{picture}(100,70)(-5,2)
       \put(40,40){\line(-1,-1){20}}
       \put(40,40){\line(-1,1){20}}
       \put(40,40){\line(-1,0){30}}
       \put(40,40){\line(1,0){30}}
       \put(40,40){\line(1,1){20}}
       \put(40,40){\line(1,-1){20}}

       \put(30,50){\thicklines\vector(1,-1){1}}
       \put(30,40){\thicklines\vector(1,0){1}}
       \put(30,30){\thicklines\vector(1,1){1}}
       \put(50,50){\thicklines\vector(1,1){1}}
       \put(50,40){\thicklines\vector(1,0){1}}
       \put(50,30){\thicklines\vector(1,-1){1}}

       \put(39,36.3){$0$}

       \put(27,22){$\Gamma_4$}
       \put(-10,18) {$\begin{pmatrix} 1&1&0&0\\0&1&0&0\\0&0&1&0\\0&1&-1&1 \end{pmatrix}$}

       \put(27,55){$\Gamma_2$}
       \put(-10,66){$\begin{pmatrix} 1&-1&0&0\\0&1&0&0\\0&0&1&0\\0&1&1&1 \end{pmatrix}$}

       \put(17,42){$\Gamma_3$}
       \put(-21,40){$\begin{pmatrix}1&0&0&0\\0&0&0&1\\0&0&1&0\\0&-1&0&0 \end{pmatrix}$}

       \put(48,22){$\Gamma_5$}
       \put(60,18){$\begin{pmatrix} 1&0&0&0\\1&1&0&0\\1&0&1&-1\\0&0&0&1 \end{pmatrix}$}

       \put(48,55){$\Gamma_1$}
       \put(60,66) {$\begin{pmatrix} 1&0&0&0 \\ -1 &1&0&0\\1&0&1&1\\0&0&0&1 \end{pmatrix}$}

       \put(58,42){$\Gamma_0$}
       \put(72,40){ $\begin{pmatrix}
0&0&1&0\\0&1&0&0\\-1&0&0&0\\0&0&0&1 \end{pmatrix}$ }

       \put(40,40){\thicklines\circle*{1}}

 \end{picture}

   \caption{The jump contours $\Gamma_k$ and the corresponding jump matrices $J_{k}$, $k=0,\ldots,5$, in the RH problem for $M$.}
   \label{fig:tacnode}
\end{center}
\end{figure}

The original tacnode RH problem is formulated on a contour consisting of ten rays emanating from the origin. Here we work with an equivalent formulation on a contour with six rays, obtained by combining the two jumps in each open quadrant; see \cites{Kuij,YZ2024}. 
Existence and uniqueness for the tacnode RH problem were proved in \cite{DKZ11} for $\tau=0$, in \cite{DG13} for the symmetric case $\vr_1=\vr_2=1$, $\vs_1=\vs_2$ with general $\tau$, and in \cite{Del} for the non-symmetric case.

The RH problem for $M$ is connected to the Hastings-McLeod solution of the Painlev\'{e} II equation \eqref{eq:PII} via the term $M_1$ in \eqref{eq:asy:M}. More precisely, the Hastings-McLeod solution and its associated Hamiltonian appear in the top right $2 \times 2$ block of $M_1$; see \cites{Del,DKZ11,DG13}. 
It is also noted that $M$ satisfies the following symmetric relations (cf. \cites{Del,DKZ11}):
\begin{align} \label{eq:symmM1}
&  M (-z;\vr_1,\vr_2,\vs_1,\vs_2,\tau) =
  \begin{pmatrix}
  \sigma_1 & 0
  \\
  0 & -\sigma_1
  \end{pmatrix}\widetilde  M(z)
  \begin{pmatrix}
  \sigma_1 & 0
  \\
  0 & -\sigma_1
  \end{pmatrix},
\\
&  M (z;\vr_1,\vr_2,\vs_1,\vs_2,\tau)^{- \rm T} =
  \begin{pmatrix}
  0 & -I_2
  \\
  I_2 & 0
  \end{pmatrix} \dot M(z)
  \begin{pmatrix}
  0 & I_2
  \\
  -I_2 & 0
  \end{pmatrix}, \label{eq:symmM2}
\end{align}
where $\sigma_1$ is defined in \eqref{def:Pauli}, $\widetilde M$ and $\dot M$ are defined through \eqref{def:tildeX} and \eqref{def:dotX}.
It then follows from \eqref{eq:asy:M} that $M_1$ satisfies the symmetric relations
\begin{align}
& M_1 =
  -\begin{pmatrix}
  \sigma_1 & 0
  \\
  0 & -\sigma_1
  \end{pmatrix}\widetilde M_1
  \begin{pmatrix}
  \sigma_1 & 0
  \\
  0 & -\sigma_1
  \end{pmatrix},\label{eq:symmM11}
\\
& ( M_1)^{\rm T} =
  -\begin{pmatrix}
  0 & -I_2
  \\
  I_2 & 0
  \end{pmatrix}\dot M_1
  \begin{pmatrix}
  0 & I_2
  \\
  -I_2 & 0
  \end{pmatrix}, \label{eq:symmM12}
\end{align}
which can be rewritten in terms of the entries of $M_1$ as
\begin{align}
(M_1)_{11}&= -(\widetilde M_1)_{22}=-(\dot M_1)_{33}=(\dot{\widetilde{M}}_1)_{44},\label{M11}\\
(M_1)_{13}&= (\widetilde M_1)_{24}=(\dot M_1)_{13},\label{M13}\\
(M_1)_{23}&= (\widetilde M_1)_{14}=(\dot M_1)_{14}.\label{M14}
\end{align}

Let $\widehat M $ denote the analytic continuation of the restriction of $M$ in the sector bounded by the rays $\Gamma_1$ and $\Gamma_2$ to the whole complex plane. 
Then, by \cite[Definition 2.6]{DKZ11}, the tacnode kernel $K^{\tac}(x,y):=K^{\tac}(x,y;\vr_1,\vr_2,\vs_1,\vs_2,\tau)$ is given by
\begin{equation} \label{def:tacnode kernel}
    K^{\tac}(x,y)
     = \frac{1}{2\pi \ii (x-y)} \begin{pmatrix} 0 & 0 & 1 & 1 \end{pmatrix}
    \widehat M(y)^{-1}
    \widehat M(x) \begin{pmatrix} 1 \\ 1 \\ 0 \\ 0
    \end{pmatrix}.
\end{equation}

\subsection{An integral representation of the generating function $F$}\label{sec: coupled system}
Our first result gives an integral representation for $F$ in terms of the Hamiltonian associated with a system of $8m+4$ coupled differential equations, involving the unknown 
functions
\begin{equation}\label{def: unknown functions}
	\{p_{j,k}(r), \ q_{j,k}(r), \ p_5(r), \ p_6(r), \ q_5(r), \ q_6(r) \}, \quad j=1, \ldots, m, \ k=1,2,3,4,
\end{equation}
all of which depend on the real variable $r$ and the parameters $\vec{u}, \vec{x}, \vr_1, \vr_2, \vs_1, \vs_2, \tau$.
The differential equations are given as follows:
{\small
\begin{align}\label{eq:coupled system}
\left\{ 
	\begin{aligned} 
		q_{j,1}'(r)&=x_j\Big((p_5(r)+\tau)q_{j,1}(r)-\ii\vr_2 q_6(r)q_{j,2}(r)+\ii\vr_1 q_{j,3}(r)\Big)+\frac1r\sum_{k=1}^m S_{jk}(r)q_{k,1}(r)+\frac1r\sum_{k=1}^m \Upsilon_{jk}(r)\widetilde q_{k,2}(r),\\[1mm]
		q_{j,2}'(r)&=x_j\Big(\ii\vr_1\widetilde q_6(r)q_{j,1}(r)-\big(\widetilde p_5(r)+\tau\big)q_{j,2}(r)+\ii\vr_2 q_{j,4}(r)\Big)+\frac1r\sum_{k=1}^m S_{jk}(r)q_{k,2}(r)+\frac1r\sum_{k=1}^m \Upsilon_{jk}(r)\widetilde q_{k,1}(r),\\[1mm]
		q_{j,3}'(r)&=\ii\vr_1 r x_j^2 q_{j,1}(r)+x_j\Big((\ii\vr_1 q_5(r)-\ii\vs_1)q_{j,1}(r)-\widetilde p_6(r)q_{j,2}(r)+(-p_5(r)+\tau)q_{j,3}(r)-\ii\vr_1 q_6(r)q_{j,4}(r)\Big)\\ &\hspace*{2.5em}+\frac1r\sum_{k=1}^m S_{jk}(r)q_{k,3}(r)-\frac1r\sum_{k=1}^m \Upsilon_{jk}(r)\widetilde q_{k,4}(r),\\[1mm]
		q_{j,4}'(r)&=-\ii\vr_2 r x_j^2 q_{j,2}(r)+x_j\Big(-p_6(r)q_{j,1}(r)+(\ii\vr_2\widetilde q_5(r)-\ii\vs_2)q_{j,2}(r)+\ii\vr_2\widetilde q_6(r)q_{j,3}(r)+(\widetilde p_5(r)-\tau)q_{j,4}(r)\Big)\\ &\hspace*{2.5em}+\frac1r\sum_{k=1}^m S_{jk}(r)q_{k,4}(r)-\frac1r\sum_{k=1}^m \Upsilon_{jk}(r)\widetilde q_{k,3}(r),\\[1mm]
		p_{j,1}'(r)&=-\ii\vr_1 r x_j^2 p_{j,3}(r)-x_j\Big((p_5(r)+\tau)p_{j,1}(r)+\ii\vr_1\widetilde q_6(r)p_{j,2}(r)+ (\ii\vr_1 q_5(r)-\ii\vs_1)p_{j,3}(r)-p_6(r)p_{j,4}(r)\Big)\\ &\hspace*{2.5em}-\frac1r\sum_{k=1}^m S_{kj}(r)p_{k,1}(r)-\frac1r\sum_{k=1}^m \widetilde{\Upsilon}_{kj}(r)\widetilde p_{k,2}(r),\\[1mm]
		p_{j,2}'(r)&=\ii\vr_2 r x_j^2 p_{j,4}(r)+x_j\Big(\ii\vr_2 q_6(r)p_{j,1}(r)+(\widetilde p_5(r)+\tau)p_{j,2}(r)+\widetilde p_6(r)p_{j,3}(r)-(\ii\vr_2\widetilde q_5(r)-\ii\vs_2)p_{j,4}(r)\Big)\\ &\hspace*{2.5em}-\frac1r\sum_{k=1}^m S_{kj}(r)p_{k,2}(r)-\frac1r\sum_{k=1}^m \widetilde{\Upsilon}_{kj}(r)\widetilde p_{k,1}(r),\\[1mm]
		p_{j,3}'(r)&=-x_j\Big(\ii\vr_1 p_{j,1}(r)+(-p_5(r)+\tau)p_{j,3}(r)+\ii\vr_2\widetilde q_6(r)p_{j,4}(r)\Big)-\frac1r\sum_{k=1}^m S_{kj}(r)p_{k,3}(r)+\frac1r\sum_{k=1}^m \widetilde{\Upsilon}_{kj}(r)\widetilde p_{k,4}(r),\\[1mm]
		p_{j,4}'(r)&=-x_j\Big(\ii\vr_2 p_{j,2}(r)-\ii\vr_1 q_6(r)p_{j,3}(r)+(\widetilde p_5(r)-\tau)p_{j,4}(r)\Big)-\frac1r\sum_{k=1}^m S_{kj}(r)p_{k,4}(r)+\frac1r\sum_{k=1}^m \widetilde{\Upsilon}_{kj}(r)\widetilde p_{k,3}(r),\\[1mm]
		p_5'(r) &= -\ii \vr_1\sum_{j=1}^m x_j\Big(q_{j,1}(r)p_{j,3}(r)+\widetilde q_{j,2}(r)\widetilde p_{j,4}(r)\Big),\\[1mm]  
		q_5'(r) &= \sum_{j=1}^m x_j\Big( q_{j,1}(r)p_{j,1}(r)-q_{j,3}(r)p_{j,3}(r) -\widetilde q_{j,2}(r)\widetilde p_{j,2}(r) +\widetilde q_{j,4}(r)\widetilde p_{j,4}(r) \Big),\\[1mm] 
		p_6'(r) &= \ii\sum_{j=1}^m x_j\Big( \vr_1 q_{j,4}(r)p_{j,3}(r) +\vr_2 q_{j,2}(r)p_{j,1}(r) -\vr_1 \widetilde q_{j,3}(r)\widetilde p_{j,4}(r) -\vr_2 \widetilde q_{j,1}(r)\widetilde p_{j,2}(r) \Big),\\[1mm]
		q_6'(r) &= -\sum_{j=1}^m x_j\Big( q_{j,1}(r)p_{j,4}(r)+\widetilde q_{j,2}(r)\widetilde p_{j,3}(r) \Big),
	\end{aligned} 
\right.
\end{align}
}
where $(\cdot)'$ denotes the derivative with respect to $r$, and for $j,k=1,\ldots,m$,
\begin{align}
& S_{jk}(r):=\sum_{\ell=1}^{4}q_{j,\ell}(r) p_{k,\ell}(r), \label{def:S_jk} \\
& \Upsilon_{jk}(r):=q_{j,1}(r)\widetilde p_{k,2}(r)+q_{j,2}(r)\widetilde p_{k,1}(r)-q_{j,3}(r)\widetilde p_{k,4}(r)-q_{j,4}(r)\widetilde p_{k,3}(r). \label{def:Upsilon_jk}
\end{align}
We remind the readers of that the functions with ``$\widetilde{(\cdot)}$'' are defined through \eqref{def:tildeX} by swapping the parameters $\vr_1\leftrightarrow \vr_2$ and $\vs_1\leftrightarrow \vs_2$.
Furthermore, the functions \eqref{def: unknown functions} satisfy the extra $m$ relations 
\begin{equation}\label{eq:extra-condition-1}
	\sum_{\ell=1}^{4} q_{j,\ell}(r)p_{j,\ell}(r)=0, \qquad j=1,2,\ldots,m,
\end{equation}
which leads to $S_{jj}(r)=0$, hence the sums $\sum_{k=1}^{m}S_{jk}(r)$ (or $\sum_{k=1}^{m}S_{kj}(r)$) 
could be taken over all $k=1,\ldots,m$ without excluding the term $k=j$.

By introducing the matrix-valued functions
\begin{equation}\label{def:A0}
A_0(r) = \begin{pmatrix}
p_5(r)+\tau & -\ii \vr_2 q_6(r) & \ii \vr_1 & 0\\
\ii \vr_1 \widetilde q_6(r) & -\widetilde p_5(r)-\tau & 0 & \ii \vr_2\\
\ii \vr_1 q_5(r)-\ii \vs_1 & -\widetilde p_6(r) & -p_5(r)+\tau & -\ii \vr_1 q_6(r)\\
-p_6(r) & \ii \vr_2 \widetilde q_5(r)-\ii \vs_2 & \ii \vr_2 \widetilde q_6(r) & \widetilde p_5(r)-\tau
\end{pmatrix},
\end{equation}
\begin{equation}\label{def:Aj(r)}
A_j(r) = \begin{pmatrix}
q_{j,1}(r)\\q_{j,2}(r)\\q_{j,3}(r)\\q_{j,4}(r)
\end{pmatrix}\begin{pmatrix}
p_{j,1}(r) & p_{j,2}(r)&p_{j,3}(r)&p_{j,4}(r)
\end{pmatrix}, \quad j=1, \ldots, m,
\end{equation}
and
\begin{equation}\label{def:A-j(r)}
	A_{-j}(r) = \begin{pmatrix}
\widetilde q_{j,2}(r)\\\widetilde q_{j,1}(r)\\ -\widetilde q_{j,4}(r)\\ -\widetilde q_{j,3}(r)
\end{pmatrix}\begin{pmatrix}
\widetilde p_{j,2}(r) & \widetilde p_{j,1}(r)& -\widetilde p_{j,4}(r)& -\widetilde p_{j,3}(r)
\end{pmatrix}, \quad j=1, \ldots, m,
\end{equation}
one can check
\begin{align}\label{eq:Hamiltonian}
& H(r)=H\left(r; \{p_{j,k}(r), q_{j,k}(r)\}_{j=1,\ldots,m}^{k=1,\ldots,4}, p_5(r), q_5(r), p_6(r), q_6(r)\right)\nonumber\\
&:=\sum_{j=1}^{m}\Bigg\{
r x_j^2\Big(
\ii\vr_1 q_{j,1}(r)p_{j,3}(r)-\ii\vr_2 q_{j,2}(r)p_{j,4}(r)
+\ii\vr_2 \widetilde q_{j,1}(r)\widetilde p_{j,3}(r)
-\ii\vr_1 \widetilde q_{j,2}(r)\widetilde p_{j,4}(r)
\Big)\nonumber\\
&\quad+x_j\Big(
(p_5(r)+\tau)q_{j,1}(r)p_{j,1}(r)
-\ii\vr_2 q_6(r)q_{j,2}(r)p_{j,1}(r)
+\ii\vr_1 q_{j,3}(r)p_{j,1}(r)\nonumber\\
&\qquad\qquad
+\ii\vr_1 \widetilde q_6(r)q_{j,1}(r)p_{j,2}(r)
-(\widetilde p_5(r)+\tau)q_{j,2}(r)p_{j,2}(r)
+\ii\vr_2 q_{j,4}(r)p_{j,2}(r)\nonumber\\
&\qquad\qquad
+(\ii\vr_1 q_5(r)-\ii\vs_1)q_{j,1}(r)p_{j,3}(r)
-\widetilde p_6(r)q_{j,2}(r)p_{j,3}(r)
+(-p_5(r)+\tau)q_{j,3}(r)p_{j,3}(r)\nonumber\\
&\qquad\qquad
-\ii\vr_1 q_6(r)q_{j,4}(r)p_{j,3}(r)
-p_6(r)q_{j,1}(r)p_{j,4}(r)
+(\ii\vr_2 \widetilde q_5(r)-\ii\vs_2)q_{j,2}(r)p_{j,4}(r)\nonumber\\
&\qquad\qquad
+\ii\vr_2 \widetilde q_6(r)q_{j,3}(r)p_{j,4}(r)
+(\widetilde p_5(r)-\tau)q_{j,4}(r)p_{j,4}(r)\nonumber\\
&\qquad\qquad
+(\widetilde p_5(r)+\tau)\widetilde q_{j,1}(r)\widetilde p_{j,1}(r)
-(p_5(r)+\tau)\widetilde q_{j,2}(r)\widetilde p_{j,2}(r)\nonumber\\
&\qquad\qquad
-(\widetilde p_5(r)-\tau)\widetilde q_{j,3}(r)\widetilde p_{j,3}(r)
+(p_5(r)-\tau)\widetilde q_{j,4}(r)\widetilde p_{j,4}(r)\nonumber\\
&\qquad\qquad
+\ii\vr_2 q_6(r)\widetilde q_{j,1}(r)\widetilde p_{j,2}(r)
-\ii\vr_1 \widetilde q_6(r)\widetilde q_{j,2}(r)\widetilde p_{j,1}(r)\nonumber\\
&\qquad\qquad
+\ii\vr_2 \widetilde q_{j,3}(r)\widetilde p_{j,1}(r)
+\ii\vr_1 \widetilde q_{j,4}(r)\widetilde p_{j,2}(r)\nonumber\\
&\qquad\qquad
+(\ii\vr_2 \widetilde q_5(r)-\ii\vs_2)\widetilde q_{j,1}(r)\widetilde p_{j,3}(r)
+(\ii\vr_1 q_5(r)-\ii\vs_1)\widetilde q_{j,2}(r)\widetilde p_{j,4}(r)\nonumber\\
&\qquad\qquad
-p_6(r)\widetilde q_{j,2}(r)\widetilde p_{j,3}(r)
-\widetilde p_6(r)\widetilde q_{j,1}(r)\widetilde p_{j,4}(r)\nonumber\\
&\qquad\qquad
-\ii\vr_2 \widetilde q_6(r)\widetilde q_{j,4}(r)\widetilde p_{j,3}(r)
+\ii\vr_1 q_6(r)\widetilde q_{j,3}(r)\widetilde p_{j,4}(r)
\Big)\nonumber\\
&\quad+\frac{x_j}{r}\sum_{\substack{k=1\\k\neq j}}^{m}\left(
\frac{S_{jk}(r)S_{kj}(r)+\widetilde S_{jk}(r)\widetilde S_{kj}(r)}{x_j-x_k}
+\frac{\Upsilon_{jk}(r)\widetilde\Upsilon_{kj}(r)+\Upsilon_{kj}(r)\widetilde\Upsilon_{jk}(r)}{x_j+x_k}
\right)\nonumber\\
&\quad+\frac{1}{r}\Upsilon_{jj}(r)\widetilde\Upsilon_{jj}(r)
\Bigg\}
\end{align}
is a Hamiltonian associated with the system of coupled differential equations \eqref{eq:coupled system} under the condition \eqref{eq:extra-condition-1}.
That is, the system \eqref{eq:coupled system} satisfies the canonical Hamiltonian equations
\begin{equation}\label{eq:Hamiltonian-derivative-1}
	q_{j}'(r)=\frac{\partial H}{\partial p_{j}}, \quad p_{j}'(r)=-\frac{\partial H}{\partial q_{j}}, \quad j=5,6,
\end{equation}
and
\begin{equation}\label{eq:Hamiltonian-derivative-2}
q_{j,k}'(r)=\frac{\partial H}{\partial p_{j,k}}, \quad p_{j,k}'(r)=-\frac{\partial H}{\partial q_{j,k}}, \quad j=1,\ldots,m, \ k=1,2,3,4.
\end{equation}
Since the Hamiltonian $H(r)$ is rather involved, we refer to Appendix \ref{app:verification-Hamiltonian} for the proof of the Hamiltonian formulation in \eqref{eq:Hamiltonian-derivative-1}--\eqref{eq:Hamiltonian-derivative-2}. 
The large and small $r$ asymptotics of the solution to the system \eqref{eq:coupled system} can also be obtained, 
which is given in the following proposition.
\begin{proposition}\label{prop:coupled system-asy}
	Let the parameters $\vr_1,\vr_2,\vs_1,\vs_2,\tau$ be real with $\vr_i>0$, $i=1,2$, $\vec{u}$, $\vec{x}$
	given in \eqref{def: vecu and vecx} and $r>0$. There exist at least one family of solutions to the system of 
	equations \eqref{eq:coupled system} and \eqref{eq:extra-condition-1} satisfying the asymptotics below.
	As $r\to+\infty$, we have
	\begin{align}
	q_{j,1}(r) &=\sqrt{2}e^{\tau r x_j}x_{j}^{-\frac14}\mathcal{A}_j\cos\left(\vartheta_j(r)-\frac{\pi}{4}\right)r^{-\frac14}\left(1+\Boh\left(r^{-\frac12}\right)\right), \label{qj1-large r}\\
	q_{j,2}(r) &=\sqrt{2}e^{\tau r x_j}\frac{(M_1)_{23}}{x_j}\mathcal{A}_j\nonumber\\ & \times\left((\mathcal{T}_+-\mathcal{T}_-)x_j^{-\frac14}\cos\left(\vartheta_j(r)-\frac{\pi}{4}\right)+\ii x_j^{\frac14}\sin\left(\vartheta_j(r)-\frac{\pi}{4}\right)\right)r^{-\frac34}\left(1+\Boh\left(r^{-\frac12}\right)\right),\label{qj2-large r}\\
	q_{j,3}(r) &=\sqrt{2}e^{\tau r x_j}\mathcal{A}_j\left(\mathcal{T}_{+}x_{j}^{-\frac14}\cos\left(\vartheta_j(r)-\frac{\pi}{4}\right)+\ii x_j^{\frac14}\sin\left(\vartheta_j(r)-\frac{\pi}{4}\right)\right)r^{\frac14}\left(1+\Boh(r^{-\frac12})\right), \label{qj3-large r}\\
	q_{j,4}(r) &=-\sqrt{2}e^{\tau r x_j}\frac{(M_1)_{23}}{x_j} \mathcal{T}_{-}\mathcal{A}_j\nonumber\\ & \times\left((\mathcal T_+-\mathcal T_-)x_j^{-\frac14}\cos\left(\vartheta_j(r)-\frac{\pi}{4}\right)+\ii x_j^{\frac14}\sin\left(\vartheta_j(r)-\frac{\pi}{4}\right)\right)r^{-\frac14}\left(1+\Boh\left(r^{-\frac12}\right)\right),\label{qj4-large r}\\
	q_5(r) &= -4r\left(\sum_{j=1}^{m}\beta_j x_{j}^{\frac12}\right)^2 + 2r^{\frac{1}{2}}\sum\limits_{j=1}^m \beta_jx_j^{-\frac{1}{2}}\Big((\dot M_1)_{13}+(M_1)_{13}\Big) \nonumber \\
	&\quad - \left(1-4\left(\sum\limits_{j=1}^m \beta_jx_j^{-\frac{1}{2}}\right)^2\right)\Big((\dot M_1)_{11}+(M_1)_{11}\Big) \nonumber \\
	&\quad - 4\left(\sum\limits_{j=1}^m \beta_jx_j^{-\frac{1}{2}}\right)^2\Big((\dot M_1)_{33}+(M_1)_{33} + \big((\dot M_1)_{13}\big)^2 + \big((M_1)_{13}\big)^2 \nonumber \\
	&\qquad \qquad \qquad \qquad \qquad + (\dot M_1)_{14}(\dot M_1)_{23} + (M_1)_{14}(M_1)_{23}\Big) + \Boh(r^{-\frac12}), \label{q5-large r}\\
	q_6(r) &= (M_1)_{14} - 2r^{-\frac{1}{2}}\sum\limits_{j=1}^m \beta_jx_j^{-\frac{1}{2}} \Big((M_1)_{12}+(M_1)_{34} + (M_1)_{13}(M_1)_{14} + (M_1)_{14}(M_1)_{24}\Big) + \Boh(r^{-1}),\label{q6-large r}\\
	p_{j,1}(r) &=\frac{e^{\sum\limits_{\ell=j+1}^{m}u_{\ell}}(1-e^{u_j})}{\sqrt2\pi}e^{-\tau r x_j}\nonumber\\ &\times\mathcal{A}_j\left(x_j^{\frac14}\sin\left(\vartheta_j(r)-\frac{\pi}{4}\right)-\ii\mathcal{T}_+x_j^{-\frac14}\cos\left(\vartheta_j(r)-\frac{\pi}{4}\right)\right)r^{\frac14}\left(1+\Boh\left(r^{-\frac12}\right)\right), \label{pj1-large r}\\
	p_{j,2}(r) &=-\frac{e^{\sum\limits_{\ell=j+1}^{m}u_{\ell}}(1-e^{u_j})}{\sqrt2\pi}e^{-\tau r x_j}\frac{(M_1)_{14}}{x_j}\mathcal{T}_{-}\mathcal{A}_j\nonumber\\ & \times\left(x_j^{\frac14}\sin\left(\vartheta_j(r)-\frac{\pi}{4}\right)-\ii(\mathcal{T}_+-\mathcal{T}_-)x_j^{-\frac14}\cos\left(\vartheta_j(r)-\frac{\pi}{4}\right)\right)r^{-\frac14}\left(1+\Boh\left(r^{-\frac12}\right)\right),\label{pj2-large r}\\
	p_{j,3}(r) &=\frac{\ii e^{\sum\limits_{\ell=j+1}^{m}u_{\ell}}(1-e^{u_j})}{\sqrt2\pi}e^{-\tau r x_j}\mathcal{A}_j x_j^{-\frac14}\cos\left(\vartheta_j(r)-\frac{\pi}{4}\right)r^{-\frac14}\left(1+\Boh\left(r^{-\frac12}\right)\right), \label{pj3-large r}\\
	p_{j,4}(r) &=-\frac{e^{\sum\limits_{\ell=j+1}^{m}u_{\ell}}(1-e^{u_j})}{\sqrt2\pi}e^{-\tau r x_j}
	\frac{(M_1)_{14}}{x_j}\mathcal{A}_j \nonumber\\ & \times\left(x_j^{\frac14}\sin\left(\vartheta_j(r)-\frac{\pi}{4}\right)-\ii(\mathcal{T}_+-\mathcal{T}_-)x_j^{-\frac14}\cos\left(\vartheta_j(r)-\frac{\pi}{4}\right)\right)r^{-\frac34}\left(1+\Boh\left(r^{-\frac12}\right)\right), \label{pj4-large r}\\
	p_5(r) &= -2\ii \vr_1 r^{\frac{1}{2}}\sum_{j=1}^{m}\beta_j x_{j}^{\frac12} + \ii \vr_1 (M_1)_{13} \nonumber \\
	&\quad + 2\ii \vr_1 r^{-\frac{1}{2}}\sum\limits_{j=1}^m \beta_jx_j^{-\frac{1}{2}} \Big((M_1)_{11}-(M_1)_{33} - \left((M_1)_{13}\right)^2 - (M_1)_{14}(M_1)_{23}\Big) + \Boh(r^{-1}),  \label{p5-large r}\\
	p_6(r) &= 2\ii r^{\frac{1}{2}}\sum\limits_{j=1}^m \beta_jx_j^{-\frac{1}{2}}\Big(\vr_1(\dot M_1)_{14}+\vr_2(\widetilde M_1)_{14}\Big) \nonumber \\
	&\quad + \ii\left(1-4\left(\sum\limits_{j=1}^m \beta_jx_j^{-\frac{1}{2}}\right)^2\right)\Big(\vr_1(\dot M_1)_{12} + \vr_2(\widetilde M_1)_{12}\Big) \nonumber \\
	&\quad -4\ii\left(\sum\limits_{j=1}^m \beta_jx_j^{-\frac{1}{2}}\right)^2\Big[\vr_1\Big((\dot M_1)_{34} + (\dot M_1)_{13}(\dot M_1)_{14} + (\dot M_1)_{14}(\dot M_1)_{24}\Big) \nonumber \\
	&\qquad \qquad \qquad \qquad \quad + \vr_2\Big((\widetilde M_1)_{34} + (\widetilde M_1)_{13}(\widetilde M_1)_{14} + (\widetilde M_1)_{14}(\widetilde M_1)_{24}\Big)\Big] + \Boh(r^{-\frac12}), \label{p6-large r}
	\end{align}
	where $M_1$ is defined in \eqref{eq:asy:M}, $\widetilde M_1$ and $\dot M_1$ are defined through \eqref{def:tildeX} and \eqref{def:dotX},  and
	\begin{align}
		&\mathcal{T}_{\pm}:=\frac{1}{\pi\ii}\sum\limits_{k=1}^{m}u_kx_k^{\pm \frac12},\\
		&\beta_j:=\frac{1}{2\pi i}u_j\in i\mathbb{R}, \quad j=1,\ldots,m, \label{def: first def beta_j}\\
		&\mathcal{A}_j:=\exp\left(-\frac12\sum\limits_{\ell=j+1}^{m}u_{\ell}-\frac{u_j}{4}\right) \left|\Gamma\left(1-\frac{u_j}{2\pi\ii}\right)\right|, \label{def:amplitude Aj}\\
		&\vartheta_j(r):=\vartheta_j(r; \vec{x}, \vec{u}, \vr_1,\vs_1)\nonumber\\
		&\hspace*{1em}= \frac{2}{3}\vr_1r^{\frac32}x_j^{\frac32}-2\vs_1r^{\frac12}x_j^{\frac12}+\frac{u_j}{2\pi}\log\left(8r^{\frac12}x_j^{\frac12}(r\vr_1x_j-\vs_1)\right) 
		+\sum_{\substack{k=1\\ k\neq j}}^{m}\frac{u_k}{2\pi}\log\left(\frac{{x}_{j}^{\frac12}+{x}_{k}^{\frac12}}{|{x}_{j}^{\frac12}-{x}_{k}^{\frac12}|}\right)+\arg\Gamma\left(1+\frac{u_j}{2\pi\ii}\right). \label{eq:def-theta-j-general}
	\end{align}
As $r\to 0$,
\begin{align}
	&q_{j,1}(r)=\Boh(1), \quad q_{j,2}(r)=\Boh(1), \quad q_{j,3}(r)=\Boh(1), \quad q_{j,4}(r)=\Boh(1), \label{qj1234-small r}\\
	&q_5(r)=-(\dot M_1)_{11}-(M_1)_{11}+\Boh(r), \label{q5-small r}\\
	&q_6(r)=(M_1)_{14}+\Boh(r), \label{q6-small r}\\
	&p_{j,1}(r)=\Boh(1), \quad p_{j,2}(r)=\Boh(1), \quad p_{j,3}(r)=\Boh(1), \quad p_{j,4}(r)=\Boh(1), \label{pj1234-small r}\\
	&p_5(r)=\ii \vr_1 (M_1)_{13} + \Boh(r), \label{p5-small r}\\
	&p_6(r)=\ii \vr_1 (\dot M_1)_{12} +\ii \vr_2 (\widetilde M_1)_{12}+\Boh(r). \label{p6-small r}
\end{align}
\end{proposition}
\begin{remark}
	For $m=1$, all cross terms disappear, and the above asymptotics reduce to those obtained previously in \cite[Proposition 8.1]{YZ2024}. 
	We also note that it remains unclear whether the asymptotic conditions in Proposition \ref{prop:coupled system-asy} determine 
	a unique solution to the system \eqref{eq:coupled system} and \eqref{eq:extra-condition-1}.
\end{remark}

\begin{remark}\label{remark:typo}
	For $m=1$, the formula for $\vartheta(s)$ given in \cite[Equation (2.29)]{YZ2024} 
	contains a typo. There, the variable $s$ corresponds to our $r$, and $\vartheta(s)$ corresponds to $\vartheta_1(r)$ with $x_1=1$. The second term $-2\vs_1 s$ should be replaced by $-2\vs_1 s^{\frac12}$, since it comes from $\theta_1$ defined in \eqref{def:theta1}.
\end{remark}

We now turn to the relation between the system above and the generating function. The next theorem shows that $F$ can be recovered from the associated Hamiltonian, and hence admits an integral representation in terms of a distinguished family of solutions to \eqref{eq:coupled system} and \eqref{eq:extra-condition-1}. Combined with the asymptotics in Proposition \ref{prop:coupled system-asy}, this representation will serve as the starting point for the derivation of the large gap asymptotics of $F$.
\begin{theorem}\label{thm: Hamiltonian and its asymptotics}
Let the parameters $\vr_1,\vr_2,\vs_1,\vs_2,\tau$ be real with $\vr_i>0$, $i=1,2$, $\vec{u}$, $\vec{x}$
given in \eqref{def: vecu and vecx} and $r>0$. With the generating function $F$ defined in \eqref{eq:generating_function}, 
we have the following integral representation for $F$:
	\begin{equation}\label{eq:integral representation via Hamiltonian}
		F(r\vec{x}, \vec{u}; \vr_1, \vr_2, \vs_1, \vs_2, \tau)=\exp\left(\int_0^r H\left(t; \vec{x}, \vec{u}, \vr_1, \vr_2, \vs_1, \vs_2, \tau\right) \dif t\right),
	\end{equation}
	with $H$ being the Hamiltonian defined in \eqref{eq:Hamiltonian}, and where {\rm (}$\{p_{j,k}, q_{j,k}\}_{j=1,\ldots,m}^{k=1,2,3,4}$, $p_5$, $p_6$, $q_5$, $q_6${\rm )} is 
	a family of solutions to the system of equations \eqref{eq:coupled system} and \eqref{eq:extra-condition-1} satisfying the asymptotics given in Proposition \ref{prop:coupled system-asy}.
	Furthermore, 
	\begin{equation}\label{eq:H-small-r}
		H(r)=\Boh(1), \qquad {\rm as} \ r\to 0^+,
	\end{equation}
	and as $r\to +\infty$,
	\begin{multline}\label{eq:H-large-r}
		H(r)=\sum_{j=1}^{m}\Bigg[\frac{u_j}{\pi}(\vr_1+\vr_2)x_j^{\frac32}r^{\frac12}
	-\frac{u_j}{\pi}(\vs_1+\vs_2)x_j^{\frac12}r^{-\frac12} \\
	\qquad\qquad\qquad\qquad -\Bigg(-\frac{3}{4\pi}u_j^2+\frac{u_j}{4\pi}\cos\left(2\vartheta_j(r)\right)+\frac{u_j}{4\pi }\cos\left(2\widetilde\vartheta_j(r)\right)\Bigg)r^{-1}\Bigg]
	+\Boh\left(r^{-\frac32}\right),
	\end{multline}
	where $\vartheta_j(r)$ is given in \eqref{eq:def-theta-j-general} and $\widetilde\vartheta_j(r):=\vartheta_j(r;\vec x, \vec u, \vr_2, \vs_2)$. 
\end{theorem}
It is also worth to mention that the local behavior of $H$ near the origin \eqref{eq:H-small-r} ensures that the well-definiteness of the integral representation \eqref{eq:integral representation via Hamiltonian} for $F$.

\subsection{Large gap asymptotics and further applications}
Since $F(r\vec{x}, \vec{u})$ is a highly transcendental function, it is of great interest to investigate its asymptotic behaviors for sufficiently small and large $r$.
As $F(r\vec{x},\vec{u})$ is a Fredholm determinant on an interval of size proportional to $r$, the small $r$ asymptotics of $F$ can be obtained rather directly.
However, obtaining its large $r$ asymptotics turns out to be more complicated. Although substituting \eqref{eq:H-large-r} into \eqref{eq:integral representation via Hamiltonian} does yield the first few terms of the expansion, 
the constant term remains undetermined in this way. A more delicate analysis is therefore required to capture the constant term, and the resulting asymptotics are stated in the following theorem.
\begin{theorem}\label{thm: large_gap}
Let the parameters $\vr_1,\vr_2,\vs_1,\vs_2,\tau$ be real with $\vr_i>0$, $i=1,2$, $\vec{u}$, $\vec{x}$
given in \eqref{def: vecu and vecx}. With the generating function $F$ defined in \eqref{eq:generating_function}, 
we have, as $r\to+\infty$,
\begin{multline}\label{eq: large_gap_asymptotics_generating_function}
	F(r\vec x, \vec u)=\exp\left(\sum_{j=1}^{m} u_j\mu(rx_j)+\sum_{j=1}^{m}\frac{u_j^2}{2}{\sigma(rx_j)}^{2}+\sum_{1\leq j<k\leq m} u_ju_k\Sigma(x_j,x_k)\right.\\ 
	\left.+2\sum_{j=1}^{m}\log\left(G\left(1-\frac{u_j}{2\pi\ii}\right)G\left(1+\frac{u_j}{2\pi\ii}\right)\right)+\Boh(r^{-\frac12})\right),
\end{multline}
where $G(\cdot)$ denotes the Barnes' $G$-function, and 
\begin{align}
	& \mu(y):=\mu(y;\vr_1,\vr_2,\vs_1,\vs_2)=\frac{2(\vr_1+\vr_2)}{3\pi} y^{\frac32}-\frac{2(\vs_1+\vs_2)}{\pi} y^{\frac12}, \label{def:mu}\\
	& {\sigma(y)}^{2}:={\sigma(y;\vr_1,\vr_2)}^{2}=\frac{3}{2\pi^2}\log y+\frac{1}{2\pi^2}\log (64\vr_1\vr_2), \label{def:sigma^2}\\
	& \Sigma(x_j,x_k)=\frac{1}{\pi^2}\log\left(\dfrac{x_j^{\frac12}+x_k^{\frac12}}{\left|x_j^{\frac12}-x_k^{\frac12}\right|}\right). \label{def:SigmaCovariance}
\end{align}
Furthermore, the asymptotic formula \eqref{eq: large_gap_asymptotics_generating_function} holds uniformly for $\vr_1$, $\vr_2$ in compact subsets 
of $\mathbb{R}^{+}$, for $\vs_1$, $\vs_2$, $\tau$ in compact subsets of $\mathbb{R}$, for $\vec{u}$ in compact subsets of $\mathbb{R}^m$ and for 
$\vec{x}$ in compact subsets of $\mathbb{R}^{+,m}_{<}$. 
Also, the asymptotic formula \eqref{eq: large_gap_asymptotics_generating_function} could be differentiated any number of times with respect to $u_1,\ldots,u_m$ 
at the expense of increasing the error terms. That is, denoting error term by $\Delta_{F}:=\log F(r\vec{x},\vec{u})-\log \hat{F}(r\vec{x},\vec{u})$ where $\hat{F}(r\vec{x},\vec{u})$ is 
the right-hand side of \eqref{eq: large_gap_asymptotics_generating_function}
without error term, then for any $k_1,\ldots,k_m\in\mathbb{N}^{+}$, we have
\begin{equation}\label{eq:derivative of error term}
	\partial_{u_1}^{k_1}\cdots \partial_{u_m}^{k_m} \Delta_{F} = \Boh\left(r^{-\frac12}(\log r)^{k_1+\cdots+k_m}\right), \quad {\rm as} \ r\to +\infty.
\end{equation}
\end{theorem}
The form of \eqref{eq: large_gap_asymptotics_generating_function} is expected to be universal, see also \cites{Charlier21a,Charlier21c,CC20,CM23} for other classical point processes where similar large gap asymptotic structures have been obtained.
The rest part of our main results are devoted to further applications of Theorem \ref{thm: large_gap}, which refers to 
statistical properties of the counting functions of the tacnode process. 
\begin{corollary}\label{cor: large r of statistics}
	Let $\vec{x}=(x_1,\ldots,x_m)\in\mathbb{R}^{+,m}_{<}$ be fixed. As $r\to+\infty$, we have
	\begin{align}
		& \mathbb{E}\left[N(rx_j)\right]=\mu(rx_j)+\Boh\left(r^{-\frac12}\log r\right), \label{eq: expectation asymptotics}\\
		& \mathrm{Var}\left[N(rx_j)\right]={\sigma(rx_j)}^{2}+\frac{1+\gamma_{E}}{\pi^2}+\Boh\left(r^{-\frac12}(\log r)^2\right), \label{eq: variance asymptotics}
	\end{align}
	where $\gamma_{E}\approx 0.5772$ denotes the Euler's gamma constant, $\mu(rx_j)$ and ${\sigma(rx_j)}^{2}$ are defined through \eqref{def:mu} and \eqref{def:sigma^2} respectively. 
	Furthermore, for $j\neq k$, we have
	\begin{equation}\label{eq: covariance asymptotics}
		\mathrm{Cov}\left(N(rx_j), N(rx_k)\right)=\Sigma(x_j,x_k)+\Boh\left(r^{-\frac12}(\log r)^2\right),
	\end{equation}
	where $\Sigma(x_j,x_k)$ is defined in \eqref{def:SigmaCovariance}.
\end{corollary}
\begin{proof}
We start with a direct calculation by using \eqref{eq:generating_function}, which shows that, for any $j,k=1,\ldots,m$,
\begin{equation}\label{eq:partialu_j logF}
	\partial_{u_j}\log F(r\vec{x}, \vec{u})=\frac{\partial_{u_j} F(r\vec{x}, \vec{u})}{F(r\vec{x}, \vec{u})}=\frac{\mathbb{E}\left[N(rx_j)e^{\sum_{\ell=1}^m u_\ell N(rx_\ell)}\right]}{\mathbb{E}\left[e^{\sum_{\ell=1}^m u_\ell N(rx_\ell)}\right]},
\end{equation}
and 
\begin{multline}\label{eq:partialu_ju_k logF}
	\partial_{u_j}\partial_{u_k}\log F(r\vec{x}, \vec{u})=\frac{\partial_{u_j}\partial_{u_k} F(r\vec{x}, \vec{u})}{F(r\vec{x}, \vec{u})}-\frac{\partial_{u_j} F(r\vec{x}, \vec{u})\partial_{u_k} F(r\vec{x}, \vec{u})}{F(r\vec{x}, \vec{u})^2}\\
	=\frac{\mathbb{E}\left[N(rx_j)N(rx_k)e^{\sum_{\ell=1}^m u_\ell N(rx_\ell)}\right]}{\mathbb{E}\left[e^{\sum_{\ell=1}^m u_\ell N(rx_\ell)}\right]}-\frac{\mathbb{E}\left[N(rx_j)e^{\sum_{\ell=1}^m u_\ell N(rx_\ell)}\right]\mathbb{E}\left[N(rx_k)e^{\sum_{\ell=1}^m u_\ell N(rx_\ell)}\right]}{\mathbb{E}\left[e^{\sum_{\ell=1}^m u_\ell N(rx_\ell)}\right]^2}.
\end{multline}
As $\vec{u}=\vec{0}\in\mathbb{R}^{m}$, we obtain the following two identities
\begin{equation}\label{eq:partialu_j logF at vecu=0}
	\partial_{u_j}\log F(r\vec{x}, \vec{u})|_{\vec{u}=\vec{0}}=\mathbb{E}\left[N(rx_j)\right], \quad \partial_{u_j}\partial_{u_k}\log F(r\vec{x}, \vec{u})|_{\vec{u}=\vec{0}}=\mathrm{Cov}\left(N(rx_j), N(rx_k)\right).
\end{equation}
In particular, taking $j=k$ gives
\begin{equation}\label{eq: partialu_j^2 logF at vecu=0}
	\partial_{u_j}^2\log F(r\vec{x}, \vec{u})|_{\vec{u}=\vec{0}}=\mathrm{Var}\left(N(rx_j)\right).
\end{equation}

Next, using the expansion
\begin{equation*}
	G(1+z)=1+\frac{\log(2\pi)-1}{2}z+\left(\frac{(\log(2\pi)-1)^2}{8}-\frac{1+\gamma_{E}}{2}\right)z^2+\Boh(z^3), \quad \ z\to 0,
\end{equation*}
we obtain that, for each fixed $j=1,\ldots,m$, as $u_j\to 0$, 
\begin{equation}\label{eq: EulerConstantEmerging}
	\log\left(G\left(1+\frac{u_j}{2\pi\ii}\right)G\left(1-\frac{u_j}{2\pi\ii}\right)\right)=\frac{1+\gamma_{E}}{4\pi^2}u_j^2+\Boh(u_j^3).
\end{equation}
Differentiating \eqref{eq: large_gap_asymptotics_generating_function} and using \eqref{eq:derivative of error term}, we then obtain, for each fixed $j=1,\ldots,m$,
\begin{align*}
\partial_{u_j}\log F(r\vec{x},\vec{u})\big|_{\vec{u}=\vec{0}}
&=\mu(rx_j)+\Boh\left(r^{-\frac12}\log r\right),\\
\partial_{u_j}^2\log F(r\vec{x},\vec{u})\big|_{\vec{u}=\vec{0}}
&={\sigma(rx_j)}^2+\frac{1+\gamma_E}{\pi^2}+\Boh\left(r^{-\frac12}(\log r)^2\right),
\end{align*}
where the second identity also uses \eqref{eq: EulerConstantEmerging}. For $j\neq k$, differentiating \eqref{eq: large_gap_asymptotics_generating_function} twice and using \eqref{eq:derivative of error term} yields
\begin{equation*}
\partial_{u_j}\partial_{u_k}\log F(r\vec{x},\vec{u})\big|_{\vec{u}=\vec{0}}
=\Sigma(x_j,x_k)+\Boh\left(r^{-\frac12}(\log r)^2\right).
\end{equation*}
Combining these identities with \eqref{eq:partialu_j logF at vecu=0}--\eqref{eq: partialu_j^2 logF at vecu=0} proves the corollary. 
\end{proof}
\begin{remark}
	The asymptotic formulae for the expectation and variance in \eqref{eq: expectation asymptotics} and \eqref{eq: variance asymptotics} are not new. When $m=1$ and $x_1=1$, they reduce to the corresponding results in \cite[Equations (2.33) and (2.34)]{YZ2024}.  
	The new feature of Corollary \ref{cor: large r of statistics} lies in the covariance asymptotics \eqref{eq: covariance asymptotics} for $j\neq k$, which reflect interactions between different $x_j$'s that are absent in the $1$-point case.
\end{remark}

The next corollary concerns the central limit theorem for the counting functions of the tacnode process, which 
shows that the joint fluctuation of these $m$ counting functions converges in distribution as $r\to+\infty$ to a multivariate normal distribution.
\begin{corollary}\label{cor: large r of fluctuation}
	Fix $\vec{x}$ such that $0<x_1<\dots<x_m<+\infty$, and consider the random variable $N^{(j)}(r)$ defined as 
	\begin{equation}
		N^{(j)}(r) = \frac{N(rx_j) - \mu(rx_j)}{{\sqrt{{\sigma(rx_j)}^2}}}, \quad j=1,\ldots,m.
	\end{equation}
	As $r\to+\infty$, we have 
	\begin{equation}
		\begin{pmatrix}\label{eq: convergence in distribution}
			N^{(1)}(r) & N^{(2)}(r) & \cdots & N^{(m)}(r)
		\end{pmatrix}^{\rm T}
		\overset{d}{\longrightarrow} \mathcal{N}(\vec{0}, I_m),
	\end{equation}
	where $\mathcal{N}$ denotes the multivariate normal distribution with mean vector $\vec{0}\in\mathbb{R}^m$ and covariance matrix $I_m$.
\end{corollary}
\begin{proof}
	For any fixed real numbers $\{\nu_j\}_{j=1}^{m}$, it follows from \eqref{eq:generating_function} and the definition of $N^{(j)}(r)$ that
	\begin{align}
		\mathbb{E}\left[\exp\left(\sum_{j=1}^m \nu_j N^{(j)}(r)\right)\right]
		&=\exp\left(-\sum_{j=1}^m \frac{\nu_j\mu(rx_j)}{\sqrt{{\sigma(rx_j)}^2}}\right)
		\mathbb{E}\left[\exp\left(\sum_{j=1}^m u_j N(rx_j)\right)\right] \nonumber\\
		&=\exp\left(-\sum_{j=1}^m \frac{\nu_j\mu(rx_j)}{\sqrt{{\sigma(rx_j)}^2}}\right)F(r\vec{x},\vec{u}), \label{eq: moment generating function of linear combination}
	\end{align}
	where we set $u_j:={\nu_j}/{\sqrt{{\sigma(rx_j)}^2}}$, $j=1,\ldots,m$.
	Since ${\sigma(rx_j)}^2=\frac{3}{2\pi^2}\log r+\Boh(1)$ as $r\to+\infty$,
	we have $u_j=\Boh\big((\log r)^{-\frac12}\big)$ for each $j=1,\ldots,m$. 
	Now substituting \eqref{eq: large_gap_asymptotics_generating_function} into \eqref{eq: moment generating function of linear combination}, 
	we obtain
	\begin{align}
		\mathbb{E}\left[\exp\left(\sum_{j=1}^m \nu_j N^{(j)}(r)\right)\right]=\exp\left(\sum_{j=1}^{m}\frac{\nu_j^2}{2}+\Boh\left((\log r)^{-1}\right)\right), \qquad r\to +\infty.
	\end{align}
	Thus, for every fixed \(\vec{\nu}\in\mathbb{R}^m\), the moment generating function of
\(\sum_{j=1}^m \nu_j N^{(j)}(r)\) converges to \(\exp\bigl(\frac12\sum_{j=1}^m \nu_j^2\bigr)\), 
which is the moment generating function of $\mathcal{N}(0,\sum_{j=1}^m \nu_j^2)$. Hence
\begin{equation*}
   \sum_{j=1}^m \nu_j N^{(j)}(r) \overset{d}{\longrightarrow}
   \mathcal{N}\left(0,\sum_{j=1}^m \nu_j^2\right).
\end{equation*}
Since this limiting distribution is exactly the distribution of the same linear combination
of an \(\mathcal{N}(\vec{0},I_m)\) vector, the Cram\'{e}r-Wold theorem \cite{CW1936} yields \eqref{eq: convergence in distribution}.
\end{proof}

Another application of Theorem \ref{thm: large_gap} concerns the global rigidity of the tacnode process, or equivalently, the \textit{maximum} fluctuation of the counting function.
Indeed, it follows from \cite[Theorem 1.1]{CC20} that
\begin{equation}
	\lim_{r\to +\infty} \mathbb{P}\left(\sup_{x>r}\left\vert \frac{N(x)-\mu (x)}{\log x} \right\vert\leq \frac{3\sqrt{2}}{2\pi}+\varepsilon\right)=1.
\end{equation}
This upper bound was also obtained in \cite{YZ2024} in the case $m=1$.
A natural question is whether this upper bound for the maximum fluctuation is optimal, that is, whether the following lower bound also holds:
\begin{equation}
	\lim_{r\to +\infty} \mathbb{P}\left(\sup_{x>r}\left\vert \frac{N(x)-\mu (x)}{\log x} \right\vert\geq \frac{3\sqrt{2}}{2\pi}-\varepsilon\right)=1.
\end{equation}
Addressing this problem would require a more delicate analysis of $\mathbb{E}\left[e^{u_j N(rx_j)+u_k N(rx_k)}\right]$ in a different regime where 
$r\to +\infty$ and $|x_j-x_k|\to 0$ simultaneously, and we leave this direction for future study.

\paragraph{Outline of the paper.} The rest of this paper is devoted to the proofs of our main results.
Following a general framework established in \cites{BD02,DIZ97}, the starting point is an identity that relates
$\partial_r\log F(r\vec{x}, \vec{u})$ to a $4\times 4$ RH problem for $X$ with constant jumps, which is presented in 
Section \ref{sec:differential identities}. We then derive a Lax pair for $X$ in Section \ref{sec: lax pair}, and the associated system of 
coupled nonlinear differential equations following from compatibility conditions. Several remarkable differential identities 
of the Hamiltonian are also include therein for later use. 
We perform the Deift-Zhou steepest descent analysis \cite{DZ1993} on 
the RH problem for $X$ as $r\to+\infty$ and $r\to 0^{+}$ in Section \ref{sec: asymptotic analysis of RH large r}
and Section \ref{sec: asymptotic analysis of RH small r}, respectively. 
The asymptotic outcomes, together with the differential identities for the Hamiltonian, are then utilized to establish 
our main results in Section \ref{sec: proofs of main results}.

\section{Differential identities for $F$}\label{sec:differential identities}
The main result of this section is to express the logarithmic derivative of the generating function
$\partial_{r}\log F(r\vec{x}, \vec{u})$ in terms of a solution $X$ to a $4\times 4$ RH problem.

It follows from \eqref{eq:Fredholm_determinant} that
\begin{equation}\label{eq:Fredholm_determinant_scale_r}
F(r\vec{x}, \vec{u})=\det\left(1-\tilde{\mathcal{K}}^{\rm tac}\right), \quad \tilde{\mathcal{K}}^{\rm tac}=\sum_{j=1}^m (1-s_j)\mathcal{K}^{\rm tac}|_{rA_j},
\end{equation}
where we recall that $s_j=e^{u_j+\cdots+u_m}\in (0, +\infty)$, $A_j$ is defined in \eqref{def: interval A_j}, $j=1,\ldots,m$,
and $\mathcal{K}^{\rm tac}|_{rA_j}$ is the trace class operator acting on $L^2(rA_j)$ whose kernel is $K^{\rm tac}$.
Let $\tilde{K}^{\tac}$ be the kernel of the operator $\tilde{\mathcal{K}}^{\tac}$ appearing in \eqref{eq:Fredholm_determinant_scale_r}, that is,
\begin{equation}\label{def:tildeKtac(x,y)}
\tilde{K}^{\tac}(x,y) = \sum_{j=1}^m (1-s_j) K^{\tac}(x,y)\chi_{rA_j}(y).
\end{equation}
From the RH representation of the tacnode kernel in \eqref{def:tacnode kernel}, we can rewrite $\tilde{K}^{\tac}(x,y)$ as
an integrable form in the sense of IIKS \cite{IIKS}, that is,
\begin{equation}\label{eq:tildeKtac(x,y)-RH}
	\tilde{K}^{\tac}(x,y) = \frac{\vec{f}(x)^{\rm T} \vec{h}(y)}{x-y},
\end{equation}
where
\begin{equation}\label{def:vecf and vech}
\begin{aligned}
	&\vec{f}(x)=\widehat M (x) \begin{pmatrix} 1 \\ 1 \\ 0 \\ 0 \end{pmatrix},  \\ 
	&\vec{h}(y)=\frac{\sum_{j=1}^{m} (1-s_j) \chi_{rA_j}(y)}{2\pi i}\widehat M (y)^{-{\rm T}} \begin{pmatrix} 0 \\ 0 \\ 1 \\ 1 \end{pmatrix}
	:=\sum_{j=1}^{m} \mathfrak{s}_j \chi_{rA_j^{\#}}(y)\widehat{M}(y)^{-\rm T}\begin{pmatrix}
		0 \\ 0 \\ 1 \\ 1
	\end{pmatrix}, 
\end{aligned}
\end{equation}
with $\mathfrak{s}_j=\frac{s_{j+1}-s_j}{2\pi \ii}$ and $A_j^{\#}=(-x_j,x_j)$, $j=1,\dots,m$. Here we use the convention $s_{m+1}=1$. 
Since $\vec{u}\in\mathbb{R}^{m}$, it follows from \eqref{eq:generating_function} that $F(r\vec{x}, \vec{u})>0$ for all $r>0$. 
Thus, by \eqref{eq:Fredholm_determinant_scale_r}, it follows that $\det(1-\tilde{\mathcal{K}}^{\tac})>0$ and in particular
$1-\tilde{\mathcal{K}}^{\tac}$ is invertible. Now using the well-known Jacobian formula for trace-class operators, we obtain
\begin{multline}\label{eq:derivative_logF}
\partial_{r}\log F(r\vec{x}, \vec{u})=\partial_{r}\log\det(1-\tilde{\mathcal{K}}^{\tac})=-\mathrm{Tr}\left((1-\tilde{\mathcal{K}}^{\tac})^{-1}\partial_{r}\tilde{\mathcal{K}}^{\tac}\right)  \\
=-\sum_{j=1}^{m} x_j\left(\lim_{v\nearrow rx_j} \mathrm{R}(v,v)+\lim_{v\searrow -rx_j} \mathrm{R}(v,v) \right)+\sum_{j=1}^{m-1} x_j\left(\lim_{v\searrow rx_j} \mathrm{R}(v,v)+\lim_{v\nearrow -rx_j} \mathrm{R}(v,v) \right),
\end{multline}
where $\mathrm{R}$ is the kernel of the resolvent operator $(1-\tilde{\mathcal{K}}^{\tac})^{-1}\tilde{\mathcal{K}}^{\tac}$.
By \cite[Lemma 2.12]{DIZ97}, the kernel $\mathrm{R}$ can be written as 
\begin{equation}\label{eq:resolvent_kernel}
	\mathrm{R}(u,v)=\frac{\vec{\bf F}(u)^{\rm T} \vec{\bf H}(v)}{u-v}, \quad u,v\in\realR,
\end{equation}
where 
\begin{equation}\label{eq: vecF and vecH}
	\vec{\bf F}(u)=(1-\tilde{\mathcal{K}}^{\tac})^{-1}\vec{f}(u)=Y_{+}\vec{f}(u), \quad \vec{\bf H}(v)=(1-\tilde{\mathcal{K}}^{\tac})^{-{\rm T}}\vec{h}(v)=Y_{+}^{-\rm T}\vec{h}(v).
\end{equation}
Here $Y$ is given by
\begin{equation}\label{def:Y}
	Y(z)=I-\int_{-rx_m}^{rx_m} \frac{\vec{\bf F}(w)\vec{h}(w)^{\rm T}}{w-z}\dif w, 
\end{equation}
which is the unique solution to the following RH problem.
\begin{paragraph}{RH problem for $Y$}
	\begin{enumerate}
\item[\rm (a)] $Y: \C\setminus[-rx_m,rx_m] \to \C^{4\times 4}$ is analytic.

\item[\rm (b)] For $x\in(-rx_m,rx_m)\setminus \cup_{j=1}^{m-1}\{\pm rx_j\}$, we have
\begin{equation}\label{eq:Y-jump}
 Y_+(x)=Y_-(x)(I-2\pi \ii \vec{f}(x)\vec{h}(x)^{\rm T}),
 \end{equation}
where the functions $\vec{f}$ and $\vec{h}$ are defined in \eqref{def:vecf and vech}.

\item[\rm (c)] As $z \to \infty$, we have
\begin{equation}\label{eq:Y-infty}
 Y(z)=I+\frac{Y_1}{z}+ \Boh(z^{-2}),
 \end{equation}
where $Y_1$ is independent of $z$.
\item[\rm (d)] As $z \to  \pm rx_j$, $j=1,\dots,m$, we have $Y(z) = \mathcal \Boh(\log(z \mp rx_j))$.
\end{enumerate}
\end{paragraph}

We now transform the RH problem for $Y$ to a new one with constant jumps by using the RH problem for $M$, 
which is also known as an undressing procedure. We start with definitions
\begin{equation}\label{def:Gammajr}
 \begin{aligned}
     &\Gamma_0^{(r)}:=(rx_m,+\infty), &&  \Gamma_1^{(r)}:=rx_m+e^{\varphi \ii }(0,+\infty),  && \Gamma_2^{(r)}:=-rx_m+e^{-\varphi \ii}(-\infty,0),\\
     &\Gamma_3^{(r)}:= (-\infty,-rx_m),&& \Gamma_4^{(r)}:=-rx_m+e^{\varphi \ii}(-\infty,0), &&\Gamma_5^{(r)}:=rx_m+e^{-\varphi \ii}(0,+\infty), \quad 0<\varphi<\frac{\pi}{3}.\\ 
    \end{aligned}
\end{equation}
Clearly, the rays $\Gamma_k^{(r)}$, $k=1,2,4,5$, and $\mathbb{R}$ divide the whole complex plane into six regions $\Omega_j^{(r)}$, $j=1,\ldots,6$; see Figure \ref{fig:X} for an illustration. 

We now define $X(z)=X(z;r\vec{x},\vec{u},\vr_1,\vr_2,\vs_1,\vs_2,\tau)$ by
\begin{align}\label{eq:YtoX}
 &X(z) =
 \begin{cases}
   Y(z)M(z), & \hbox{ $z \in \Omega_1^{(r)}\cup \Omega_3^{(r)} \cup \Omega_4^{(r)} \cup \Omega_6^{(r)} $,} \\
   Y(z)\widehat M(z), & \hbox{ $z \in \Omega_2^{(r)}$,} \\
   Y(z)\widehat M(z)\begin{pmatrix} 1 &0 &-1 &-1 \\0 &1 &-1 &-1 \\ 0&0&1&0\\0&0&0&1 \end{pmatrix}, & \hbox{ $z \in \Omega_5^{(r)}$.}
 \end{cases}
\end{align}
On account of the RH problems for $M$ and $Y$, it is straightforward to check that $X$ satisfies the following RH problem.
\begin{figure}[t]
\begin{center}
   \setlength{\unitlength}{1truemm}
   \begin{picture}(100,70)(-5,2)
       \put(25,40){\line(-1,0){30}}
       \put(55,40){\line(1,0){30}}

       \put(25,40){\line(1,0){30}}
       \put(25,40){\line(-1,-1){25}}
       \put(25,40){\line(-1,1){25}}

       \put(55,40){\line(1,1){25}}
       \put(55,40){\line(1,-1){25}}

       \put(15,40){\thicklines\vector(1,0){1}}
       \put(65,40){\thicklines\vector(1,0){1}}

       \put(10,55){\thicklines\vector(1,-1){1}}
       \put(10,25){\thicklines\vector(1,1){1}}
       \put(70,25){\thicklines\vector(1,-1){1}}
       \put(70,55){\thicklines\vector(1,1){1}}

       \put(-2,11){$\Gamma_4^{(r)}$}

       \put(-2,67){$\Gamma_2^{(r)}$}
       \put(0,42){$\Gamma_3^{(r)}$}
       \put(80,11){$\Gamma_5^{(r)}$}
       \put(80,67){$\Gamma_1^{(r)}$}
       \put(76,42){$\Gamma_0^{(r)}$}

       \put(10,46){$\Omega_3^{(r)}$}
       \put(10,34){$\Omega_4^{(r)}$}
       \put(68,46){$\Omega_1^{(r)}$}
       \put(68,34){$\Omega_6^{(r)}$}
       \put(38,55){$\Omega_2^{(r)}$}
       \put(38,20){$\Omega_5^{(r)}$}
    
       \put(25,40){\thicklines\circle*{1}}
       \put(55,40){\thicklines\circle*{1}}

       \put(24,36.3){$-rx_m$}
       \put(54,36.3){$rx_m$}

   \end{picture}
\caption{The jump contours and regions of the RH problem for $X$.} 
   \label{fig:X}
\end{center}
\end{figure}
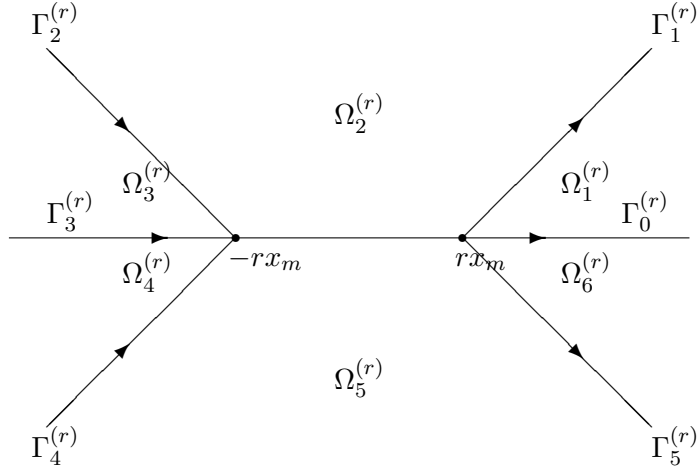

\begin{paragraph}{RH problem for $X$}
\begin{itemize}
\item[\rm (a)] $X(z)$ is defined and analytic in $\mathbb{C}\setminus \Gamma_X$,
where
\begin{equation}\label{def:gammaX}
\Gamma_X:=\bigcup_{j=0}^{5}\Gamma_j^{(r)} \cup [-rx_m,rx_m],
\end{equation}
with the rays $\Gamma_j^{(r)}$, $j=0,1,\ldots,5$, defined in \eqref{def:Gammajr}; see Figure \ref{fig:X} for the orientations of $\Gamma_X$.

\item[\rm (b)] For $z \in \Gamma_X$, we have the jump condition
\begin{equation}\label{eq:X-jump}
 X_+(z)=X_-(z)J_X(z),
\end{equation}
where
\begin{equation}\label{def:JX}
J_X(z):=\left\{
 \begin{array}{ll}
          \begin{pmatrix}0&0&1&0\\0&1&0&0\\-1&0&0&0\\0&0&0&1 \end{pmatrix}, & \qquad \hbox{$z\in \Gamma_0^{(r)}$,} \\
          I-E_{21}+E_{31}+E_{34},  & \qquad  \hbox{$z\in \Gamma_1^{(r)}$,} \\
          I-E_{12}+E_{42}+E_{43},  & \qquad \hbox{$z\in \Gamma_2^{(r)}$,} \\
          \begin{pmatrix}1&0&0&0\\0&0&0&1\\0&0&1&0\\0&-1&0&0 \end{pmatrix}, & \qquad  \hbox{$z\in \Gamma_3^{(r)}$,} \\
          I+E_{12}+E_{42}-E_{43}, & \qquad  \hbox{$z\in \Gamma_4^{(r)}$,} \\
          I+E_{21}+E_{31}-E_{34}, & \qquad  \hbox{$z\in \Gamma_5^{(r)}$,} \\
          \begin{pmatrix}1&0&s_j&s_j\\0&1&s_j&s_j\\0&0&1&0\\0&0&0&1 \end{pmatrix}, & \qquad  \hbox{$z\in rA_j$, \ $j=1,2,\ldots,m$.}
        \end{array}
      \right.
 \end{equation}

\item[\rm (c)] 
As $z \to \infty$ with $z\in \mathbb{C} \setminus \Gamma_X$, we have
\begin{align}\label{eq:asyX}
X(z)&=\left( I+\frac{X_1}{z}+ \Boh(z^{-2}) \right) \diag \left((-z)^{-\frac14},z^{-\frac14},(-z)^{\frac14},z^{\frac14} \right)
\nonumber  \\
& \quad \times A \diag \left(
e^{-\theta_1(z)+\tau z}, e^{-\theta_2(z)- \tau z}, e^{\theta_1(z)+\tau
z},e^{\theta_2(z)- \tau z} \right),
\end{align}
where $A$, $\theta_1$ and $\theta_2$ are defined in \eqref{def:A}--\eqref{def:theta2}, respectively and
\begin{equation}\label{def:X1}
X_1 = Y_1 + M_1
\end{equation}
with $Y_1$ and $M_1$ given in \eqref{eq:Y-infty} and \eqref{eq:asy:M}.

\item[\rm (d)] As $z \to rx_j$, $j=1,2,\ldots,m$, we have
\begin{align} \label{eq:X-near-rx_j}
	X(z) &=  X_{R,j}(z) \begin{pmatrix}
	1 & 0 & -\mathfrak{s}_j \log(z-rx_j) & -\mathfrak{s}_j \log(z-rx_j)
   \\
	0 & 1 & -\mathfrak{s}_j \log(z-rx_j)& -\mathfrak{s}_j \log(z-rx_j)
    \\
	0 & 0 & 1 & 0
    \\
    0 & 0 & 0 & 1	
\end{pmatrix}
\nonumber
\\
& \quad \times
\begin{cases}
	I, & z \in \Omega_2^{(r)}, \\
	\begin{pmatrix}
    1&0& -s_{j+1} & -s_{j+1}
    \\
    0& 1 & -s_{j+1} & -s_{j+1}
    \\
    0&0&1&0
    \\
    0&0&0&1
    \end{pmatrix}, & z \in \Omega_5^{(r)},
	\end{cases}
\end{align}
where we recall that $\mathfrak{s}_j:=(s_{j+1}-s_j)/2\pi \ii$, and the principal branch is taken for $\log(z-rx_j)$. 
The matrix $X_{R,j}(z)$ is analytic at $z=rx_j$ satisfying
\begin{equation}\label{eq: Phi-expand-rx_j}
  X_{R,j}(z) = X_{R,j}^{(0)}(r)\left(I+X_{R,j}^{(1)}(r)(z-rx_j)+\Boh((z-rx_j)^2)\right),\qquad z\to r x_j,
 \end{equation}
for some functions $X_{R,j}^{(0)}(r)$ and $ X_{R,j}^{(1)}(r)$ depending on the parameters $\vr_1$,$\vr_2$,$\vs_1$,$\vs_2$,$\tau$ and $\vec{u}$.

\item[\rm (e)] As $z \to -rx_j$, $j=1,2,\ldots,m$, we have
\begin{align} \label{eq:X-near--rx_j}
	X(z) &=  X_{L,j}(z) \begin{pmatrix}
	1 & 0 & -\mathfrak{s}_j \log(-z-rx_j) & -\mathfrak{s}_j \log(-z-rx_j)
   \\
	0 & 1 & -\mathfrak{s}_j \log(-z-rx_j) & -\mathfrak{s}_j \log(-z-rx_j)
	\\
	0 & 0 & 1 & 0
	\\
	0 & 0 & 0 & 1	
\end{pmatrix}
\nonumber
\\
& \quad \times
\begin{cases}
	 \begin{pmatrix}
	0&1& 0 & 0
	\\
	1& 0 & 0 & 0
	\\
	0&0&0&-1
	\\
	0&0&-1&0
	\end{pmatrix}, & z \in \Omega_2^{(r)}, \\
	\begin{pmatrix}
	0&1& s_{j+1} & s_{j+1}
	\\
	1& 0 & s_{j+1} & s_{j+1}
	\\
	0&0&0&-1
	\\
	0&0&-1&0
	\end{pmatrix}, & z \in \Omega_5^{(r)},
	\end{cases}
\end{align}
where the principal branch is taken for $\log(-z-rx_j)$, 
and the matrix $X_{L,j}(z)$ is analytic at $z=-rx_j$ satisfying
\begin{equation}\label{eq: Phi-expand--rx_j}
  X_{L,j}(z) = X_{L,j}^{(0)}(r)\left(I+X_{L,j}^{(1)}(r)(z+rx_j)+\Boh((z+rx_j)^2)\right),\qquad z\to -r x_j,
 \end{equation}
for some functions $X_{L,j}^{(0)}(r)$ and $ X_{L,j}^{(1)}(r)$ depending on the parameters $\vr_1$,$\vr_2$,$\vs_1$,$\vs_2$,$\tau$ and $\vec{u}$.

\item[\rm (f)] We have the following symmetry relations for $X$ and $X_1$:
\begin{align} \label{eq:symm}
  & X (-z) =
  \begin{pmatrix}
  \sigma_1 & 0
  \\
  0 & -\sigma_1
  \end{pmatrix}
  \widetilde X(z)
  \begin{pmatrix}
  \sigma_1 & 0
  \\
  0 & -\sigma_1
  \end{pmatrix},
\\
  & {X (z)}^{- \rm T}=
  \begin{pmatrix}
  0 & -I_2
  \\
  I_2 & 0
  \end{pmatrix}\dot X(z)
  \begin{pmatrix}
  0 & I_2
  \\
  -I_2 & 0
  \end{pmatrix},\label{eq:symmX2}
\end{align}
where $\sigma_1$ is given in \eqref{def:Pauli}, $\widetilde X$ and $\dot X$ are defined through \eqref{def:tildeX} and \eqref{def:dotX}. Moreover, the matrix
$X_1 = X_1(\vec{x}, \vec{u}; \vr_1,\vr_2,\vs_1,\vs_2,\tau)$ in \eqref{eq:asyX} satisfies
\begin{align}
 & X_1 =
  -\begin{pmatrix}
  \sigma_1 & 0
  \\
  0 & -\sigma_1
  \end{pmatrix} \widetilde X_1
  \begin{pmatrix}
  \sigma_1 & 0
  \\
  0 & -\sigma_1
  \end{pmatrix},\label{eq:symmX11}
\\
  & X_1^{\rm T} =
  -\begin{pmatrix}
  0 & -I_2
  \\
  I_2 & 0
  \end{pmatrix}\dot X_1
  \begin{pmatrix}
  0 & I_2
  \\
  -I_2 & 0
  \end{pmatrix}.\label{eq:symmX12}
\end{align}
\end{itemize}
\end{paragraph}

\begin{proof}
	The proof of the items (a)--(c) is straightforward by using the definitions of $X$ and the RH problems for $Y$ and $M$. 
	We prove the local behavior of $X$ near $\pm rx_j$, $j=1,2,\ldots,m$, in (d)--(e), and then turn to the symmetry relations in (f).
	
	\begin{itemize}
	\item[\rm (d) and (e)] Fix $\varepsilon>0$ so that $(rx_j-\varepsilon,rx_j+\varepsilon)$ contains no other point of $\{\pm rx_1,\ldots,\pm rx_m\}$. For $z$ near $rx_j$, the representation \eqref{def:Y} yields
	\begin{align*}
	Y(z)
	&=I-\int_{-rx_m}^{rx_j-\varepsilon}\frac{\vec{\bf F}(w)\vec{h}(w)^{\rm T}}{w-z}\dif w-\int_{rx_j+\varepsilon}^{rx_m}\frac{\vec{\bf F}(w)\vec{h}(w)^{\rm T}}{w-z}\dif w
	\\
	&\quad-\int_{rx_j-\varepsilon}^{rx_j}\frac{\vec{\bf F}(w)\vec{h}(w)^{\rm T}}{w-z}\dif w-\int_{rx_j}^{rx_j+\varepsilon}\frac{\vec{\bf F}(w)\vec{h}(w)^{\rm T}}{w-z}\dif w.
	\end{align*}
	By \eqref{def:vecf and vech}, the first two integrals are analytic at $z=rx_j$, whereas the last two produce the logarithmic singularity coming from the jump of $\vec h$ at $rx_j$. Hence, as $z\to rx_j$,
	\begin{equation*}
	Y(z) = Y^{\rm reg}(z)-\vec{\bf F}(rx_j)\vec{h}(rx_j^{-})^{\rm T}\log(z-rx_j)+\vec{\bf F}(rx_j)\vec{h}(rx_j^{+})^{\rm T}\log(z-rx_j),
	\end{equation*}
	where $Y^{\rm reg}$ is analytic at $rx_j$. Moreover, \eqref{def:vecf and vech} gives
	\begin{equation*}
	\vec{h}(rx_j^{+})-\vec{h}(rx_j^{-})=-\mathfrak{s}_j\widehat M(rx_j)^{-\rm T}
	\begin{pmatrix}0\\0\\1\\1\end{pmatrix}.
	\end{equation*}
	Therefore,
	\begin{align*}
	Y(z)&=\widetilde Y^{\rm reg}(z)-\mathfrak{s}_j\vec{\bf F}(rx_j)\begin{pmatrix}0 & 0 & 1 & 1\end{pmatrix}\widehat M(rx_j)^{-1}\log(z-rx_j),
	\end{align*}
	with $\widetilde Y^{\rm reg}$ analytic at $rx_j$. Multiplying on the right by $\widehat M(z)$ and using the analyticity of $\widehat M$ at $rx_j$, the regular part is absorbed into an analytic prefactor $X_{R,j}(z)$. Since \eqref{def:vecf and vech} and \eqref{eq: vecF and vecH} imply
	\begin{equation*}
	\vec{\bf F}(rx_j)=X_{R,j}(rx_j)\begin{pmatrix}1 & 1 & 0 & 0\end{pmatrix}^{\rm T},
	\end{equation*}
	we obtain
	\begin{equation*}
	Y(z)\widehat M(z)=X_{R,j}(z)\bigl(I-\mathfrak{s}_j\log(z-rx_j)E_N\bigr),
	\qquad z\to rx_j,
	\end{equation*}
	where $X_{R,j}(z)$ is analytic at $z=rx_j$, and $E_N=E_{13}+E_{14}+E_{23}+E_{24}$. 
	This proves \eqref{eq:X-near-rx_j} in $\Omega_2^{(r)}$. The corresponding behavior in $\Omega_5^{(r)}$ follows from the jump relation \eqref{def:JX} on $(rx_j,rx_{j+1})$, and the analysis near $-rx_j$ is analogous.

	\item[\rm (f)] The two sides of \eqref{eq:symm} satisfy the same RH problem, so \eqref{eq:symm} follows by uniqueness. The same argument proves \eqref{eq:symmX2}. Substituting the expansion \eqref{eq:asyX} into \eqref{eq:symm} and \eqref{eq:symmX2}, one immediately obtains \eqref{eq:symmX11} and \eqref{eq:symmX12}.
\end{itemize}
\end{proof}

The connection between $\partial_r\log F(r\vec{x}, \vec{u})$ and the RH problem for $X$ is revealed in the following proposition.
\begin{proposition}[Differential identities for $F$]\label{prop:diff_identities}
	With $F$ defined in \eqref{eq:Fredholm_determinant}, we have
\begin{align}\label{eq:diff_identity}
\partial_{r}\log F(r\vec{x}, \vec{u};\vr_1,\vr_2,\vs_1,\vs_2,\tau)&=-\sum_{j=1}^{m} x_j\mathfrak{s}_j \sum_{k=3}^4\sum_{l=1}^2\left[X_{R,j}^{(1)}(r)-X_{L,j}^{(1)}(r)\right]_{kl},
\end{align}
where $X_{R,j}^{(1)}(r)$ and $X_{L,j}^{(1)}(r)$ are defined in \eqref{eq: Phi-expand-rx_j} and \eqref{eq: Phi-expand--rx_j}, respectively. 
We also have the following identities for the derivatives of $\log F$ with respect to $\vs_1$, $\vs_2$ and $\tau$:
\begin{align}
\frac{\partial}{\partial \vs_{1}}\log F(r\vec{x}, \vec{u}; \vr_1, \vr_2, \vs_1, \vs_2, \tau) &= 2\ii\left((X_1)_{13}-(M_1)_{13}\right),
\label{eq:derivativeFs12}\\
\frac{\partial}{\partial \vs_{2}}\log F(r\vec{x}, \vec{u}; \vr_1, \vr_2, \vs_1, \vs_2, \tau) &= 2\ii\left((\widetilde X_1)_{13}-(\widetilde M_1)_{13}\right),
\label{eq:derivativeFs122}\\
\frac{\partial}{\partial \tau}\log F(r\vec{x}, \vec{u}; \vr_1, \vr_2, \vs_1, \vs_2, \tau) &= -(X_1)_{11}- (\widetilde X_1)_{11}+(\dot X_1)_{11}+(\dot{\widetilde{X}}_1)_{11}\nonumber\\
&\quad+(M_1)_{11}+ (\widetilde M_1)_{11}-(\dot M_1)_{11}-(\dot{\widetilde{M}}_1)_{11},
\label{eq:derivativeFtau}
\end{align}
where $X_1$ and $M_1$ are given in \eqref{eq:asyX} and \eqref{eq:asy:M}, respectively, and $\dot{\widetilde{f}}=f(\cdot;\vr_2,\vr_1,\vs_2,\vs_1,-\tau)$ is defined by combining \eqref{def:tildeX} and \eqref{def:dotX}.
\end{proposition}
\begin{proof}
	By \eqref{def:vecf and vech}, \eqref{eq: vecF and vecH} and \eqref{eq:YtoX}, we have 
	\begin{equation}
		\begin{aligned}
		&\vec{\bf F}(v)=X_+(v)\begin{pmatrix}1 \\ 1 \\ 0 \\ 0\end{pmatrix}, \\
		&\vec{\bf H}(v)=\sum_{j=1}^{m} \mathfrak{s}_j \chi_{rA_j^{\#}}(v) {X_{+}(v)}^{-\rm T}\begin{pmatrix}0 \\ 0 \\ 1 \\ 1\end{pmatrix}.
	\end{aligned}
	\end{equation}
Combining the above formulae and \eqref{eq:resolvent_kernel}, we obtain
\begin{align}
	\mathrm{R}(v,v)=\sum_{j=1}^{m} \mathfrak{s}_j \chi_{rA_j^{\#}}(v) \sum_{k=3}^4\sum_{l=1}^2 \left[X_{+}(v)^{-1}X_{+}'(v)\right]_{kl},
\end{align}
which allows us to rewrite \eqref{eq:derivative_logF} as
\begin{align*}
\partial_{r}\log F(r\vec{x}, \vec{u})
&=-\sum_{j=1}^{m} x_j\mathfrak{s}_j \sum_{k=3}^4\sum_{l=1}^2 \left[X_{+}(rx_j)^{-1}X_{+}'(rx_j)+X_{+}(-rx_j)^{-1}X_{+}'(-rx_j)\right]_{kl}.
\end{align*}
It also follows from \eqref{eq:X-near-rx_j}--\eqref{eq: Phi-expand-rx_j} that 
\begin{equation*}
	\sum_{k=3}^4\sum_{l=1}^2 [X_{+}(rx_j)^{-1}X_{+}'(rx_j)]_{kl}=\sum_{k=3}^4\sum_{l=1}^2\left[X_{R,j}^{(1)}(r)\right]_{kl},
\end{equation*}
where $X_{R,j}^{(1)}(r)$ is defined in \eqref{eq: Phi-expand-rx_j}. Similarly, we obtain from \eqref{eq:X-near--rx_j}--\eqref{eq: Phi-expand--rx_j} that
\begin{equation*}
	\sum_{k=3}^4\sum_{l=1}^2 [X_{+}(-rx_j)^{-1}X_{+}'(-rx_j)]_{kl}=-\sum_{k=3}^4\sum_{l=1}^2\left[X_{L,j}^{(1)}(r)\right]_{kl},
\end{equation*}
where $X_{L,j}^{(1)}(r)$ is defined in \eqref{eq: Phi-expand--rx_j}. The above formulae give us \eqref{eq:diff_identity}.

The differential identities \eqref{eq:derivativeFs12}--\eqref{eq:derivativeFtau} involve only the coefficients in the expansions of $X$ and $M$ at infinity, and hence are independent of the local behavior of $X$ near $\pm rx_j$. The proof is therefore the same as in the case $m=1$, and we refer to \cite[Proposition 3.4 and Lemma 3.5]{YZ2024} for the details.

This completes the proof of Proposition \ref{prop:diff_identities}.
\end{proof}

\section{Lax pair, coupled differential equations and Hamiltonian} \label{sec: lax pair}
In this section, we derive a Lax pair for $X(z;r)$, which also depends on the parameters $\vec{x}$, $\vec{u}$, $\vr_1$, $\vr_2$, $\vs_1$, $\vs_2$, and $\tau$.
We also establish several differential identities for the associated Hamiltonian for later use.

\subsection{The Lax pair for $X$ and the associated system}
\begin{proposition}\label{prop:lax pair}
For the matrix-valued function $X(z)=X(z;r)$ defined in \eqref{eq:YtoX}, the following Lax pair holds:
\begin{equation}\label{eq:lax}
\frac{\partial}{\partial z} X(z;r) = L(z;r) X(z;r), \qquad \frac{\partial}{\partial r} X(z;r) = U(z;r) X(z;r),
\end{equation}
where
\begin{equation}\label{eq:L}
L(z;r) =\ii(\vr_1 E_{31}- \vr_2 E_{42})z+A_0(r)+\sum_{j=1}^{m}\left(\frac{A_j(r)}{z-rx_j}+\frac{A_{-j}(r)}{z+rx_j}\right),
\end{equation}
and
\begin{equation}\label{eq:U}
U(z;r) = \sum_{j=1}^{m}\left(-x_j\frac{A_j(r)}{z-rx_j}+x_j\frac{A_{-j}(r)}{z+rx_j}\right)
\end{equation}
with the functions $A_0(r)$ and $A_{\pm j}(r)$, $j=1,\ldots,m$, given in \eqref{def:A0}--\eqref{def:A-j(r)}.
Define the functions $(\{p_{j,k}(r), q_{j,k}(r)\}_{j=1,\ldots,m}^{k=1,\ldots,4}, p_5(r), q_5(r), p_6(r), q_6(r))$ in terms of the RH problem for $X$ by
\begin{align}
p_{5}(r) & = \ii \vr_1 (X_1)_{13}, & p_{6}(r)&  = -\ii \vr_1 (X_1)_{43} -\ii \vr_2 (X_1)_{21},\label{def:p56}\\
q_{5}(r) & = (X_1)_{33}-(X_1)_{11}, & q_{6}(r)&  = (X_1)_{14},\label{def:q56}
\end{align}
and 
\begin{equation}\label{def:qkpk}
\vec{q}_j(r):=\begin{pmatrix}
q_{j,1}(r)\\q_{j,2}(r)\\q_{j,3}(r)\\q_{j,4}(r)
\end{pmatrix}=X_{R,j}^{(0)}(r) \begin{pmatrix}
1\\1\\0\\0
\end{pmatrix} \quad \textrm{and} \quad \vec{p}_j(r):=\begin{pmatrix}
p_{j,1}(r)\\p_{j,2}(r)\\p_{j,3}(r)\\p_{j,4}(r)
\end{pmatrix}=-\mathfrak{s}_j X_{R,j}^{(0)}(r)^{-\rm T} \begin{pmatrix}
0\\0\\1\\1
\end{pmatrix},
\end{equation}
for $j=1,\ldots,m$. Then these functions satisfy the system \eqref{eq:coupled system}, the constraint \eqref{eq:extra-condition-1}, and
\begin{equation}\label{eq:extra-condition-2}
	\sum_{j=1}^{m}\left(q_{j,1}(r)p_{j,3}(r) -\widetilde q_{j,2}(r) \widetilde p_{j,4}(r) \right) = -\ii \vr_1 q_5(r)-\frac{p_5(r)^2}{\ii \vr_1}+\ii \vr_2 q_6(r)\widetilde q_6(r)-\ii\vs_1.
\end{equation}
\end{proposition}

\begin{proof}
The proof is based on the RH problem for $X$. Since the jump matrices given in \eqref{def:JX}
are all independent of $z$ and $r$, it follows from the analyticity of $X$ that 
\begin{equation}
L(z;r) := \frac{\partial X(z;r)}{\partial z}{X(z;r)}^{-1}, \qquad U(z;r) := \frac{\partial X(z;r)}{\partial r}{X(z;r)}^{-1}
\end{equation}
are analytic in $\mathbb{C}$, with possible isolated poles at $z=\pm rx_j$, $j=1,2,\ldots,m$, and at $z=\infty$.
It remains to determine the explicit forms of $L(z;r)$ and $U(z;r)$ by analyzing the behavior of $X$ near these points.

From the asymptotic expansion of $X$ in \eqref{eq:asyX}, we obtain, as $z\to\infty$,
\begin{align}\label{eq:inftyL}
L(z;r)=\ii(\vr_1 E_{31}- \vr_2 E_{42})z+A_0(r)+\frac{L_1}{z}+\Boh(z^{-2}),
\end{align}
where $X_1$ is given in \eqref{eq:asyX}, and
\begin{align}
	&A_0(r)=
	\begin{pmatrix}
		\tau & 0 & \ii \vr_1 & 0\\
		0 & -\tau & 0 & \ii \vr_2\\
		-\ii \vs_1 & 0 & \tau & 0\\
		0 & -\ii \vs_2 & 0 & -\tau
	\end{pmatrix}
	+\left[X_1,\ii(\vr_1 E_{31}- \vr_2 E_{42})\right], \label{def:A0-via-RH}\\
	& L_1=-\ii \vs_1 E_{13}+\ii \vs_2 E_{24}+\diag\left(-\frac14,-\frac14,\frac14,\frac14\right)+
	\left[X_1, \begin{pmatrix}
		\tau & 0 & \ii \vr_1 & 0\\
		0 & -\tau & 0 & \ii \vr_2\\
		-\ii \vs_1 & 0 & \tau & 0\\
		0 & -\ii \vs_2 & 0 & -\tau
	\end{pmatrix}\right] \nonumber\\
	&\hspace*{2.5em}+\left[X_2, \ii(\vr_1 E_{31}- \vr_2 E_{42})\right]
	+\ii(\vr_1 E_{31}- \vr_2 E_{42})X_1^2-X_1\ii(\vr_1 E_{31}- \vr_2 E_{42})X_1. \label{eq:L1-via-RH}
\end{align}
Here $[A,B]:=AB-BA$ denotes the commutator.
According to the symmetric relation of $X_{1}$ established in \eqref{eq:symmX11}, it follows that
\begin{equation}
\widetilde A_0(r)=-\begin{pmatrix}
\sigma_1& 0\\ 0& -\sigma_1
\end{pmatrix}A_0(r)
\begin{pmatrix}
\sigma_1&0\\0&-\sigma_1
\end{pmatrix},
\end{equation}
where $\sigma_1$ is given in \eqref{def:Pauli}. 
Thus, defining $p_5(r), q_5(r), p_6(r), q_6(r)$ via \eqref{def:p56}--\eqref{def:q56}, we obtain
the expression for $A_0(r)$ in \eqref{def:A0} by combining \eqref{def:p56}, \eqref{def:q56} and \eqref{def:A0-via-RH}.

Now as $z\to rx_j$, $j=1,\ldots,m$, it follows from the local behavior of $X$ given in \eqref{eq:X-near-rx_j}--\eqref{eq: Phi-expand-rx_j} 
that
\begin{equation}
	L(z;r) = \frac{A_j(r)}{z-rx_j}+\Boh(1), \qquad z\to rx_j, \ j=1,\ldots,m,
\end{equation}
where
\begin{equation}\label{def:Aj(r)-via-RH}
A_j(r) = -\mathfrak{s}_j X_{R,j}^{(0)}(r) \begin{pmatrix}
0 & 0 & 1 & 1\\
0 & 0 & 1 & 1\\
0 & 0 & 0 & 0\\
0 & 0 & 0 & 0
\end{pmatrix}X_{R,j}^{(0)}(r)^{-1}
\end{equation}
with $X_{R,j}^{(0)}(r)$ defined in \eqref{eq: Phi-expand-rx_j}. By setting
\eqref{def:qkpk}, we obtain the $A_j(r)$ in \eqref{def:Aj(r)} for $j=1,2,\ldots,m$. 
Using \eqref{def:Aj(r)-via-RH} again, it is also readily seen that 
\begin{equation*}
	{\rm Tr}(A_j(r))=\sum_{k=1}^{4} q_{j,k}(r)p_{j,k}(r)=0, \qquad j=1,2,\ldots,m,
\end{equation*}
which verifies \eqref{eq:extra-condition-1}.
Now turning to the local behavior of $X$ near $-rx_j$, $j=1,\ldots,m$, we have from \eqref{eq:X-near--rx_j}--\eqref{eq: Phi-expand--rx_j} that
\begin{equation}
	L(z;r) = \frac{A_{-j}(r)}{z+rx_j}+\Boh(1), \qquad z\to -rx_j, \ j=1,\ldots,m,
\end{equation}
where \begin{equation}\label{def:A-j(r)-via-RH}
A_{-j}(r) = -\mathfrak{s}_j X_{L,j}^{(0)}(r) \begin{pmatrix}
0 & 0 & 1 & 1\\
0 & 0 & 1 & 1\\
0 & 0 & 0 & 0\\
0 & 0 & 0 & 0
\end{pmatrix}X_{L,j}^{(0)}(r)^{-1}
\end{equation}
with $X_{L,j}^{(0)}(r)$ defined in \eqref{eq: Phi-expand--rx_j}. 
Recalling the symmetric relation \eqref{eq:symm}, it follows from 
\eqref{eq:X-near-rx_j}--\eqref{eq: Phi-expand--rx_j} that
\begin{equation}\label{eq:symmXLR}
	X_{L,j}^{(0)}(r) = \begin{pmatrix}
\sigma_1 & 0\\
0 & -\sigma_1
\end{pmatrix} \widetilde X_{R,j}^{(0)}(r).
\end{equation}
This, together with \eqref{def:Aj(r)-via-RH} and \eqref{def:A-j(r)-via-RH}, imply 
that
\begin{equation}
	A_{-j}(r) = 
\begin{pmatrix}
\sigma_1 & 0\\
0 & -\sigma_1
\end{pmatrix} \widetilde A_j(r) \begin{pmatrix}
\sigma_1 & 0\\
0 & -\sigma_1
\end{pmatrix}, \qquad j=1,2,\ldots,m,
\end{equation} 
which yields the formula of $A_{-j}(r)$ in \eqref{def:A-j(r)}.
It follows from the expression of $L(z;r)$ in \eqref{eq:L} 
that the $\mathcal{O}(z^{-1})$ term as $z\to \infty$ is given by
\begin{equation*}
	L_1 = \sum_{j=1}^{m} (A_j(r)+A_{-j}(r)).
\end{equation*}
Comparing the $(1,3)$-entry of the above equation with that of $L_1$ given in \eqref{eq:L1-via-RH}, 
and using \eqref{def:p56}, \eqref{def:q56}, \eqref{def:Aj(r)} and \eqref{def:A-j(r)},
we arrive at \eqref{eq:extra-condition-2}.

The explicit formula of $U(z;r)$ in \eqref{eq:U} could be obtained in a similar manner. Indeed, by using \eqref{eq:X-near-rx_j}--\eqref{eq: Phi-expand--rx_j}
one can check that 
\begin{equation}
	U(z;r)=\mathcal{O}(z^{-1}), \qquad z\to \infty,
\end{equation}
and for $j=1,\ldots,m$,
\begin{equation}
	U(z;r) = -x_j\frac{A_j(r)}{z-rx_j}+\Boh(1), \quad z\to rx_j; \qquad U(z;r) = x_j\frac{A_{-j}(r)}{z+rx_j}+\Boh(1), \quad z\to -rx_j,
\end{equation}
where $A_j(r)$ and $A_{-j}(r)$ are given in \eqref{def:Aj(r)} and \eqref{def:A-j(r)}, respectively. The above formulae give us the expression of $U(z;r)$ in \eqref{eq:U}.

It remains to show that the functions $\{p_{j,k}(r), q_{j,k}(r)\}_{j=1,\ldots,m}^{k=1,\ldots,4}$, $p_5(r)$, $q_5(r)$, $p_6(r)$, $q_6(r)$ satisfy the desired differential equations.
On one side, we note the compatibility condition $\partial_z\partial_r X(z;r)=\partial_r\partial_z X(z;r)$
is equivalent to the zero curvature equation
\begin{equation}\label{eq:zero-curvature}
\frac{\partial L(z;r)}{\partial r}-\frac{\partial U(z;r)}{\partial z}=[U(z;r), L(z;r)].
\end{equation}
On the other side, by \eqref{eq:L} and \eqref{eq:U}, we have
\begin{equation}
\partial_r L(z;r)-\partial_z U(z;r) = A_0'(r)+\sum_{j=1}^{m}\left(\frac{A_j'(r)}{z-rx_j}+\frac{A_{-j}'(r)}{z+rx_j}\right).
\end{equation}
Thus 
\begin{equation}\label{eq:zero-curvature2}
	A_0'(r)+\sum_{j=1}^{m}\left(\frac{A_j'(r)}{z-rx_j}+\frac{A_{-j}'(r)}{z+rx_j}\right)=[U(z;r), L(z;r)].
\end{equation}
Inserting \eqref{eq:L} and \eqref{eq:U} into the above equation, and then taking $z\to\infty$, 
we obtain
\begin{equation}
	A_0'(r)=\sum_{j=1}^{m} x_j\left[A_{-j}(r)-A_j(r), \ \ii\left(\vr_1E_{31}-\vr_2E_{42}\right)\right],
\end{equation}
which yields the last four equations in \eqref{eq:coupled system} by using \eqref{def:A0}--\eqref{def:A-j(r)}.
We now prove the first four equations in \eqref{eq:coupled system}.  
On account of \eqref{eq:lax}--\eqref{eq:U}, a direct calculation shows that 
\begin{align}\label{eq:xjL+U-1}
x_jL(z;r)+U(z;r)&=\left(x_j\partial_z X(z;r)+\partial_r X(z;r)\right)X(z;r)^{-1} \nonumber\\
& =\partial_rX_{R,j}^{(0)}(r)\cdot X_{R,j}^{(0)}(r)^{-1}+o(1), \qquad z\to rx_j, \ j=1,\ldots,m,
\end{align}
where we have used the local behavior of $X$ near $rx_j$ given in \eqref{eq:X-near-rx_j}--\eqref{eq: Phi-expand-rx_j}
for the second equality. On the other hand, by \eqref{eq:L} and \eqref{eq:U}, we have, as $z\to rx_j$,
\begin{align*}
L(z;r)&=\frac{A_j(r)}{z-rx_j}+rx_j\ii(\vr_1E_{31}-\vr_2E_{42})+A_0(r)
+\sum_{\substack{k=1 \\ k \neq j}}^{m}\frac{A_k(r)}{rx_j-rx_k}
+\sum_{k=1}^{m}\frac{A_{-k}(r)}{rx_j+rx_k}+\Boh(z-rx_j),
\\
U(z;r)&=-x_j\frac{A_j(r)}{z-rx_j}
-\sum_{\substack{k=1 \\ k \neq j}}^{m}\frac{x_kA_k(r)}{rx_j-rx_k}
+\sum_{k=1}^{m}\frac{x_kA_{-k}(r)}{rx_j+rx_k}+\Boh(z-rx_j),
\end{align*}
which imply that, as $z\to rx_j$,
\begin{multline}
	x_jL(z;r)+U(z;r) = r x_j^2\ii(\vr_1E_{31}-\vr_2E_{42})+x_jA_0(r)+\frac1r\sum_{\substack{k=1 \\ k \neq j}}^{m}A_k(r)+\frac1r\sum_{k=1}^{m}A_{-k}(r)+\Boh(z-rx_j) \\
	=r x_j^2\ii(\vr_1E_{31}-\vr_2E_{42})+x_jA_0(r)+\frac1r\sum_{\substack{k=-m \\ k \neq 0}}^{m}A_k(r)-\frac{1}{r}A_j(r)+\Boh(z-rx_j).
\end{multline}	
Comparing \eqref{eq:xjL+U-1} with the above equation, we obtain
\begin{equation}
\partial_rX_{R,j}^{(0)}(r) = \left(r x_j^2\ii(\vr_1E_{31}-\vr_2E_{42})+x_jA_0(r)+\frac1r\sum_{\substack{k=-m \\ k \neq 0}}^{m}A_k(r)-\frac{1}{r}A_j(r)\right) \cdot X_{R,j}^{(0)}(r).
\end{equation}
Recalling \eqref{def:qkpk}, it is readily seen that
\begin{equation*}
\begin{pmatrix}
q_{j,1}'(r)\\q_{j,2}'(r)\\q_{j,3}'(r)\\q_{j,4}'(r)
\end{pmatrix}=
\left(r x_j^2\ii(\vr_1E_{31}-\vr_2E_{42})+x_jA_0(r)
+\frac1r\sum_{\substack{k=-m \\ k \neq 0}}^{m}A_k(r)-\frac{1}{r}A_j(r)\right)
\begin{pmatrix}
q_{j,1}(r)\\q_{j,2}(r)\\q_{j,3}(r)\\q_{j,4}(r)
\end{pmatrix}.
\end{equation*}
By \eqref{eq:extra-condition-1} and \eqref{def:Aj(r)}, it follows that 
$A_j(r)(q_{j,1}(r),q_{j,2}(r),q_{j,3}(r),q_{j,4}(r))^{\rm T}=0$, thus
\begin{equation}\label{eq:qj'-calculation}
\begin{pmatrix}
q_{j,1}'(r)\\q_{j,2}'(r)\\q_{j,3}'(r)\\q_{j,4}'(r)
\end{pmatrix}=
\left(r x_j^2\ii(\vr_1E_{31}-\vr_2E_{42})+x_jA_0(r)
+\frac1r\sum_{\substack{k=-m \\ k \neq 0}}^{m}A_k(r)\right)
\begin{pmatrix}
q_{j,1}(r)\\q_{j,2}(r)\\q_{j,3}(r)\\q_{j,4}(r)
\end{pmatrix},
\end{equation}
which gives the first four equations in \eqref{eq:coupled system}. 
We finally verify the middle four equations in \eqref{eq:coupled system}. 
Calculating the residue at $z=rx_j$, $j=1,\ldots,m$ on both sides of \eqref{eq:zero-curvature2}
leads to
\begin{equation}\label{eq:A_j'(r)-calculation}
A_j'(r) = 
\left[
r x_j^2\ii(\vr_1E_{31}-\vr_2E_{42})+x_jA_0(r)
+\frac1r\sum_{\substack{k=-m \\ k \neq 0}}^{m}A_k(r), \, A_j(r)
\right],
\end{equation}
where we have used $[A_j(r),A_j(r)]=0$. 
Combining \eqref{eq:A_j'(r)-calculation} with \eqref{eq:qj'-calculation}, 
it follows that 
\begin{multline}
	\begin{pmatrix}
		p_{j,1}'(r) & p_{j,2}'(r) & p_{j,3}'(r) & p_{j,4}'(r)
	\end{pmatrix}\\
	=
	-\begin{pmatrix}
		p_{j,1}(r) & p_{j,2}(r) & p_{j,3}(r) & p_{j,4}(r)
	\end{pmatrix}
	\left(r x_j^2\ii(\vr_1E_{31}-\vr_2E_{42})+x_jA_0(r)+\frac1r\sum_{\substack{k=-m \\ k \neq 0}}^{m}A_k(r)\right),
\end{multline}
which gives the middle four equations in \eqref{eq:coupled system}.

This completes the proof.
\end{proof}

\subsection{Differential identities for the Hamiltonian}
From the general theory of Jimbo-Miwa-Ueno \cite{JMU81}, the Hamiltonian associated with the Lax system \eqref{eq:lax} is given by
\begin{align}\label{eq:Hamiltonian-via-RH}
H(r;\vec{x},\vec{u},\vr_1,\vr_2,\vs_1,\vs_2,\tau):
&=-\sum_{j=1}^{m} x_j\mathfrak{s}_{j}\cdot {\rm Tr}\left(
\left(X_{R,j}^{(1)}(r)-X_{L,j}^{(1)}(r)\right)
\begin{pmatrix}
0 & 0 & 1 & 1\\
0 & 0 & 1 & 1\\
0 & 0 & 0 & 0\\
0 & 0 & 0 & 0
\end{pmatrix}
\right)
\nonumber\\
&= -\sum_{j=1}^{m} x_{j}\mathfrak{s}_{j} \sum_{k=3}^4 \sum_{l=1}^2
\left(X_{R,j}^{(1)}(r)-X_{L,j}^{(1)}(r)\right)_{kl}.
\end{align}
Now we turn to the explicit formula of $H(r)$ in terms of the functions $\{p_{j,k}(r), q_{j,k}(r)\}_{j=1,\ldots,m}^{k=1,\ldots,4}$ and $p_5(r), q_5(r), p_6(r), q_6(r)$.
By reading the $\mathcal{O}(1)$ term in the expansion of $\partial_zX(z;r)=L(z;r)X(z;r)$ as $z\to rx_j$, using also \eqref{eq:X-near-rx_j}, \eqref{eq: Phi-expand-rx_j} and \eqref{eq:L}, 
we have 
\begin{multline}
	X_{R,j}^{(1)}(r) = \mathfrak{s}_j 
	\left[X_{R,j}^{(1)}(r), \begin{pmatrix}
	0 & 0 & 1 & 1\\
	0 & 0 & 1 & 1\\
	0 & 0 & 0 & 0\\
	0 & 0 & 0 & 0
	\end{pmatrix}\right]\\
	+X_{R,j}^{(0)}(r)^{-1}\left(rx_j\ii \left(\vr_1E_{31}-\vr_2E_{42}\right)+A_0(r)+\sum_{\substack{k=1 \\ k\neq j}}^{m}\left(\frac{A_k(r)}{rx_j-rx_k}+\frac{A_{-k}(r)}{rx_j+rx_k}\right) +\frac{A_{-j}(r)}{2rx_j}\right)X_{R,j}^{(0)}(r).
\end{multline}
Similarly, by looking for the $\mathcal{O}(1)$ term in the expansion of $\partial_zX(z;r)=L(z;r)X(z;r)$ as $z\to -rx_j$, we have
\begin{multline}
	X_{L,j}^{(1)}(r) = \mathfrak{s}_j 
	\left[X_{L,j}^{(1)}(r), \begin{pmatrix}
	0 & 0 & 1 & 1\\
	0 & 0 & 1 & 1\\
	0 & 0 & 0 & 0\\
	0 & 0 & 0 & 0
	\end{pmatrix}\right]\\
	+X_{L,j}^{(0)}(r)^{-1}\left(-rx_j\ii \left(\vr_1E_{31}-\vr_2E_{42}\right)+A_0(r)+\sum_{\substack{k=1 \\ k\neq j}}^{m}\left(\frac{A_k(r)}{-rx_j-rx_k}+\frac{A_{-k}(r)}{-rx_j+rx_k}\right) +\frac{A_j(r)}{-2rx_j}\right)X_{L,j}^{(0)}(r).
\end{multline}
Substituting the above two equalities into \eqref{eq:Hamiltonian-via-RH} yields
\begin{align*}
	&H(r)=\\
	&-\sum_{j=1}^{m} x_j\mathfrak{s}_j 
	\begin{pmatrix}
	0 & 0 & 1 & 1\\
	\end{pmatrix}\times \\
	&\left[X_{R,j}^{(0)}(r)^{-1}\left(rx_j\ii \left(\vr_1E_{31}-\vr_2E_{42}\right)+A_0(r)+\sum_{\substack{k=1 \\ k\neq j}}^{m}\left(\frac{A_k(r)}{rx_j-rx_k}+\frac{A_{-k}(r)}{rx_j+rx_k}\right) +\frac{A_{-j}(r)}{2rx_j}\right)X_{R,j}^{(0)}(r)\right.\\ 
	&\left. + X_{L,j}^{(0)}(r)^{-1}\left(rx_j\ii \left(\vr_1E_{31}-\vr_2E_{42}\right)-A_0(r)-\sum_{\substack{k=1 \\ k\neq j}}^{m}\left(\frac{A_k(r)}{-rx_j-rx_k}+\frac{A_{-k}(r)}{-rx_j+rx_k}\right) +\frac{A_j(r)}{2rx_j}\right)X_{L,j}^{(0)}(r)\right]
	\begin{pmatrix}
		1 \\
		1 \\
		0 \\
		0
	\end{pmatrix}.
\end{align*}
Recalling \eqref{def:qkpk} and \eqref{eq:symmXLR}, the above formula can be rewritten 
as 
\begin{align}\label{eq:Hamiltonian-via-RH2}
& H(r):=\nonumber\\
& \sum_{j=1}^{m}x_j\left\{\begin{pmatrix}
p_{j,1}(r) \\ p_{j,2}(r) \\ p_{j,3}(r) \\ p_{j,4}(r)
\end{pmatrix}^{\rm T}\left(rx_j\ii \left(\vr_1E_{31}-\vr_2E_{42}\right)+A_0(r)+\sum_{\substack{k=1 \\ k\neq j}}^{m}\left(\frac{A_k(r)}{rx_j-rx_k}+\frac{A_{-k}(r)}{rx_j+rx_k}\right) +\frac{A_{-j}(r)}{2rx_j}\right)
\begin{pmatrix}
q_{j,1}(r)\\q_{j,2}(r)\\q_{j,3}(r)\\q_{j,4}(r)
\end{pmatrix}\nonumber \right.\\ 
&\left.+\begin{pmatrix}
\widetilde p_{j,2}(r) \\ \widetilde p_{j,1}(r) \\ -\widetilde p_{j,4}(r) \\ -\widetilde p_{j,3}(r)
\end{pmatrix}^{\rm T} \left(rx_j\ii \left(\vr_1E_{31}-\vr_2E_{42}\right)-A_0(r)-\sum_{\substack{k=1 \\ k\neq j}}^{m}\left(\frac{A_k(r)}{-rx_j-rx_k}+\frac{A_{-k}(r)}{-rx_j+rx_k}\right) +\frac{A_j(r)}{2rx_j}\right) \begin{pmatrix}
\widetilde q_{j,2}(r)\\ \widetilde q_{j,1}(r)\\ -\widetilde q_{j,4}(r)\\ -\widetilde q_{j,3}(r)
\end{pmatrix}\right\}.
\end{align}
On account of \eqref{def:A0}--\eqref{def:A-j(r)}, the above formula can be further simplified to the one given in \eqref{eq:Hamiltonian}.

For later use, we state the following differential identities for $H(r)$.
\begin{proposition}\label{prop:Hamiltonian-diff-identity}
Along a solution of the Hamiltonian system \eqref{eq:Hamiltonian-derivative-1}--\eqref{eq:Hamiltonian-derivative-2}, the Hamiltonian $H(r)$ defined in \eqref{eq:Hamiltonian} satisfies the following differential identity:
\begin{multline}
\frac{\dif}{\dif r}H(r) =
\sum_{j=1}^{m}x_j^2\Big(
\ii\vr_1 q_{j,1}(r)p_{j,3}(r)-\ii\vr_2 q_{j,2}(r)p_{j,4}(r)
+\ii\vr_2 \widetilde q_{j,1}(r)\widetilde p_{j,3}(r)
-\ii\vr_1 \widetilde q_{j,2}(r)\widetilde p_{j,4}(r)
\Big)
\\-\frac{1}{r^2}\sum_{j=1}^{m}\Bigg[
x_j\sum_{\substack{k=1\\k\neq j}}^{m}\left(
\frac{S_{jk}(r)S_{kj}(r)+\widetilde S_{jk}(r)\widetilde S_{kj}(r)}{x_j-x_k}
+\frac{\Upsilon_{jk}(r)\widetilde\Upsilon_{kj}(r)+\Upsilon_{kj}(r)\widetilde\Upsilon_{jk}(r)}{x_j+x_k}
\right)
+\Upsilon_{jj}(r)\widetilde\Upsilon_{jj}(r)
\Bigg].
\end{multline}
As $\tau=0$, 
\begin{multline}\label{eq:Hamiltonian-diff-identity-tau0}
\sum_{k=5}^{6}\left(p_k(r)q_k'(r)+\widetilde p_k(r)\widetilde q_k'(r)\right)
+\sum_{j=1}^{m}\sum_{k=1}^{4}\left(p_{j,k}(r)q_{j,k}'(r)+\widetilde p_{j,k}(r)\widetilde q_{j,k}'(r)\right)-H(r)\\
=H(r)-\frac13\frac{\dif}{\dif r}\Bigg(
2rH(r)
+\sum_{j=1}^{m}\Big(
p_{j,1}(r)q_{j,1}(r)+p_{j,2}(r)q_{j,2}(r)
+\widetilde p_{j,1}(r)\widetilde q_{j,1}(r)
+\widetilde p_{j,2}(r)\widetilde q_{j,2}(r)
\Big)\\
-2p_5(r)q_5(r)-2\widetilde p_5(r)\widetilde q_5(r)
-p_6(r)q_6(r)-\widetilde p_6(r)\widetilde q_6(r)
+2\frac{\vs_1}{\vr_1}p_5(r)+2\frac{\vs_2}{\vr_2}\widetilde p_5(r)
\Bigg).
\end{multline}
Furthermore, for $\gamma$ being any parameter, we have
\begin{multline}\label{eq:Hamiltonian-diff-identity-parameter}
\partial_\gamma \left(\sum_{k=5}^{6}\left(p_k(r)q_k'(r)+\widetilde p_k(r)\widetilde q_k'(r)\right)
+\sum_{j=1}^{m}\sum_{k=1}^{4}\left(p_{j,k}(r)q_{j,k}'(r)+\widetilde p_{j,k}(r)\widetilde q_{j,k}'(r)\right)-H(r)\right) \\
=\frac{\dif}{\dif r}
\left(\sum_{k=5}^{6}\left(p_k(r)\partial_{\gamma}q_k(r)+\widetilde p_k(r)\partial_{\gamma}\widetilde q_k(r)\right)
+\sum_{j=1}^{m}\sum_{k=1}^{4}\left(p_{j,k}(r)\partial_{\gamma}q_{j,k}(r)+\widetilde p_{j,k}(r)\partial_{\gamma}\widetilde q_{j,k}(r)\right)\right).
\end{multline}
\end{proposition}

\begin{proof}
Denote
\begin{align*}
\mathcal I(r):=
\sum_{j=1}^{m}\Bigg[
x_j\sum_{\substack{k=1\\k\neq j}}^{m}\left(
\frac{S_{jk}(r)S_{kj}(r)+\widetilde S_{jk}(r)\widetilde S_{kj}(r)}{x_j-x_k}
+\frac{\Upsilon_{jk}(r)\widetilde\Upsilon_{kj}(r)+\Upsilon_{kj}(r)\widetilde\Upsilon_{jk}(r)}{x_j+x_k}
\right)
+\Upsilon_{jj}(r)\widetilde\Upsilon_{jj}(r)
\Bigg],
\end{align*}
and
\begin{align*}
\mathcal F_j(r):=
\ii\vr_1 q_{j,1}(r)p_{j,3}(r)-\ii\vr_2 q_{j,2}(r)p_{j,4}(r)
+\ii\vr_2 \widetilde q_{j,1}(r)\widetilde p_{j,3}(r)
-\ii\vr_1 \widetilde q_{j,2}(r)\widetilde p_{j,4}(r).
\end{align*}
Then, by \eqref{eq:Hamiltonian}, $H(r)$ can be written in the form
\begin{equation*}
H(r)=\sum_{j=1}^{m} r x_j^2 \mathcal F_j(r)+(\textrm{terms with no explicit $r$-dependence})
+\frac{\mathcal I(r)}{r}.
\end{equation*}
By using \eqref{eq:Hamiltonian-derivative-1}--\eqref{eq:Hamiltonian-derivative-2}, it is seen that
\begin{align*}
\frac{\dif}{\dif r}H(r)
&=\frac{\partial H}{\partial r}
+\sum_{j=5}^{6}\left(
\frac{\partial H}{\partial q_j}q_j'(r)
+\frac{\partial H}{\partial p_j}p_j'(r)
\right)
+\sum_{j=1}^{m}\sum_{k=1}^{4}\left(
\frac{\partial H}{\partial q_{j,k}}q_{j,k}'(r)
+\frac{\partial H}{\partial p_{j,k}}p_{j,k}'(r)
\right) \\
&=\frac{\partial H}{\partial r}
+\sum_{j=5}^{6}\left(
\frac{\partial H}{\partial q_j}\frac{\partial H}{\partial p_j}
-\frac{\partial H}{\partial p_j}\frac{\partial H}{\partial q_j}
\right)
+\sum_{j=1}^{m}\sum_{k=1}^{4}\left(
\frac{\partial H}{\partial q_{j,k}}\frac{\partial H}{\partial p_{j,k}}
-\frac{\partial H}{\partial p_{j,k}}\frac{\partial H}{\partial q_{j,k}}
\right) \\
&=\frac{\partial H}{\partial r}.
\end{align*}
Therefore it remains to compute only the explicit $r$-dependence of $H$. Here the remaining terms may still depend on $r$ through the canonical variables $q_{j,k}(r)$, $p_{j,k}(r)$, $q_j(r)$, and $p_j(r)$, but they contain no explicit factor of $r$. Thus the explicit $r$-dependence of \eqref{eq:Hamiltonian} appears only in the terms
$\sum_{j=1}^{m} r x_j^2\mathcal F_j(r)$ and $\mathcal I(r)/r$. Hence
\begin{align*}
\frac{\partial H}{\partial r}
=\sum_{j=1}^{m}x_j^2\mathcal F_j(r)-\frac{\mathcal I(r)}{r^2},
\end{align*}
which gives the desired formula.
Formula \eqref{eq:Hamiltonian-diff-identity-tau0} follows directly from \eqref{eq:coupled system} and \eqref{eq:extra-condition-1}.

To see the last differential identity, it follows from the canonical relations 
\eqref{eq:Hamiltonian-derivative-1}--\eqref{eq:Hamiltonian-derivative-2} that
\begin{align*}
&\partial_\gamma H(r)=\partial_\gamma \left(\sum_{k=5}^{6}\left(\frac{\partial H}{\partial p_k}\frac{\partial p_{k}(r)}{\partial \gamma}+\frac{\partial H}{\partial q_k}\frac{\partial q_{k}(r)}{\partial \gamma}
+\frac{\partial H}{\partial \widetilde p_k}\frac{\partial \widetilde p_{k}(r)}{\partial \gamma}+\frac{\partial H}{\partial \widetilde q_k}\frac{\partial \widetilde q_{k}(r)}{\partial \gamma}\right)
\right.\\ 
&\left. \hspace*{2.5em}+\sum_{j=1}^{m}\sum_{k=1}^{4}\left(\frac{\partial H}{\partial p_{j,k}}\frac{\partial p_{j,k}(r)}{\partial \gamma}+\frac{\partial H}{\partial q_{j,k}}\frac{\partial q_{j,k}(r)}{\partial \gamma}
+\frac{\partial H}{\partial \widetilde p_{j,k}}\frac{\partial \widetilde p_{j,k}(r)}{\partial \gamma}+\frac{\partial H}{\partial \widetilde q_{j,k}}\frac{\partial \widetilde q_{j,k}(r)}{\partial \gamma}\right)\right)\\
&=\sum_{k=5}^{6}\left(q_{k}'(r)\frac{\partial p_{k}(r)}{\partial \gamma}-p_k'(r)\frac{\partial q_{k}(r)}{\partial \gamma}
+\widetilde q_k'(r)\frac{\partial \widetilde p_{k}(r)}{\partial \gamma}-\widetilde p_k'(r)\frac{\partial \widetilde q_{k}(r)}{\partial \gamma}\right)
\\
&\hspace*{2.5em}+\sum_{j=1}^{m}\sum_{k=1}^{4}\left(q_{j,k}'(r)\frac{\partial p_{j,k}(r)}{\partial \gamma}-p_{j,k}'(r)\frac{\partial q_{j,k}(r)}{\partial \gamma}
+\widetilde q_{j,k}'(r)\frac{\partial \widetilde p_{j,k}(r)}{\partial \gamma}-\widetilde p_{j,k}'(r)\frac{\partial \widetilde q_{j,k}(r)}{\partial \gamma}\right),
\end{align*}
which leads to \eqref{eq:Hamiltonian-diff-identity-parameter}.

This completes the proof.
\end{proof}

\section{Asymptotic analysis of the RH problem for $X$ as $r\to+\infty$}\label{sec: asymptotic analysis of RH large r}
In this section, we perform the Deift-Zhou steepest descent analysis to the RH problem for $X$ as $r\to+\infty$.

\subsection{First transformation $X\to T$}
Define 
\begin{align}\label{def:X to T}
T(z) &= \diag \left( r^{\frac14},r^{\frac14}, r^{-\frac14},r^{-\frac14} \right) X(rz)
\nonumber
\\
& \quad \times 
\diag \left(
e^{\theta_1(rz)-\tau r z}, e^{\theta_2(rz)+ \tau r z}, e^{-\theta_1(rz)-\tau
r z},e^{-\theta_2(rz) + \tau r z} \right),
\end{align}
where the functions $\theta_1$ and $\theta_2$ are given in \eqref{def:theta1} and \eqref{def:theta2}, respectively.
On account of the RH problem for $X$, it is readily seen that $T$ satisfies the following RH problem.

\begin{paragraph}{RH problem for $T$}
	\begin{itemize}
		\item[(a)] $T(z)$ is analytic in $\mathbb{C}\setminus \Gamma_T$, where 
		$\Gamma_T:=\cup_{j=0}^{5}\Gamma_{j}^{(1)}\cup[-x_m,x_m]$ with $\Gamma_j^{(1)}$, $j=0,1,\ldots,5$, being the contours defined in \eqref{def:Gammajr} with $r=1$.
		\item[(b)] For $z\in \Gamma_T$, we have $T_+(z)=T_-(z)J_T(z)$, where
		{\small
		\begin{equation}\label{def:Jump-T}
		J_{T}(z):=\left\{
 		\begin{array}{ll}
          \begin{pmatrix}0&0&1&0\\0&1&0&0\\-1&0&0&0\\0&0&0&1 \end{pmatrix}, & \qquad \hbox{$z\in \Gamma_0^{(1)}$,} \\
          \begin{pmatrix}
			1 & 0 & 0 & 0\\
			-e^{\theta_1(rz)-\theta_2(rz)-2\tau r z} & 1 & 0 & 0\\
			e^{2\theta_1(rz)} & 0 & 1 & e^{\theta_1(rz)-\theta_2(rz)+2\tau rz}\\
			0 & 0 & 0 & 1
		  \end{pmatrix},  & \qquad  \hbox{$z\in \Gamma_1^{(1)}$,} \\
          \begin{pmatrix}
		1 & -e^{-\theta_1(rz)+\theta_2(rz)+2\tau r z} & 0 & 0 \\
		0 & 1 & 0 & 0 \\
		0 & 0 & 1 & 0 \\
		0 &  e^{2\theta_2(rz)} & e^{-\theta_1(rz)+\theta_2(rz)-2\tau r z} & 1
	\end{pmatrix},  & \qquad \hbox{$z\in \Gamma_2^{(1)}$,} \\
          \begin{pmatrix}1&0&0&0\\0&0&0&1\\0&0&1&0\\0&-1&0&0 \end{pmatrix}, & \qquad  \hbox{$z\in \Gamma_3^{(1)}$,} \\
          \begin{pmatrix}
		1 & e^{-\theta_1(rz)+\theta_2(rz)+2\tau r z} & 0 & 0 \\
		0 & 1 & 0 & 0 \\
		0 & 0 & 1 & 0 \\
		0 &  e^{2\theta_2(rz)} & -e^{-\theta_1(rz)+\theta_2(rz)-2\tau r z} & 1
	\end{pmatrix}, & \qquad  \hbox{$z\in \Gamma_4^{(1)}$,} \\
          \begin{pmatrix}
		1 & 0 & 0 & 0\\
		e^{\theta_1(rz)-\theta_2(rz)-2\tau r z} & 1 & 0 & 0\\
		e^{2\theta_1(rz)} & 0 & 1 & -e^{\theta_1(rz)-\theta_2(rz)+2\tau rz}\\
		0 & 0 & 0 & 1
	\end{pmatrix}, & \qquad  \hbox{$z\in \Gamma_5^{(1)}$,} \\
         \begin{pmatrix}
          1&0& s_je^{-2\theta_1(rz)}& s_j e^{-\theta_{1}(rz)-\theta_{2,+}(rz)+2\tau rz}
          \\0&e^{\theta_{2,+}(rz)-\theta_{2,-}(rz)}& s_je^{-\theta_1(rz)-\theta_{2,-}(rz)-2\tau rz}& s_j
          \\0&0&1&0
          \\0&0&0&e^{\theta_{2,-}(rz)-\theta_{2,+}(rz)}
  		\end{pmatrix}, & \qquad  \hbox{$z \in (-x_{j},-x_{j-1})$,}
          \\
         \begin{pmatrix}e^{\theta_{1,+}(rz)-\theta_{1,-}(rz)}&0& s_j & s_j e^{-\theta_{1,-}(rz)-\theta_2(rz)+2\tau rz}
          \\0&1&s_j e^{-\theta_{1,+}(rz)-\theta_2(rz)-2\tau rz}&s_j e^{-2\theta_2(rz)}\\0&0&e^{\theta_{1,-}(rz)-\theta_{1,+}(rz)}&0\\0&0&0&1
          \end{pmatrix}, & \qquad  \hbox{$z \in (x_{j-1},x_j)$.}
        \end{array}
      \right.
 	  \end{equation}
	}
Here $j=1,2,\ldots,m$, $x_0:=0$, and we have used the facts that
\begin{equation}\label{eq:thetairelations}
\begin{aligned}
\theta_{1,+}(rz)+\theta_{1,-}(rz)& = 0, \qquad z>0,
\\
\theta_{2,+}(rz)+\theta_{2,-}(rz)&=0, \qquad z<0.
\end{aligned}
\end{equation}

\item[\rm (c)] As $z \to \infty$ with $z\in \mathbb{C} \setminus \Gamma_{T}$, we have
\begin{align}\label{eq:asy T at infinity}
T(z)=\left( I+ \frac{T_1}{z} + \Boh(z^{-2}) \right) \diag \left((-z)^{-\frac14},z^{-\frac14},(-z)^{\frac14},z^{\frac14} \right)A,
\end{align}
where $T_1$ is independent of $z$, and $A$ is defined in \eqref{def:A}.

\item[\rm (d)]As $z \to \pm x_j$, we have $T(z)=\Boh(\log(z \mp x_j))$.
\end{itemize}
\end{paragraph}

\subsection{Second transformation $T\to S$}
The second transformation involves the standard opening of lenses around those intervals $(-x_j,-x_{j-1})$ and $(x_{j-1},x_j)$, $j=1,\ldots,m$.
For each $j=1,\ldots,m$, let $\gamma_{j,+}$ and $\gamma_{j,-}$ denote the upper and lower lips of the lens around $(x_{j-1},x_j)$, oriented from $x_{j-1}$ to $x_j$. We then define
$\gamma_{-j,+}:=-\gamma_{j,-}$ and $\gamma_{-j,-}:=-\gamma_{j,+}$,
which are the upper and lower lips of the lens around $(-x_j,-x_{j-1})$, oriented from $-x_j$ to $-x_{j-1}$; see Figure \ref{fig:jumps-S} for an illustration. 

Let
\begin{equation}\label{def:gamma_m+1}
\gamma_{m+1,+}:=\Gamma_1^{(1)}, \qquad \gamma_{m+1,-}:=\Gamma_5^{(1)}, \qquad \gamma_{-m-1,+}:=\Gamma_2^{(1)}, \qquad \gamma_{-m-1,-}:=\Gamma_4^{(1)},
\end{equation}
where $x_{m+1}:=\infty$ and $s_{m+1}:=1$. 
For $j=1,\ldots,m$, let $\Omega_{\pm j,+}$ and $\Omega_{\pm j,-}$ denote the regions enclosed by $\gamma_{\pm j,+}$ and $\gamma_{\pm j,-}$ together with the corresponding intervals. To match the notation in \eqref{def:gamma_m+1}, we also denote by $\Omega_{m+1,+}$ and $\Omega_{m+1,-}$ the regions enclosed by $\gamma_{m+1,+}$ and $\gamma_{m+1,-}$ together with $\Gamma_0^{(1)}$, and by $\Omega_{-m-1,+}$ and $\Omega_{-m-1,-}$ the regions enclosed by $\gamma_{-m-1,+}$ and $\gamma_{-m-1,-}$ together with $\Gamma_3^{(1)}$.

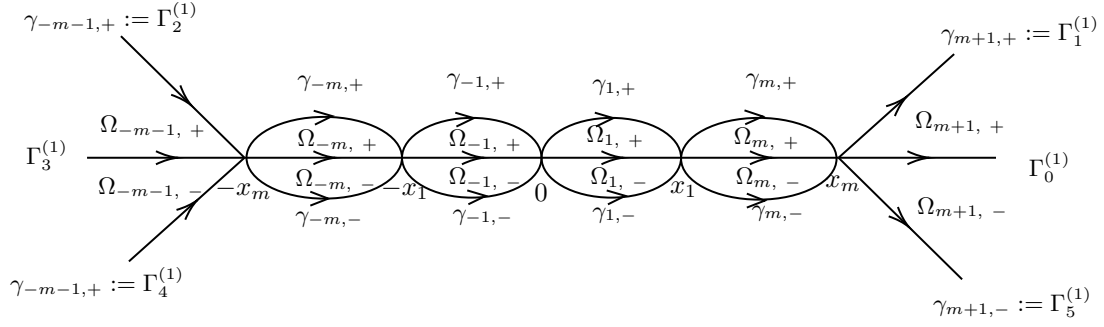
\begin{figure}[t]
\centering
\tikzset{every picture/.style={line width=0.75pt}} 
\begin{tikzpicture}[x=0.75pt,y=0.75pt,yscale=-1,xscale=1]
\draw    (161,131) -- (240,131) ;
\draw [shift={(206.5,131)}, rotate = 180] [color={rgb, 255:red, 0; green, 0; blue, 0 }  ][line width=0.75]    (10.93,-3.29) .. controls (6.95,-1.4) and (3.31,-0.3) .. (0,0) .. controls (3.31,0.3) and (6.95,1.4) .. (10.93,3.29)   ;
\draw    (380,131) -- (459,131) ;
\draw [shift={(425.5,131)}, rotate = 180] [color={rgb, 255:red, 0; green, 0; blue, 0 }  ][line width=0.75]    (10.93,-3.29) .. controls (6.95,-1.4) and (3.31,-0.3) .. (0,0) .. controls (3.31,0.3) and (6.95,1.4) .. (10.93,3.29)   ;
\draw    (99,70) -- (161,131) ;
\draw [shift={(134.28,104.71)}, rotate = 224.53] [color={rgb, 255:red, 0; green, 0; blue, 0 }  ][line width=0.75]    (10.93,-3.29) .. controls (6.95,-1.4) and (3.31,-0.3) .. (0,0) .. controls (3.31,0.3) and (6.95,1.4) .. (10.93,3.29)   ;
\draw    (459,131) -- (521,192) ;
\draw [shift={(494.28,165.71)}, rotate = 224.53] [color={rgb, 255:red, 0; green, 0; blue, 0 }  ][line width=0.75]    (10.93,-3.29) .. controls (6.95,-1.4) and (3.31,-0.3) .. (0,0) .. controls (3.31,0.3) and (6.95,1.4) .. (10.93,3.29)   ;
\draw    (161,131) -- (103,183) ;
\draw [shift={(137.21,152.33)}, rotate = 138.12] [color={rgb, 255:red, 0; green, 0; blue, 0 }  ][line width=0.75]    (10.93,-3.29) .. controls (6.95,-1.4) and (3.31,-0.3) .. (0,0) .. controls (3.31,0.3) and (6.95,1.4) .. (10.93,3.29)   ;
\draw    (517,79) -- (459,131) ;
\draw [shift={(493.21,100.33)}, rotate = 138.12] [color={rgb, 255:red, 0; green, 0; blue, 0 }  ][line width=0.75]    (10.93,-3.29) .. controls (6.95,-1.4) and (3.31,-0.3) .. (0,0) .. controls (3.31,0.3) and (6.95,1.4) .. (10.93,3.29)   ;
\draw   (240,131) .. controls (240,119.95) and (255.67,111) .. (275,111) .. controls (294.33,111) and (310,119.95) .. (310,131) .. controls (310,142.05) and (294.33,151) .. (275,151) .. controls (255.67,151) and (240,142.05) .. (240,131) -- cycle ;
\draw   (310,131) .. controls (310,119.95) and (325.67,111) .. (345,111) .. controls (364.33,111) and (380,119.95) .. (380,131) .. controls (380,142.05) and (364.33,151) .. (345,151) .. controls (325.67,151) and (310,142.05) .. (310,131) -- cycle ;
\draw   (266.73,108.38) -- (275.99,111.13) -- (267.3,115.35) ;
\draw   (272.73,147.38) -- (281.99,150.13) -- (273.3,154.35) ;
\draw   (336.73,108.38) -- (345.99,111.13) -- (337.3,115.35) ;
\draw   (342.71,147.41) -- (351.98,150.1) -- (343.32,154.38) ;
\draw    (240,131) -- (310,131) ;
\draw [shift={(281,131)}, rotate = 180] [color={rgb, 255:red, 0; green, 0; blue, 0 }  ][line width=0.75]    (10.93,-3.29) .. controls (6.95,-1.4) and (3.31,-0.3) .. (0,0) .. controls (3.31,0.3) and (6.95,1.4) .. (10.93,3.29)   ;
\draw    (310,131) -- (380,131) ;
\draw [shift={(351,131)}, rotate = 180] [color={rgb, 255:red, 0; green, 0; blue, 0 }  ][line width=0.75]    (10.93,-3.29) .. controls (6.95,-1.4) and (3.31,-0.3) .. (0,0) .. controls (3.31,0.3) and (6.95,1.4) .. (10.93,3.29)   ;
\draw   (380,131.5) .. controls (380,120.18) and (397.46,111) .. (419,111) .. controls (440.54,111) and (458,120.18) .. (458,131.5) .. controls (458,142.82) and (440.54,152) .. (419,152) .. controls (397.46,152) and (380,142.82) .. (380,131.5) -- cycle ;
\draw   (413,149) -- (422,152.5) -- (413,156) ;
\draw   (411.83,107.23) -- (420.99,110.28) -- (412.18,114.22) ;
\draw   (162,131) .. controls (162,119.68) and (179.46,110.5) .. (201,110.5) .. controls (222.54,110.5) and (240,119.68) .. (240,131) .. controls (240,142.32) and (222.54,151.5) .. (201,151.5) .. controls (179.46,151.5) and (162,142.32) .. (162,131) -- cycle ;
\draw   (193.99,148.01) -- (203,151.49) -- (194.01,155.01) ;
\draw   (195.71,107.41) -- (204.98,110.1) -- (196.32,114.38) ;
\draw    (459,131) -- (538,131) ;
\draw [shift={(504.5,131)}, rotate = 180] [color={rgb, 255:red, 0; green, 0; blue, 0 }  ][line width=0.75]    (10.93,-3.29) .. controls (6.95,-1.4) and (3.31,-0.3) .. (0,0) .. controls (3.31,0.3) and (6.95,1.4) .. (10.93,3.29)   ;
\draw    (82,131) -- (161,131) ;
\draw [shift={(127.5,131)}, rotate = 180] [color={rgb, 255:red, 0; green, 0; blue, 0 }  ][line width=0.75]    (10.93,-3.29) .. controls (6.95,-1.4) and (3.31,-0.3) .. (0,0) .. controls (3.31,0.3) and (6.95,1.4) .. (10.93,3.29)   ;

\draw (451,138.4) node [anchor=north west][inner sep=0.75pt]  [font=\small]  {$x_{m}$};
\draw (305,142.4) node [anchor=north west][inner sep=0.75pt]  [font=\small]  {$0$};
\draw (373,141.4) node [anchor=north west][inner sep=0.75pt]  [font=\small]  {$x_{1}$};
\draw (145,140.4) node [anchor=north west][inner sep=0.75pt]  [font=\small]  {$-x_{m}$};
\draw (228,140.4) node [anchor=north west][inner sep=0.75pt]  [font=\small]  {$-x_{1}$};
\draw (334,90.4) node [anchor=north west][inner sep=0.75pt]  [font=\footnotesize]  {$\gamma _{1,+}$};
\draw (334,154.4) node [anchor=north west][inner sep=0.75pt]  [font=\footnotesize]  {$\gamma _{1,-}$};
\draw (261,87.4) node [anchor=north west][inner sep=0.75pt]  [font=\footnotesize]  {$\gamma _{-1,+}$};
\draw (264,155.4) node [anchor=north west][inner sep=0.75pt]  [font=\footnotesize]  {$\gamma _{-1,-}$};
\draw (411,87.4) node [anchor=north west][inner sep=0.75pt]  [font=\footnotesize]  {$\gamma _{m,+}$};
\draw (414,153.4) node [anchor=north west][inner sep=0.75pt]  [font=\footnotesize]  {$\gamma _{m,-}$};
\draw (186,88.4) node [anchor=north west][inner sep=0.75pt]  [font=\footnotesize]  {$\gamma _{-m,+}$};
\draw (185,156.4) node [anchor=north west][inner sep=0.75pt]  [font=\footnotesize]  {$\gamma _{-m,-}$};
\draw (509,58.4) node [anchor=north west][inner sep=0.75pt]  [font=\footnotesize]  {$\gamma _{m+1,+} :=\Gamma _{1}^{( 1)}$};
\draw (553,126.4) node [anchor=north west][inner sep=0.75pt]  [font=\footnotesize]  {$\Gamma _{0}^{( 1)}$};
\draw (49,51.4) node [anchor=north west][inner sep=0.75pt]  [font=\footnotesize]  {$\gamma _{-m-1,+} :=\Gamma _{2}^{( 1)}$};
\draw (50,120.4) node [anchor=north west][inner sep=0.75pt]  [font=\footnotesize]  {$\Gamma _{3}^{( 1)}$};
\draw (42,184.4) node [anchor=north west][inner sep=0.75pt]  [font=\footnotesize]  {$\gamma _{-m-1,+} :=\Gamma _{4}^{( 1)}$};
\draw (505,194.4) node [anchor=north west][inner sep=0.75pt]  [font=\footnotesize]  {$\gamma _{m+1,-} :=\Gamma _{5}^{( 1)}$};
\draw (496,106.4) node [anchor=north west][inner sep=0.75pt]  [font=\footnotesize]  {$\Omega _{m+1,\ +}$};
\draw (497,149.4) node [anchor=north west][inner sep=0.75pt]  [font=\footnotesize]  {$\Omega _{m+1,\ -}$};
\draw (87,108.4) node [anchor=north west][inner sep=0.75pt]  [font=\footnotesize]  {$\Omega _{-m-1,\ +}$};
\draw (86,139.4) node [anchor=north west][inner sep=0.75pt]  [font=\footnotesize]  {$\Omega _{-m-1,\ -}$};
\draw (405,114.4) node [anchor=north west][inner sep=0.75pt]  [font=\footnotesize]  {$\Omega _{m,\ +}$};
\draw (406,134.4) node [anchor=north west][inner sep=0.75pt]  [font=\footnotesize]  {$\Omega _{m,\ -}$};
\draw (331,113.4) node [anchor=north west][inner sep=0.75pt]  [font=\footnotesize]  {$\Omega _{1,\ +}$};
\draw (333,134.4) node [anchor=north west][inner sep=0.75pt]  [font=\footnotesize]  {$\Omega _{1,\ -}$};
\draw (262,115.4) node [anchor=north west][inner sep=0.75pt]  [font=\footnotesize]  {$\Omega _{-1,\ +}$};
\draw (262,134.4) node [anchor=north west][inner sep=0.75pt]  [font=\footnotesize]  {$\Omega _{-1,\ -}$};
\draw (186,115.4) node [anchor=north west][inner sep=0.75pt]  [font=\footnotesize]  {$\Omega _{-m,\ +}$};
\draw (185,135.4) node [anchor=north west][inner sep=0.75pt]  [font=\footnotesize]  {$\Omega _{-m,\ -}$};
\end{tikzpicture}\caption{The jump contours for the RH problem for $S$ with $m=2$.}
\label{fig:jumps-S}
\end{figure}

For $j=1,\ldots,m, m+1$, now we define the following matrices:
\begin{align}
	&J_{\gamma_{j,+}}(z):=
	\begin{pmatrix}
		1 & 0 & 0 & 0\\
		-e^{\theta_1(rz)-\theta_2(rz)-2\tau r z} & 1 & 0 & 0\\
		\frac{e^{2\theta_1(rz)}}{s_j} & 0 & 1 & e^{\theta_1(rz)-\theta_2(rz)+2\tau rz}\\
		0 & 0 & 0 & 1
	\end{pmatrix},\\
	&J_{\gamma_{j,-}}(z):=
	\begin{pmatrix}
		1 & 0 & 0 & 0\\
		e^{\theta_1(rz)-\theta_2(rz)-2\tau r z} & 1 & 0 & 0\\
		\frac{e^{2\theta_1(rz)}}{s_j} & 0 & 1 & -e^{\theta_1(rz)-\theta_2(rz)+2\tau rz}\\
		0 & 0 & 0 & 1
	\end{pmatrix},\\
	& J_{\gamma_{-j,+}}(z):=
	\begin{pmatrix}
		1 & -e^{-\theta_1(rz)+\theta_2(rz)+2\tau r z} & 0 & 0 \\
		0 & 1 & 0 & 0 \\
		0 & 0 & 1 & 0 \\
		0 &  \frac{e^{2\theta_2(rz)}}{s_j} & e^{-\theta_1(rz)+\theta_2(rz)-2\tau r z} & 1
	\end{pmatrix},\\
	& J_{\gamma_{-j,-}}(z):=
	\begin{pmatrix}
		1 & e^{-\theta_1(rz)+\theta_2(rz)+2\tau r z} & 0 & 0 \\
		0 & 1 & 0 & 0 \\
		0 & 0 & 1 & 0 \\
		0 &  \frac{e^{2\theta_2(rz)}}{s_j} & -e^{-\theta_1(rz)+\theta_2(rz)-2\tau r z} & 1
	\end{pmatrix}.
\end{align}
In what follows, on the intervals $(x_{j-1},x_j)$ and $(-x_j,-x_{j-1})$, the matrices $J_{\gamma_{j,\pm}}(z)$ and $J_{\gamma_{-j,\pm}}(z)$ are understood as 
taking the corresponding boundary values from the $\pm$-side.
An observation shows that the jump matrix $J_T(z)$ on the intervals $(-x_j,-x_{j-1})$ and $(x_{j-1},x_j)$, $j=1,\ldots,m$, can be factorized as
\begin{equation}
	\begin{aligned}
	&J_T(z)=J_{\gamma_{j,-}}(z)
	\begin{pmatrix}
		0 & 0 & s_j & 0 \\
		0 & 1 & 0 & 0 \\
		-\frac{1}{s_j} & 0 & 0 & 0\\
		0 & 0 & 0 & 1
	\end{pmatrix}
	J_{\gamma_{j,+}}(z), \qquad z\in (x_{j-1},x_j),\\
	&J_T(z)=J_{\gamma_{-j,-}}(z)
	\begin{pmatrix}
		1 & 0 & 0 & 0 \\
		0 & 0 & 0 & s_j \\
		0 & 0 & 1 & 0\\
		0 & -\frac{1}{s_j} & 0 & 0
	\end{pmatrix}
	J_{\gamma_{-j,+}}(z), \qquad z\in (-x_j,-x_{j-1}),
	\end{aligned}
\end{equation}
respectively. It is also noted that for $j=1,\ldots,m$, 
\begin{equation}
	J_{\gamma_{\pm j,\pm}}(z)^{-1} = 2I-J_{\gamma_{\pm j,\pm}}(z).
\end{equation}
Now we define
\begin{equation} \label{def:T to S}
	S(z)=T(z)\left\{
	\begin{array}{ll}
	J_{\gamma_{j,+}}(z)^{-1}, & \qquad \text{$z$ is inside the lens around $(x_{j-1},x_j)$ and $\im z>0$},\\
	J_{\gamma_{j,-}}(z), & \qquad \text{$z$ is inside the lens around $(x_{j-1},x_j)$ and $\im z<0$},\\
	J_{\gamma_{-j,+}}(z)^{-1}, & \qquad \text{$z$ is inside the lens around $(-x_j,-x_{j-1})$ and $\im z>0$},\\
	J_{\gamma_{-j,-}}(z), & \qquad \text{$z$ is inside the lens around $(-x_j,-x_{j-1})$ and $\im z<0$}, \\
	I, & \qquad \text{otherwise},
	\end{array}
	\right.
\end{equation}
where $j=1,\ldots,m$.

Then it is straightforward to verify that $S$ satisfies the following RH problem.
\begin{paragraph}{RH problem for $S$}
	\begin{itemize}
		\item[(a)] $S(z)$ is analytic in $\mathbb{C}\setminus \Gamma_S$, where 
		\begin{equation}\label{jump contour: Gamma_S}
			\Gamma_S:=\mathbb{R}\cup \bigcup_{j=1}^{m+1} (\gamma_{j,+} \cup \gamma_{j,-} \cup \gamma_{-j,+} \cup \gamma_{-j,-}).
		\end{equation}
		\item[(b)] For $z\in \Gamma_S\setminus \cup_{j=0}^{m}\{\pm x_j\}$, we have $S_+(z)=S_-(z)J_S(z)$, where
		\begin{equation}\label{eq:jump-S}
		J_S(z)=\begin{cases}
			J_{\gamma_{j,+}}(z), & z \in \gamma_{j,+}, \\
			J_{\gamma_{j,-}}(z), & z \in \gamma_{j,-}, \\
			J_{\gamma_{-j,+}}(z), & z \in \gamma_{-j,+}, \\
			J_{\gamma_{-j,-}}(z), & z \in \gamma_{-j,-}, \\
			\begin{pmatrix}
				0 & 0 & s_j & 0 \\
				0 & 1 & 0 & 0 \\
				-\frac{1}{s_j} & 0 & 0 & 0\\
				0 & 0 & 0 & 1
			\end{pmatrix}, & z \in (x_{j-1},x_j), \\
			\begin{pmatrix}
				1 & 0 & 0 & 0 \\
				0 & 0 & 0 & s_j \\
				0 & 0 & 1 & 0\\
				0 & -\frac{1}{s_j} & 0 & 0
			\end{pmatrix}, & z \in (-x_j,-x_{j-1}),
		\end{cases}
		\end{equation}
		where $j=1,\ldots,m+1$, and we recall that $s_{m+1}=1$ and $x_{m+1}=+\infty$.
		\item[\rm (c)] As $z \to \infty$ with $z\in \mathbb{C} \setminus \Gamma_{S}$, we have
		\begin{align}\label{eq:asy S at infinity}
			S(z)=\left( I+ \frac{T_1}{z} + \Boh(z^{-2}) \right) \diag \left((-z)^{-\frac14},z^{-\frac14},(-z)^{\frac14},z^{\frac14} \right)A,
		\end{align}
		where $T_1$ is the same as in \eqref{eq:asy T at infinity} and $A$ is defined in \eqref{def:A}.
		\item[\rm (d)]As $z \to \pm x_j$, we have $S(z)=\Boh(\log(z \mp x_j))$, $j=1,\ldots,m$.
	\end{itemize}
\end{paragraph}

\subsection{Global parametrix}
Recalling the definitions of $\theta_1(z)$ and $\theta_2(z)$ in \eqref{def:theta1} and \eqref{def:theta2}, it is readily
seen that $J_S(z)\to I$ as $r\to+\infty$ for $z\in\cup_{j=1}^{m+1} (\gamma_{j,+} \cup \gamma_{j,-} \cup \gamma_{-j,+} \cup \gamma_{-j,-})$ 
but bounded away from fixed neighborhoods of the points $\pm x_j$, $j=0, 1, \ldots, m$. 
This motivates us to consider a global parametrix $N$, which satisfies the same jump condition as $S$ on the real axis.

\begin{paragraph}{RH problem for $N$}
	\begin{itemize}
		\item [(a)] $N(z)$ is analytic in $\mathbb{C}\setminus \mathbb{R}$.
		\item [(b)] For $z\in \mathbb{R}$, we have $N_+(z)=N_-(z)J_N(z)$, where
		\begin{equation}\label{eq:jump-N}
		J_N(z)=\begin{cases}
			\begin{pmatrix}
				1 & 0 & 0 & 0 \\
				0 & 0 & 0 & s_j \\
				0 & 0 & 1 & 0\\
				0 & -\frac{1}{s_j} & 0 & 0
			\end{pmatrix}, & z \in (-x_j,-x_{j-1}),\ j=1,\ldots,m, m+1, \\
			\begin{pmatrix}
				0 & 0 & s_j & 0 \\
				0 & 1 & 0 & 0 \\
				-\frac{1}{s_j} & 0 & 0 & 0\\
				0 & 0 & 0 & 1
			\end{pmatrix}, & z \in (x_{j-1},x_j), \ j=1,\ldots,m, m+1,
		\end{cases}
	\end{equation}
	where we recall that $s_{m+1}=1$ and $x_{m+1}=+\infty$.
		\item[\rm (c)] As $z \to \infty$ with $z\in \mathbb{C} \setminus \mathbb{R}$, we have
		\begin{align}\label{eq:asy N at infinity}
			N(z)=\left( I+\Boh(z^{-1}) \right) \diag \left((-z)^{-\frac14},z^{-\frac14},(-z)^{\frac14},z^{\frac14} \right)A,
		\end{align}
		where $A$ is defined in \eqref{def:A}.
		\item[\rm (d)] As $z \to \pm x_j$, we have $N(z)=\Boh(1)$, $j=1,\ldots,m$.
	\end{itemize}
\end{paragraph}

To solve the above RH problem, we can split it into two $2\times 2$ block RH problems for $N^{(L)}$ and $N^{(R)}$ as below.
Such decomposition was used in \cite[Section 6.3]{YZ2024} for the case of $m=1$. However, the analysis are much more complicated in the present case of
general $m\in\mathbb{N}^{+}$. 

Let 
\begin{align}\label{def: NL and NR}
N_L(z):=
\begin{pmatrix}
N_{22}(z) & N_{24}(z) \\
N_{42}(z) & N_{44}(z)
\end{pmatrix}, 
\quad 
N_R(z):=
\begin{pmatrix}
N_{11}(z) & N_{13}(z) \\
N_{31}(z) & N_{33}(z)
\end{pmatrix}.
\end{align}
Then $N_L$ and $N_R$ satisfy the following RH problems.
\begin{paragraph}{RH problems for $N_L$ and $N_R$}
	\begin{itemize}
		\item [(La)] $N_L(z)$ is analytic in $\mathbb{C}\setminus (-\infty,0]$.
		\item [(Lb)] The jump condition for $N_L$ is given by
		\begin{equation}\label{eq:jump-NL}
		N_{L,+}(z)=N_{L,-}(z)
		\begin{pmatrix}
			0 & s_j \\
			-\frac{1}{s_j} & 0
		\end{pmatrix}, \quad z \in (-x_j,-x_{j-1}), \quad j=1,\ldots,m ,m+1.
		\end{equation}
		\item [(Lc)] As $z \to \infty$ with $z\in \mathbb{C}\setminus (-\infty,0]$, we have
		\begin{align}\label{eq:asy NL at infinity}
			N_L(z)=\left( I+\Boh(z^{-1}) \right) z^{-\frac14\sigma_3}\frac{I+\ii\sigma_1}{\sqrt{2}}.
		\end{align}
	\end{itemize}
	Correspondingly, 
	\begin{itemize}
		\item [(Ra)] $N_R(z)$ is analytic in $\mathbb{C}\setminus [0,+\infty)$.
		\item [(Rb)] The jump condition for $N_R$ is given by
		\begin{equation}\label{eq:jump-NR}
		N_{R,+}(z)=N_{R,-}(z)
		\begin{pmatrix}
			0 & s_j \\
			-\frac{1}{s_j} & 0
		\end{pmatrix}, \quad z \in (x_{j-1},x_j), \quad j=1,\ldots,m ,m+1.
		\end{equation}
		\item [(Rc)] As $z \to \infty$ with $z\in \mathbb{C}\setminus [0,+\infty)$, we have
		\begin{align}\label{eq:asy NR at infinity}
			N_R(z)=\left( I+\Boh(z^{-1}) \right) (-z)^{-\frac14\sigma_3}\frac{I-\ii\sigma_1}{\sqrt{2}}.
		\end{align}
	\end{itemize}
\end{paragraph}
A direct verification shows the relation
\begin{equation}\label{symmetry NL and NR}
	N_{R}(z)=\sigma_3 N_L(-z)\sigma_3, \qquad z\in \mathbb{C}\setminus [0,+\infty).
\end{equation}
Hence it is sufficient to solve the RH problem for $N_L$.
To this end, define two auxiliary functions
\begin{equation}\label{def: d1d2}
	d_1(z)=\lambda(z^{\frac12}), \quad {\rm and} \quad d_2(z)=\lambda(-z^{\frac12}),
\end{equation}
where 
\begin{equation}\label{def: lambda}
	\lambda(\zeta)=\prod_{j=1}^m \left( \frac{\zeta-i{x}_{j}^{\frac12}}{\zeta+i{x}_{j}^{\frac12}} \right)^{\beta_j}, \quad \zeta \in \mathbb{C}\setminus [-i{x}_{m}^{\frac12},i{x}_{m}^{\frac12}],
\end{equation}
and $\beta_j$ is given in \eqref{def: first def beta_j}, or equivalently,
\begin{equation}\label{def: second def beta_j}
	\beta_j=\frac{1}{2\pi \ii}\log \frac{s_j}{s_{j+1}}, \qquad j=1,\ldots,m.
\end{equation}
Here the branch cut is chosen such that $\lambda(\zeta)\to 1$ as $\zeta \to \infty$
with the orientation of the branch cut being from  $i{x}_{m}^{\frac12}$ to $-i{x}_{m}^{\frac12}$.
Then it is readily verified that $\lambda$ satisfies the jump condition
\begin{equation}\label{eq:jump lambda}
\lambda_{+}(\zeta)=\lambda_{-}(\zeta)s_j^{-1}, \qquad \zeta\in \ii ({x}_{j-1}^{\frac12},{x}_{j}^{\frac12})\cup \ii (-{x}_{j}^{\frac12},-{x}_{j-1}^{\frac12}), \quad j=1,\ldots,m.
\end{equation}
Necessary properties of $d_j(z)$, $j=1,2$ are collected in the following proposition, which can be verified by direct calculations.
\begin{proposition}\label{prop:d1d2}
	The functions $d_j(z)$, $j=1,2$ defined in \eqref{def: d1d2} satisfy the following properties.
	\begin{itemize}
		\item [\rm (i)] $d_j(z)$, $j=1,2$ are analytic in $\mathbb{C}\setminus (-\infty, 0]$. 
		Moreover, we have
		\begin{align}\label{eq:dijump-multi}
			d_{1,\pm}(z)&=s_j^{-1}d_{2,\mp}(z), & z\in (-x_j,-x_{j-1}),
		\end{align}
		where $j=1,\ldots,m, m+1$, and we recall that $s_{m+1}=1$ and $x_{m+1}=+\infty$.
		
		\item [\rm (ii)] As $z \to \infty$ with $z\in \mathbb{C}\setminus (-\infty,0]$, we have
		\begin{align}
			& d_1(z)=1-\frac{2\ii}{z^{\frac12}}\sum_{j=1}^{m}\beta_j{x}_{j}^{\frac12}-\frac{2}{z}\left(\sum_{j=1}^{m}\beta_j{x}_{j}^{\frac12}\right)^2+\Boh(z^{-\frac32}), \label{eq:asy d1 at infinity}\\
			& d_2(z)=1+\frac{2\ii}{z^{\frac12}}\sum_{j=1}^{m}\beta_j{x}_{j}^{\frac12}-\frac{2}{z}\left(\sum_{j=1}^{m}\beta_j{x}_{j}^{\frac12}\right)^2+\Boh(z^{-\frac32}). \label{eq:asy d2 at infinity}
		\end{align}

		\item [\rm (iii)] As $z\to -x_j$ with $\im z>0$, $j=1,\ldots,m$, we have
		\begin{align}
			d_1(z)&=d_{1,-x_j}^{(0)}(z+x_j)^{\beta_j}\left(1+d_{1,-x_j}^{(1)}(z+x_j)+\Boh((z+x_j)^2)\right), \\
			d_2(z)&=d_{2,-x_j}^{(0)}(z+x_j)^{-\beta_j}\left(1+d_{2,-x_j}^{(1)}(z+x_j)+\Boh((z+x_j)^2)\right), \label{eq:d2-endpoint-multi}
		\end{align}	
		where 
		\begin{align}
			d_{1,-x_j}^{(0)}&=e^{-\pi\ii\sum_{k=j}^{m}\beta_k}(4x_j)^{-\beta_j} 
			\prod_{k=1}^{j-1}\left(\frac{{x}_{j}^{\frac12}-{x}_{k}^{\frac12}}{{x}_{j}^{\frac12}+{x}_{k}^{\frac12}}\right)^{\beta_k}
			\prod_{k=j+1}^{m}\left(\frac{{x}_{k}^{\frac12}-{x}_{j}^{\frac12}}{{x}_{k}^{\frac12}+{x}_{j}^{\frac12}}\right)^{\beta_k},\label{def:d1-xj0}\\
			d_{1,-x_j}^{(1)}&=\frac{\beta_j}{2x_j}-\sum_{\substack{k=1\\k\neq j}}^{m}\frac{\beta_k{x}_{k}^{\frac12}}{{x}_{j}^{\frac12}(x_j-x_k)}, \label{def:d1-xj1}\\
			d_{2,-x_j}^{(0)}&=e^{\pi\ii\sum_{k=j}^{m}\beta_k}(4x_j)^{\beta_j}
			\prod_{k=1}^{j-1}\left(\frac{{x}_{j}^{\frac12}+{x}_{k}^{\frac12}}{{x}_{j}^{\frac12}-{x}_{k}^{\frac12}}\right)^{\beta_k}
			\prod_{k=j+1}^{m}\left(\frac{{x}_{k}^{\frac12}+{x}_{j}^{\frac12}}{{x}_{k}^{\frac12}-{x}_{j}^{\frac12}}\right)^{\beta_k},\label{def:d2-xj0}\\
			d_{2,-x_j}^{(1)}&=-\frac{\beta_j}{2x_j}
			+\sum_{\substack{k=1\\k\neq j}}^{m}\frac{\beta_k{x}_{k}^{\frac12}}{{x}_{j}^{\frac12}(x_j-x_k)}. \label{def:d2-xj1}
		\end{align}
	
		\item [\rm (iv)] As $z\to x_j$, $j=1,\ldots,m$, we have
		\begin{align}
			d_1(z)&=d_{1,x_j}^{(0)}\left(1+d_{1,x_j}^{(1)}(z-x_j)+\Boh((z-x_j)^2)\right), \\
			d_2(z)&=d_{2,x_j}^{(0)}\left(1+d_{2,x_j}^{(1)}(z-x_j)+\Boh((z-x_j)^2)\right), \label{eq:d2-plus-endpoint-multi}
		\end{align}
		where
		\begin{align}
			d_{1,x_j}^{(0)}&=\prod_{k=1}^{m}\left(\frac{x_j-x_k-2\ii({x}_{j}x_k)^{\frac12}}{x_j+x_k}\right)^{\beta_k}=\exp\left(-2\ii\sum_{k=1}^{m}\beta_k\arctan\sqrt{\frac{x_k}{x_j}}\right), \label{def:d1+xj0}\\
			d_{1,x_j}^{(1)}&=\frac{\ii}{{x}_{j}^{\frac12}}\sum_{k=1}^{m}\frac{\beta_k{x}_{k}^{\frac12}}{x_j+x_k}, \label{def:d1+xj1}\\			
			d_{2,x_j}^{(0)}&=\prod_{k=1}^{m}\left(\frac{x_j-x_k+2\ii({x}_{j}x_k)^{\frac12}}{x_j+x_k}\right)^{\beta_k}=\exp\left(2\ii\sum_{k=1}^{m}\beta_k\arctan\sqrt{\frac{x_k}{x_j}}\right), \label{def:d2+xj0}\\
			d_{2,x_j}^{(1)}&=-\frac{\ii}{{x}_{j}^{\frac12}}\sum_{k=1}^{m}\frac{\beta_k{x}_{k}^{\frac12}}{x_j+x_k}. \label{def:d2+xj1}
		\end{align}

		\item [\rm (v)] As $z \to 0$, we have
		\begin{align}
			d_1(z)&=d_{1,0}^{(0)}\left(1+d_{1,0}^{(1)}z^{\frac12}+d_{1,0}^{(2)}z+d_{1,0}^{(3)}z^{\frac32}+\Boh(z^2)\right), \\
			d_2(z)&=d_{2,0}^{(0)}\left(1+d_{2,0}^{(1)}z^{\frac12}+d_{2,0}^{(2)}z+d_{2,0}^{(3)}z^{\frac32}+\Boh(z^2)\right), \label{eq:d2-zero-multi}
		\end{align}
		where
		\begin{align}
			& d_{1,0}^{(0)}=e^{-\pi\ii\sum_{j=1}^{m}\beta_j}=s_1^{-\frac12}, \quad  d_{1,0}^{(1)}=2\ii\sum_{j=1}^{m}\beta_j x_{j}^{-\frac12},\\
			& d_{1,0}^{(2)}=-2\Bigg(\sum_{j=1}^{m}\beta_j x_{j}^{-\frac12}\Bigg)^2, \quad d_{1,0}^{(3)}=-\frac{2\ii}{3}\Bigg(\sum_{j=1}^{m}\beta_j x_{j}^{-\frac32}+2\Bigg(\sum_{j=1}^{m}\beta_j x_{j}^{-\frac12}\Bigg)^3\Bigg),\\
			& d_{2,0}^{(0)}=e^{\pi\ii\sum_{j=1}^{m}\beta_j}=s_1^{\frac12}, \quad  d_{2,0}^{(1)}=-2\ii\sum_{j=1}^{m}\beta_j x_{j}^{-\frac12},\\
			& d_{2,0}^{(2)}=-2\Bigg(\sum_{j=1}^{m}\beta_j x_{j}^{-\frac12}\Bigg)^2, \quad d_{2,0}^{(3)}=\frac{2\ii}{3}\Bigg(\sum_{j=1}^{m}\beta_j x_{j}^{-\frac32}+2\Bigg(\sum_{j=1}^{m}\beta_j x_{j}^{-\frac12}\Bigg)^3\Bigg).
		\end{align}
	\end{itemize}
\end{proposition}

With the above preparations, we can construct the solution to the RH problem for $N_L$ as
\begin{equation}\label{def: NL}
N_L(z)=\begin{pmatrix}
	1 & 0 \\
	-2\sum\limits_{j=1}^{m}\beta_j x_{j}^{\frac12} & 1
\end{pmatrix}
z^{-\frac14\sigma_3}\frac{I+\ii\sigma_1}{\sqrt{2}}\begin{pmatrix}
	d_1(z) & 0 \\
	0 & d_2(z)
\end{pmatrix}, 
\end{equation}
where $d_1$ and $d_2$ are defined in \eqref{def: d1d2}. Indeed, using 
the jump relations of $d_1$ and $d_2$ stated in item (i) of the 
Proposition \ref{prop:d1d2}, it follows that for $z\in(-x_j,-x_{j-1})$,
\begin{align*}
	N_{L,-}(z)^{-1}N_{L,+}(z)&=\begin{pmatrix}
		{d_{1,-}(z)}^{-1} & 0 \\
		0 & {d_{2,-}(z)}^{-1}
	\end{pmatrix}\frac{I-i\sigma_1}{\sqrt{2}}z_{-}^{\frac14\sigma_3}
	z_{+}^{-\frac14\sigma_3}\frac{I+\ii\sigma_1}{\sqrt{2}}\begin{pmatrix}
		d_{1,+}(z) & 0 \\
		0 & d_{2,+}(z)
	\end{pmatrix}\\
	&=\begin{pmatrix}
		0 & d_{1,-}(z)^{-1}d_{2,+}(z) \\
		-d_{2,-}(z)^{-1}d_{1,+}(z) & 0
	\end{pmatrix}\\
	&\overset{\eqref{eq:dijump-multi}}{=}
	\begin{pmatrix}
		0 & s_j \\
		-\frac{1}{s_j} & 0
	\end{pmatrix},
\end{align*}
where we have used the fact that $z_{-}^{\frac14}z_{+}^{-\frac14}=-\ii$ for $z\in(-x_j,-x_{j-1})$ to the second equality.
By \eqref{eq:asy d1 at infinity} and \eqref{eq:asy d2 at infinity}, it is also readily verified that $N_L$ satisfies the asymptotic condition 
\eqref{eq:asy NL at infinity}. 

Now we are ready to give the explicit expression of $N$.
\begin{proposition}
The solution of the RH problem for $N$ is given by
\begin{equation}\label{def: N-global}
	N(z)=\mathcal{C}_N
	\diag\left((-z)^{-\frac14}, z^{-\frac14}, (-z)^{\frac14},  z^{\frac14} \right)A\diag\left(d_1(-z), d_1(z), d_2(-z), d_2(z)\right),
\end{equation}
where $A$ is defined in \eqref{def:A}, $d_1$ and $d_2$ are defined in \eqref{def: d1d2}, and
\begin{equation}
	\mathcal{C}_N=\begin{pmatrix}
		1 & 0 & 0 & 0 \\
		0 & 1 & 0 & 0 \\
		2\sum\limits_{j=1}^{m}\beta_j x_{j}^{\frac12} & 0 & 1 & 0 \\
		0 & -2\sum\limits_{j=1}^{m}\beta_j x_{j}^{\frac12} & 0 & 1
	\end{pmatrix}.
\end{equation}
Moreover, as $z\to \infty$, we have
\begin{equation}\label{eq:asy N at infinity-2}
	N(z)=\left(I+\frac{N_1}{z}+\Boh(z^{-2})\right)\diag\left((-z)^{-\frac14}, z^{-\frac14}, (-z)^{\frac14},  z^{\frac14} \right)A,
\end{equation}
where
\begin{equation}\label{eq: N_1}
	N_1=\begin{pmatrix}
		2\left(\sum\limits_{j=1}^{m}\beta_j x_{j}^{\frac12}\right)^2 & 0 & -2 \sum\limits_{j=1}^{m}\beta_j x_{j}^{\frac12} & 0 \\
		0 & * & 0 & * \\
		* & 0 & * & 0 \\
		0 & * & 0 & *
	\end{pmatrix}.
\end{equation}
The entries marked by $*$ are not required below and are therefore left unspecified.

As $z\to -x_j$, $j=1,\ldots,m$ with $\im z>0$, we have
\begin{align}\label{eq:N-asy-(-x_j)}
	N(z)=\left(N_{-x_j}^{(0)}+N_{-x_j}^{(1)}(z+x_j)+\Boh((z+x_j)^2)\right)(z+x_j)^{\beta_j(E_{22}-E_{44})}, 
\end{align}
where 
\begin{align}
	N_{-x_j}^{(0)}&=\mathcal{C}_N\diag(x_j^{-\frac14}, e^{-\frac{\pi i}{4}}x_j^{-\frac14}, x_j^{\frac14}, e^{\frac{\pi i}{4}}x_j^{\frac14})A\diag\left(d_{1,x_j}^{(0)}, d_{1,-x_j}^{(0)}, d_{2,x_j}^{(0)}, d_{2,-x_j}^{(0)}\right), \label{eq:N(-x_j)0}\\
	N_{-x_j}^{(1)}&=\mathcal{C}_N\diag(x_j^{-\frac14}, e^{-\frac{\pi i}{4}}x_j^{-\frac14}, x_j^{\frac14}, e^{\frac{\pi i}{4}}x_j^{\frac14}) \nonumber\\ 
	&\times\left(\frac{1}{4x_j}\diag(1,1,-1,-1)A\diag(d_{1,x_j}^{(0)}, d_{1,-x_j}^{(0)}, d_{2,x_j}^{(0)}, d_{2,-x_j}^{(0)}) \right. \nonumber\\ 
	&\left. \hspace*{5em}+A\diag\left(-d_{1,x_j}^{(0)}d_{1,x_j}^{(1)}, d_{1,-x_j}^{(0)}d_{1,-x_j}^{(1)}, -d_{2,x_j}^{(0)}d_{2,x_j}^{(1)}, d_{2,-x_j}^{(0)}d_{2,-x_j}^{(1)}\right) \right). \label{eq:N(-x_j)1}
\end{align}
As $z\to x_j$, $j=1,\ldots,m$ with $\im z>0$, we have
\begin{align}\label{eq:N-asy-(x_j)}
    N(z)=\left(N_{x_j}^{(0)}+N_{x_j}^{(1)}(z-x_j)+\Boh((z-x_j)^2)\right)(x_j-z)^{\beta_j(E_{33}-E_{11})},
\end{align}
where
\begin{align}
    N_{x_j}^{(0)}&=\mathcal{C}_N\diag\left(e^{\frac{\pi i}{4}}x_j^{-\frac14}, x_j^{-\frac14}, e^{-\frac{\pi i}{4}}x_j^{\frac14}, x_j^{\frac14}\right)A \diag\left(s_j^{-1}d_{2,-x_j}^{(0)}, d_{1,x_j}^{(0)}, s_j d_{1,-x_j}^{(0)}, d_{2,x_j}^{(0)}\right), \label{equ: Nx_j^(0)}\\
    N_{x_j}^{(1)}&=\mathcal{C}_N\diag\left(e^{\frac{\pi i}{4}}x_j^{-\frac14}, x_j^{-\frac14}, e^{-\frac{\pi i}{4}}x_j^{\frac14}, x_j^{\frac14}\right)\nonumber\\
    &\times\Bigg(\frac{1}{4x_j}\diag(-1,-1,1,1)A\diag\left(s_j^{-1}d_{2,-x_j}^{(0)}, d_{1,x_j}^{(0)}, s_j d_{1,-x_j}^{(0)}, d_{2,x_j}^{(0)}\right)\nonumber\\
    &\hspace*{4em}+A\diag\left(-s_j^{-1}d_{2,-x_j}^{(0)}d_{2,-x_j}^{(1)}, d_{1,x_j}^{(0)}d_{1,x_j}^{(1)}, -s_j d_{1,-x_j}^{(0)}d_{1,-x_j}^{(1)}, d_{2,x_j}^{(0)}d_{2,x_j}^{(1)}\right)\Bigg). \label{equ: Nx_j^(1)}
\end{align}
As $z\to 0$ with $\im z>0$, we have
\begin{align}\label{eq:N-asy-0}
    N(z)&=\mathcal{C}_N
    \diag\left(e^{\frac{\pi i}{4}}z^{-\frac14}, z^{-\frac14}, e^{-\frac{\pi i}{4}}z^{\frac14}, z^{\frac14}\right)
    A\diag\left(s_1^{-\frac12}, s_1^{-\frac12}, s_1^{\frac12}, s_1^{\frac12}\right)\nonumber\\
    &\quad \times \left(I+\left(2\sum_{j=1}^{m}\beta_j x_j^{-\frac12}\right)\diag(1,\ii,-1,-\ii)z^{\frac12}+\Boh(z)\right).
\end{align}
\end{proposition}

\begin{proof}
The formula \eqref{def: N-global} is obtained by combining the two $2\times 2$ model problems $N_L$ and $N_R$ in \eqref{def: NL and NR} with the symmetry relation \eqref{symmetry NL and NR}. The jump relations on $(-\infty,0)$ and $(0,+\infty)$ then follow directly from those of $N_L$ and $N_R$, and they agree with the jump matrix in \eqref{eq:jump-N}. The expansion \eqref{eq:asy N at infinity-2} follows from \eqref{eq:asy d1 at infinity}, \eqref{eq:asy d2 at infinity}, and the definition of $\mathcal C_N$, which cancels the $z^{-\frac12}$ term in the large-$z$ expansion.

The local expansions near $\pm x_j$, $j=1,\ldots,m$, are obtained by substituting the asymptotic formulas for $d_1$ and $d_2$ from items (iii) and (iv) of Proposition \ref{prop:d1d2} into \eqref{def: N-global}. The only point that requires care is the expansion at $x_j$: if $z\to x_j$ with $\im z>0$, then $-z\to -x_j$ with $\im(-z)<0$. Therefore the jump relations in item (i) of Proposition \ref{prop:d1d2} must be applied to $d_1(-z)$ and $d_2(-z)$, which explains the factors $s_j^{-1}$ and $s_j$ in $N_{x_j}^{(0)}$ and $N_{x_j}^{(1)}$.

The expansion near $0$ is derived in the same way from item (v) of Proposition \ref{prop:d1d2}.
\end{proof}

\subsection{Local parametrix near $-x_j$, $j=1,\ldots,m$}
Since the convergence of jump matrix $J_S$ to $I$ for large positive $r$ is not uniform in the neighborhoods of the 
points $0$ and $\pm x_j$, $j=1,\ldots,m$, we need to construct a local parametrix near each of these points. 
Let us start with the local parametrix near $-x_j$, $j=1,\ldots,m$.

Near $z=-x_j$, we intend to find a solution $P^{(-x_j)}$ to the following RH problem.
\begin{paragraph}{RH problem for $P^{(-x_j)}$}
	\begin{itemize}
		\item [(a)] $P^{(-x_j)}(z)$ is analytic in $D_{-x_j}\setminus \Gamma_S$, where $\Gamma_S$ is defined in \eqref{jump contour: Gamma_S}.
		\item [(b)] For $z\in D_{-x_j}\cap \Gamma_S$, we have $P^{(-x_j)}_+(z)=P^{(-x_j)}_-(z)J_S(z)$, where $J_S$ is defined in \eqref{eq:jump-S}.
		\item [(c)] As $r\to +\infty$, we have the matching condition
		\begin{equation}\label{eq:matching-P-xj}
		P^{(-x_j)}(z)N(z)^{-1}=I+\Boh(r^{-\frac32}), \quad z\in \partial D_{-x_j},
		\end{equation}
		where $N$ is defined in \eqref{def: N-global}.
	\end{itemize}
\end{paragraph}
Since the entries of $J_S(z)$ have discontinuities at $z=-x_j$, $j=1,\ldots,m$, 
we follow \cite{IK} and construct $P^{(-x_j)}$ using the model RH problem $\Phi_{\CH}$
associated with the confluent hypergeometric functions; see Appendix \ref{app:CH-parametrix}. 

Let 
\begin{equation}\label{def: f-xj}
	f_{-x_j}(z)=-2\ii r^{-\frac32}
	\begin{cases}
		\theta_2(rz)-\theta_{2,+}(-r x_j), & \im z>0,\\
		-\theta_2(rz)+\theta_{2,-}(-r x_j), & \im z<0.
	\end{cases}
\end{equation}
As $z\to -x_j$, $j=1,\ldots,m$, it follows from the definition of $\theta_2$ in \eqref{def:theta2} that
\begin{equation}\label{eq:f-xj}
	f_{-x_j}(z)=2\left(\vr_2 x_j^{\frac12}-\frac{\vs_2}{r x_j^{\frac12}}\right)(z+x_j)+\Boh((z+x_j)^2).
\end{equation}
Thus, $f_{-x_j}(z)$ is a conformal mapping near $z=-x_j$ for sufficiently large positive $r$.

We now define 
\begin{align}\label{def:P-(-x_j)}
    P^{(-x_j)}(z)=&\,\mathcal E_{-x_j}(z)
    \begin{pmatrix}
        1 & 0 & 0 & 0 \\
        0 & \left(\Phi_{\CH}\right)_{11}\left(r^{\frac32}f_{-x_j}(z);-\beta_j\right) & 0 & \left(\Phi_{\CH}\right)_{12}\left(r^{\frac32}f_{-x_j}(z);-\beta_j\right) \\
        0 & 0 & 1 & 0 \\
        0 & \left(\Phi_{\CH}\right)_{21}\left(r^{\frac32}f_{-x_j}(z);-\beta_j\right) & 0 & \left(\Phi_{\CH}\right)_{22}\left(r^{\frac32}f_{-x_j}(z);-\beta_j\right)
    \end{pmatrix}\nonumber\\
    &\times (s_js_{j+1})^{-\frac14(E_{22}-E_{44})}e^{\theta_2(rz)(E_{22}-E_{44})}\mathcal G_{-x_j}(z),
\end{align}
where
\begin{multline}
    \mathcal G_{-x_j}(z)=\\
    \left\{
    \begin{array}{ll}
        \begin{pmatrix}
            1 & -e^{-\theta_1(rz)+\theta_2(rz)+2\tau rz} & 0 & 0 \\
            0 & 1 & 0 & 0 \\
            0 & 0 & 1 & 0 \\
            0 & 0 & e^{-\theta_1(rz)+\theta_2(rz)-2\tau rz} & 1
        \end{pmatrix}, & z\in D_{-x_j}\cap \mathbb{C}^{+}\setminus(\Omega_{-j,+}\cup\Omega_{-j-1,+}),\\[4mm]
        \begin{pmatrix}
            1 & e^{-\theta_1(rz)+\theta_2(rz)+2\tau rz} & 0 & 0 \\
            0 & 1 & 0 & 0 \\
            0 & 0 & 1 & 0 \\
            0 & 0 & -e^{-\theta_1(rz)+\theta_2(rz)-2\tau rz} & 1
        \end{pmatrix}, & z\in D_{-x_j}\cap \mathbb{C}^{-}\setminus(\Omega_{-j,-}\cup\Omega_{-j-1,-}),\\[4mm]
        I, & \text{otherwise,}
    \end{array}
    \right.
\end{multline}
and
\begin{multline}\label{def:mathcal E-(-xj)}
	\mathcal E_{-x_j}(z)=N(z)(s_js_{j+1})^{\frac14(E_{22}-E_{44})}
	\left\{
	\begin{array}{ll}       
	 	I,  & \im z>0 \\ 
	\left(s_js^{-1}_{j+1}\right)^{\frac12 (E_{22}-E_{44})}
	\begin{pmatrix}
		1 & 0 & 0 & 0 \\
		0 & 0 & 0 & 1 \\
		0 & 0 & 1 & 0 \\
		0 & -1 & 0 & 0
	\end{pmatrix}, & \im z<0
	\end{array}
\right\}\\
\times e^{-\theta_{2,+}(-r x_j)(E_{22}-E_{44})}\left(r^{\frac32}f_{-x_j}(z)\right)^{-\beta_j(E_{22}-E_{44})}.
\end{multline}

\begin{proposition}
	The local parametrix defined in \eqref{def:P-(-x_j)} is the solution of the RH problem for $P^{(-x_j)}$.
\end{proposition}
\begin{proof}
We first show that $\mathcal E_{-x_j}$ is analytic in $D_{-x_j}$. By \eqref{def:mathcal E-(-xj)}, the only possible jump of $\mathcal E_{-x_j}$ lies on $(-x_j-\delta,-x_j+\delta)$. We begin with the interval $(-x_j,-x_j+\delta)$. Since \eqref{eq:f-xj} gives $f_{-x_j,+}(z)=f_{-x_j,-}(z)>0$ there, \eqref{eq:jump-N} implies
\begin{align*}
	&\mathcal E_{-x_j,-}^{-1}(z)\mathcal E_{-x_j,+}(z)\\
	&=\left(r^{\frac32}f_{-x_j}(z)\right)^{\beta_j(E_{22}-E_{44})}e^{\theta_{2,+}(-r x_j)(E_{22}-E_{44})}\begin{pmatrix}
		1 & 0 & 0 & 0 \\
		0 & 0 & 0 & -1 \\
		0 & 0 & 1 & 0 \\
		0 & 1 & 0 & 0
	\end{pmatrix}\left(s_js^{-1}_{j+1}\right)^{-\frac12 (E_{22}-E_{44})}(s_js_{j+1})^{-\frac14(E_{22}-E_{44})}\\
	&\hspace*{2.5em}\times N_-^{-1}(z)N_+(z)(s_js_{j+1})^{\frac14(E_{22}-E_{44})}e^{-\theta_{2,+}(-r x_j)(E_{22}-E_{44})}\left(r^{\frac32}f_{-x_j}(z)\right)^{-\beta_j(E_{22}-E_{44})}\\
	&=\left(r^{\frac32}f_{-x_j}(z)\right)^{\beta_j(E_{22}-E_{44})}e^{\theta_{2,+}(-r x_j)(E_{22}-E_{44})}\begin{pmatrix}
		1 & 0 & 0 & 0 \\
		0 & 0 & 0 & -1 \\
		0 & 0 & 1 & 0 \\
		0 & 1 & 0 & 0
	\end{pmatrix}\left(s_js^{-1}_{j+1}\right)^{-\frac12 (E_{22}-E_{44})}(s_js_{j+1})^{-\frac14(E_{22}-E_{44})}\\
	&\hspace*{2.5em}\times 
	\begin{pmatrix}
		1 & 0 & 0 & 0 \\
		0 & 0 & 0 & s_{j} \\
		0 & 0 & 1 & 0 \\
		0 & -\frac{1}{s_{j}} & 0 & 0
	\end{pmatrix}(s_js_{j+1})^{\frac14(E_{22}-E_{44})}e^{-\theta_{2,+}(-r x_j)(E_{22}-E_{44})}\left(r^{\frac32}f_{-x_j}(z)\right)^{-\beta_j(E_{22}-E_{44})}\\
	&=I.
\end{align*}
Next considering $z\in (-x_j-\delta,-x_j)$, we have from \eqref{eq:f-xj} that $f_{-x_j,+}(z)=f_{-x_j,-}(z)<0$, hence
\begin{equation*}
f_{-x_j,+}(z)^{\beta_j}=e^{2\pi\ii\beta_j}f_{-x_j,-}(z)^{\beta_j}.
\end{equation*}
Using again \eqref{eq:jump-N}, we obtain
\begin{align*}
	&\mathcal E_{-x_j,-}^{-1}(z)\mathcal E_{-x_j,+}(z)\\
	&=\left(r^{\frac32}f_{-x_j}(z)\right)_{-}^{\beta_j(E_{22}-E_{44})}e^{\theta_{2,+}(-r x_j)(E_{22}-E_{44})}\begin{pmatrix}
		1 & 0 & 0 & 0 \\
		0 & 0 & 0 & -1 \\
		0 & 0 & 1 & 0 \\
		0 & 1 & 0 & 0
	\end{pmatrix}\left(s_js^{-1}_{j+1}\right)^{-\frac12 (E_{22}-E_{44})}(s_js_{j+1})^{-\frac14(E_{22}-E_{44})}\\
	&\hspace*{2.5em}\times \begin{pmatrix}
		1 & 0 & 0 & 0 \\
		0 & 0 & 0 & s_{j+1} \\
		0 & 0 & 1 & 0 \\
		0 & -\frac{1}{s_{j+1}} & 0 & 0
	\end{pmatrix}(s_js_{j+1})^{\frac14(E_{22}-E_{44})}e^{-\theta_{2,+}(-r x_j)(E_{22}-E_{44})}\left(r^{\frac32}f_{-x_j}(z)\right)_{+}^{-\beta_j(E_{22}-E_{44})}\\
	&=\diag(1, s_js_{j+1}^{-1}, 1, s_j^{-1}s_{j+1})\diag(1, e^{-2\pi\ii\beta_j}, 1, e^{2\pi\ii\beta_j})
	=I.
\end{align*}
Here the last equality follows from \eqref{def: second def beta_j}, since $e^{-2\pi\ii\beta_j}=s_j^{-1}s_{j+1}$. Hence $\mathcal E_{-x_j}$ has no jump in $D_{-x_j}$.

It remains to show that $\mathcal E_{-x_j}$ is bounded at $z=-x_j$.
By \eqref{eq:f-xj}, we have 
\begin{equation}
	f_{-x_j}(-x_j)=0, \quad \quad f'_{-x_j}(-x_j)=2\left(\vr_2 x_j^{\frac12}-\frac{\vs_2}{r x_j^{\frac12}}\right), \quad f''_{-x_j}(-x_j)=- \left( \vr_2 x_j^{-\frac12} + \frac{\vs_2}{r x_j^{\frac32}} \right).
\end{equation}
As $z\to -x_j$, using \eqref{eq:N-asy-(-x_j)}, \eqref{def:mathcal E-(-xj)} and \eqref{eq:f-xj},
we get
\begin{align}\label{eq: E_{-x_j}(z) at -x_j}
	\mathcal E_{-x_j}(z)=\mathcal E_{-x_j}(-x_j)\left(I+\mathcal E_{-x_j}(-x_j)^{-1}\mathcal E_{-x_j}'(-x_j) (z+x_j)+\Boh((z+x_j)^2)\right),
\end{align}
where 
\begin{align}
	&\mathcal E_{-x_j}(-x_j)= N_{-x_j}^{(0)} (s_j s_{j+1})^{\frac14(E_{22}-E_{44})} e^{-\theta_{2,+}(-r x_j)(E_{22}-E_{44})} \left( r^{\frac32} f'_{-x_j}(-x_j)\right)^{-\beta_j(E_{22}-E_{44})}, \label{eq:E_{-x_j}(-x_j)} \\
	&\mathcal E_{-x_j}(-x_j)^{-1}\mathcal E_{-x_j}'(-x_j)=\left(
	\begin{array}{cccc}
		-d_{1,x_j}^{(1)} & 0 & -\frac{\ii}{4x_j}\frac{d_{2,x_j}^{(0)}}{d_{1,x_j}^{(0)}} & 0 \\
		0 & d_{1,-x_j}^{(1)}-\frac{\beta_jf''_{-x_j}(-x_j)}{2 f'_{-x_j}(-x_j)} & 0 & \frac{\ii}{4x_jc_{-x_j}}\frac{d_{2,-x_j}^{(0)}}{d_{1,-x_j}^{(0)}} \\
		\frac{\ii}{4x_j}\frac{d_{1,x_j}^{(0)}}{d_{2,x_j}^{(0)}} & 0 & -d_{2,x_j}^{(1)} & 0 \\
		0 & -\frac{\ii c_{-x_j}}{4x_j}\frac{d_{1,-x_j}^{(0)}}{d_{2,-x_j}^{(0)}} & 0 & d_{2,-x_j}^{(1)}+\frac{\beta_jf''_{-x_j}(-x_j)}{2 f'_{-x_j}(-x_j)}
	\end{array}
	\right), \label{eq:E_{-x_j}(-x_j)E_{-x_j}'(-x_j)}
\end{align}
with
\begin{equation}
	c_{-x_j}=(s_js_{j+1})^{\frac12}e^{-2\theta_{2,+}(-r x_j)}\left(r^{3/2}f_{-x_j}'(-x_j)\right)^{-2\beta_j}.
\end{equation}
Thus, the factor $\mathcal E_{-x_j}(z)$ is indeed analytic in $D_{-x_j}$.
The jump relations of $P^{(-x_j)}$ on $\Gamma_S\cap D_{-x_j}$ can be verified by direct calculations using the definition of $P^{(-x_j)}$ and the jump relations of $\Phi_{\CH}$ stated in \eqref{HJumps}. 

From the definitions of $\theta_1$ and $\theta_2$ in \eqref{def:theta1} and \eqref{def:theta2}, we have
that the exponential terms in $\Gvar_{-x_j}$ are exponentially small as $r\to +\infty$ uniformly for $z\in D_{-x_j}$. Thus, we have $\mathcal G_{-x_j}(z)=I+\Boh(e^{-c r^{3/2}})$ for some positive constant $c$.
Now using \eqref{def: f-xj}, \eqref{def:P-(-x_j)}, \eqref{def:mathcal E-(-xj)} and \eqref{H at infinity}, we have that 
as $r\to+\infty$
\begin{multline}\label{eq: calculation matching condition -x_j}
	P^{(-x_j)}(z)N(z)^{-1}=\\
	I+\frac{1}{r^{\frac32}f_{-x_j}(z)}
	\mathcal E_{-x_j}(z)
	\begin{pmatrix}
		1 & 0 & 0 & 0 \\
		0 & (\Phi_{\CH,1})_{11}(-\beta_j) & 0 & (\Phi_{\CH,2})_{12}(-\beta_j) \\
		0 & 0 & 1 & 0 \\
		0 & (\Phi_{\CH,3})_{21}(-\beta_j) & 0 & (\Phi_{\CH,4})_{22}(-\beta_j)
	\end{pmatrix}
	\mathcal E_{-x_j}(z)^{-1}
	+\Boh(r^{-3}),
\end{multline}
which confirms the matching condition in \eqref{eq:matching-P-xj}.
\end{proof}

\subsection{Local parametrix near $x_j$, $j=1,\ldots,m$}
Near $z=x_j$, $j=1,\ldots,m$, we intend to find a solution $P^{(x_j)}$ to the following RH problem.
\begin{paragraph}{RH problem for $P^{(x_j)}$}
	\begin{itemize}
		\item [(a)] $P^{(x_j)}(z)$ is analytic in $D_{x_j}\setminus \Gamma_S$, where $\Gamma_S$ is defined in \eqref{jump contour: Gamma_S}.
		\item [(b)] For $z\in D_{x_j}\cap \Gamma_S$, we have $P^{(x_j)}_+(z)=P^{(x_j)}_-(z)J_S(z)$, where $J_S$ is defined in \eqref{eq:jump-S}.
		\item [(c)] As $r\to +\infty$, we have the matching condition
		\begin{equation}\label{eq:matching-P+xj}
		P^{(x_j)}(z)N(z)^{-1}=I+\Boh(r^{-\frac32}), \quad z\in \partial D_{x_j},
		\end{equation}
		where $N$ is defined in \eqref{def: N-global}.
	\end{itemize}
\end{paragraph}
Analogously, the above RH problem for $P^{(x_j)}$ could be solved in terms of the confluent hypergeometric parametrix $\Phi_{\CH}$ shown in Appendix \ref{app:CH-parametrix}. 

Define
\begin{equation}\label{def: f+xj}
	f_{x_j}(z)=-2\ii r^{-\frac32}
	\begin{cases}
		\theta_1(rz)-\theta_{1,+}(r x_j), & \im z>0,\\
		-\theta_1(rz)+\theta_{1,-}(r x_j), & \im z<0.
	\end{cases}
\end{equation}
By the definition of $\theta_1$ in \eqref{def:theta1}, we have that as $z\to x_j$, $j=1,\ldots,m$,
\begin{equation}\label{eq:f+xj}
	f_{x_j}(z)=2\left(\vr_1 x_j^{\frac12}-\frac{\vs_1}{r x_j^{\frac12}}\right)(z-x_j)+\Boh((z-x_j)^2). 
\end{equation}
Thus, $f_{x_j}(z)$ is a conformal mapping near $z=x_j$ for sufficiently large positive $r$.

We now define 
\begin{align}\label{def:P+x_j}
    P^{(x_j)}(z)=&\,\mathcal E_{x_j}(z)
    \begin{pmatrix}
        \left(\Phi_{\CH}\right)_{11}\left(r^{\frac32}f_{x_j}(z);\beta_j\right) & 0 & \left(\Phi_{\CH}\right)_{12}\left(r^{\frac32}f_{x_j}(z);\beta_j\right) & 0 \\
        0 & 1 & 0 & 0 \\
        \left(\Phi_{\CH}\right)_{21}\left(r^{\frac32}f_{x_j}(z);\beta_j\right) & 0 & \left(\Phi_{\CH}\right)_{22}\left(r^{\frac32}f_{x_j}(z);\beta_j\right) & 0 \\
        0 & 0 & 0 & 1
    \end{pmatrix}\nonumber\\
    &\times (s_js_{j+1})^{-\frac14(E_{11}-E_{33})}e^{\theta_1(rz)(E_{11}-E_{33})}\mathcal G_{x_j}(z),
\end{align}
where
\begin{multline}\label{def:G+x_j}
    \mathcal G_{x_j}(z)=\\
    \left\{
    \begin{array}{ll}
        \begin{pmatrix}
            1 & 0 & 0 & 0 \\
            -e^{\theta_1(rz)-\theta_2(rz)-2\tau rz} & 1 & 0 & 0 \\
            0 & 0 & 1 & e^{\theta_1(rz)-\theta_2(rz)+2\tau rz} \\
            0 & 0 & 0 & 1
        \end{pmatrix}, & z\in D_{x_j}\cap \mathbb{C}^{+}\setminus(\Omega_{j,+}\cup\Omega_{j+1,+}),\\[4mm]
        \begin{pmatrix}
            1 & 0 & 0 & 0 \\
            e^{\theta_1(rz)-\theta_2(rz)-2\tau rz} & 1 & 0 & 0 \\
            0 & 0 & 1 & -e^{\theta_1(rz)-\theta_2(rz)+2\tau rz} \\
            0 & 0 & 0 & 1
        \end{pmatrix}, & z\in D_{x_j}\cap \mathbb{C}^{-}\setminus(\Omega_{j,-}\cup\Omega_{j+1,-}),\\[4mm]
        I, & \text{otherwise,}
    \end{array}
    \right.
\end{multline}
and
\begin{multline}\label{def:mathcal E+xj}
	\mathcal E_{x_j}(z)=N(z)(s_js_{j+1})^{\frac14(E_{11}-E_{33})}\\
	\times \left\{
	\begin{array}{ll}       
	 	I,  & \im z>0 \\ 
	\left(s_j^{-1}s_{j+1}\right)^{\frac12 (E_{11}-E_{33})}
	\begin{pmatrix}
		0 & 0 & 1 & 0 \\
		0 & 1 & 0 & 0 \\
		-1 & 0 & 0 & 0 \\
		0 & 0 & 0 & 1
	\end{pmatrix}, & \im z<0
	\end{array}
    \right\}\\
    \times e^{-\theta_{1,+}(r x_j)(E_{11}-E_{33})}\left(r^{\frac32}f_{x_j}(z)\right)^{\beta_j(E_{11}-E_{33})}.
\end{multline}

\begin{proposition}
	The local parametrix defined in \eqref{def:P+x_j} is the solution of the RH problem for $P^{(x_j)}$.
\end{proposition}
\begin{proof}
The proof is similar to that of the local parametrix near $-x_j$, and we omit the details here.
\end{proof}

For later use, we also need the asymptotic expansion of $\mathcal E_{x_j}(z)$ as $z\to x_j$.
By \eqref{eq:f+xj}, we have 
\begin{equation}\label{eq: fx_j and ' and ''}
	f_{x_j}(x_j)=0, \quad \quad f'_{x_j}(x_j)=2\left(\vr_1 x_j^{\frac12}-\frac{\vs_1}{r x_j^{\frac12}}\right), \quad f''_{x_j}(x_j)=\frac{\vr_1}{\sqrt{x_j}} + \frac{\vs_1}{r x_j^{3/2}}.
\end{equation}
Using \eqref{eq:N-asy-(x_j)}, \eqref{def:mathcal E+xj} and \eqref{eq:f+xj}, we have that as $z\to x_j$,
\begin{align}\label{eq: E_{x_j}(z) at x_j}
	\mathcal E_{x_j}(z)=\mathcal E_{x_j}(x_j)\left(I+\mathcal E_{x_j}(x_j)^{-1}\mathcal E_{x_j}'(x_j) (z-x_j)+\Boh((z-x_j)^2)\right),
\end{align}
where
\begin{align}
	&\mathcal E_{x_j}(x_j)= N_{x_j}^{(0)} (s_j s_{j+1})^{\frac14(E_{11}-E_{33})} e^{-\theta_{1,+}(r x_j)(E_{11}-E_{33})} \left( e^{\pi\ii} r^{\frac32} f'_{x_j}(x_j)\right)^{\beta_j(E_{11}-E_{33})}, \label{eq:E_{x_j}(x_j)} \\
	&\mathcal E_{x_j}(x_j)^{-1}\mathcal E_{x_j}'(x_j)=\left(
	\begin{array}{cccc}
			-d_{2,-x_j}^{(1)}+\frac{\beta_jf''_{x_j}(x_j)}{2 f'_{x_j}(x_j)} & 0 & \frac{\ii s_j^2}{4x_jc_{x_j}}\frac{d_{1,-x_j}^{(0)}}{d_{2,-x_j}^{(0)}} & 0 \\
			0 & d_{1,x_j}^{(1)} & 0 & -\frac{\ii}{4x_j}\frac{d_{2,x_j}^{(0)}}{d_{1,x_j}^{(0)}}  \\
			-\frac{\ii s_j^{-2} c_{x_j}}{4x_j}\frac{d_{2,-x_j}^{(0)}}{d_{1,-x_j}^{(0)}} & 0 & -d_{1,-x_j}^{(1)}-\frac{\beta_jf''_{x_j}(x_j)}{2 f'_{x_j}(x_j)} & 0 \\
			0 & \frac{\ii}{4x_j}\frac{d_{1,x_j}^{(0)}}{d_{2,x_j}^{(0)}} & 0 & d_{2,x_j}^{(1)}
	\end{array}
	\right), \label{eq:E_{x_j}(x_j)E_{x_j}'(x_j)}
\end{align}
with 
\begin{equation}\label{def: c_xj}
    c_{x_j}=(s_js_{j+1})^{\frac12} e^{-2\theta_{1,+}(r x_j)}\left(e^{\pi\ii}r^{3/2}f_{x_j}'(x_j)\right)^{2\beta_j},
\end{equation}
and $N_{x_j}^{(0)}$ being defined in \eqref{equ: Nx_j^(0)}.

\subsection{Local parametrix near $0$}
Near $z=0$, we intend to find a solution $P^{(0)}$ to the following RH problem.
\begin{paragraph}{RH problem for $P^{(0)}$}
	\begin{itemize}
		\item [(a)] $P^{(0)}(z)$ is analytic in $D_{0}\setminus \Gamma_S$, where $\Gamma_S$ is defined in \eqref{jump contour: Gamma_S}.
		\item [(b)] For $z\in D_{0}\cap \Gamma_S$, we have $P^{(0)}_+(z)=P^{(0)}_-(z)J_S(z)$, where $J_S$ is defined in \eqref{eq:jump-S}.
		\item [(c)] As $r\to +\infty$, we have the matching condition
		\begin{equation}\label{eq:matching-P-0}
		P^{(0)}(z)N(z)^{-1}=I+\Boh(r^{-\frac12}), \quad z\in \partial D_{0},
		\end{equation}
		where $N$ is defined in \eqref{def: N-global}.
	\end{itemize}
\end{paragraph}
To solve the above RH problem, we need to use the model tacnode RH problem for $M$ given in Section \ref{sec:tacnode-RH}.
Precisely, we define
\begin{equation}\label{def:P-0}
	P^{(0)}(z)=\mathcal E_0(z)M(rz)\diag\left(s_{1}^{-\frac12}e^{\theta_1(rz)-\tau rz}, s_{1}^{-\frac12}e^{\theta_2(rz)+\tau rz}, s_{1}^{\frac12}e^{-\theta_1(rz)-\tau rz}, s_{1}^{\frac12}e^{-\theta_2(rz)+\tau rz}\right),
\end{equation}
where 
\begin{align}\label{def:mathcal E-0}
	\mathcal E_0(z)=N(z)\diag\left(s_{1}^{\frac12}, s_{1}^{\frac12}, s_{1}^{-\frac12}, s_{1}^{-\frac12}\right)A^{-1}
	\diag\left((-rz)^\frac14, (rz)^\frac14, (-rz)^{-\frac14}, (rz)^{-\frac14}\right),
\end{align}
with $A$ being defined in \eqref{def:A} and $N(z)$ being defined in \eqref{def: N-global}.

\begin{proposition}
	The local parametrix defined in \eqref{def:P-0} is the solution of the RH problem for $P^{(0)}$.
\end{proposition}
\begin{proof}
We first show that $\mathcal E_0(z)$ is analytic in $D_0$. By \eqref{def:mathcal E-0},
$\mathcal E_0(z)$ has only possible jump on $(-\delta,\delta)$ and singularity at $z=0$. 
For $z\in (0,\delta)$, a direct calculation shows that
\begin{align*}
	&\mathcal E_{0,-}^{-1}(z)\mathcal E_{0,+}(z)\\
	&=\diag\left((-rz)_{-}^{-\frac14}, (rz)^{-\frac14}, (-rz)_{-}^{\frac14}, (rz)^{\frac14}\right)A \diag\left(s_{1}^{-\frac12}, s_{1}^{-\frac12}, s_{1}^{\frac12}, s_{1}^{\frac12}\right)\\
	&\hspace*{2.5em}\begin{pmatrix}
		0 & 0 & s_1 & 0 \\
		0 & 1 & 0 & 0 \\
		-1/s_1 & 0 & 0 & 0 \\
		0 & 0 & 0 & 1
	\end{pmatrix}	
	\diag\left(s_{1}^{\frac12}, s_{1}^{\frac12}, s_{1}^{-\frac12}, s_{1}^{-\frac12}\right)A^{-1}\diag\left((-rz)_{+}^{\frac14}, (rz)^{\frac14}, (-rz)_{+}^{-\frac14}, (rz)^{-\frac14}\right)\\
	&=\diag\left((-rz)_{-}^{-\frac14}, (rz)^{-\frac14}, (-rz)_{-}^{\frac14}, (rz)^{\frac14}\right)\diag(\ii, 1, -\ii, 1)\diag\left((-rz)_{+}^{\frac14}, (rz)^{\frac14}, (-rz)_{+}^{-\frac14}, (rz)^{-\frac14}\right)\\
	&=I.
\end{align*}
Analogously, for $z\in (-\delta,0)$, we have 
\begin{align*}
	&\mathcal E_{0,-}^{-1}(z)\mathcal E_{0,+}(z)\\
	&=\diag\left((-rz)^{-\frac14}, (rz)_{-}^{-\frac14}, (-rz)^{\frac14}, (rz)_{-}^{\frac14}\right)A \diag\left(s_{1}^{-\frac12}, s_{1}^{-\frac12}, s_{1}^{\frac12}, s_{1}^{\frac12}\right)\\
	&\hspace*{2.5em}\begin{pmatrix}
		1 & 0 & 0 & 0 \\
		0 & 0 & 0 & s_1 \\
		0 & 0 & 1 & 0 \\
		0 & -1/s_1 & 0 & 0
	\end{pmatrix}	
	\diag\left(s_{1}^{\frac12}, s_{1}^{\frac12}, s_{1}^{-\frac12}, s_{1}^{-\frac12}\right)A^{-1}\diag\left((-rz)^{\frac14}, (rz)_{+}^{\frac14}, (-rz)^{-\frac14}, (rz)_{+}^{-\frac14}\right)\\
	&=\diag\left((-rz)^{-\frac14}, (rz)_{-}^{-\frac14}, (-rz)^{\frac14}, (rz)_{-}^{\frac14}\right)\diag(1, -\ii, 1, \ii)\diag\left((-rz)^{\frac14}, (rz)_{+}^{\frac14}, (-rz)^{-\frac14}, (rz)_{+}^{-\frac14}\right)\\
	&=I.
\end{align*}
Moreover, as $z\to 0$, it follows from \eqref{eq:N-asy-0} and \eqref{def:mathcal E-0} that
\begin{align}\label{eq: mathcalE_0(z) asy near z=0}
	\mathcal E_0(z)=\mathcal E_0(0)\left(I+\mathcal E_0(0)^{-1}\mathcal E_0'(0) z + \Boh(z^2)\right)\diag\left(r^{\frac14}, r^{\frac14}, r^{-\frac14}, r^{-\frac14}\right),
\end{align}
where 
\begin{align}\label{eq: mathcalE_0(0)}
	& \mathcal E_0(0)=\left(
	\begin{array}{cccc}
		1 & 0 & -2\sum\limits_{j=1}^m \beta_jx_j^{-\frac12} & 0 \\
		0 & 1 & 0 & 2\sum\limits_{j=1}^m \beta_jx_j^{-\frac12} \\
		2\sum\limits_{j=1}^m \beta_jx_j^{-\frac12} & 0 & 1-4\left(\sum\limits_{j=1}^m \beta_jx_j^{-\frac12}\right)^2 & 0 \\
		0 & -2\sum\limits_{j=1}^m \beta_jx_j^{-\frac12} & 0 & 1-4\left(\sum\limits_{j=1}^m \beta_jx_j^{-\frac12}\right)^2
	\end{array}
	\right),
\end{align}
and
\begin{multline}
\mathcal E_0(0)^{-1}\mathcal E_0'(0)=\\
\left(	\begin{array}{cccc}
	-2 (\sum\limits_{j=1}^m \beta_j x_j^{-\frac12})^2 & 0 & -\frac{2}{3} \sum\limits_{j=1}^m \beta_j x_j^{-\frac32} + \frac{8}{3} (\sum\limits_{j=1}^m \beta_j x_j^{-\frac12})^3 & 0 \\
	0 & 2 (\sum\limits_{j=1}^m  \beta_j x_j^{-\frac12})^2 & 0 & -\frac{2}{3} \sum\limits_{j=1}^m  \beta_j x_j^{-\frac32} + \frac{8}{3} (\sum\limits_{j=1}^m  \beta_j x_j^{-\frac12})^3 \\
	-2 \sum\limits_{j=1}^m  \beta_j x_j^{-\frac12} & 0 & 2 (\sum\limits_{j=1}^m  \beta_j x_j^{-\frac12})^2 & 0 \\
	0 & -2 \sum\limits_{j=1}^m  \beta_j x_j^{-\frac12} & 0 & -2 (\sum\limits_{j=1}^m  \beta_j x_j^{-\frac12})^2
\end{array} \right).
\end{multline}
Thus $\mathcal E_0(z)$ is analytic in $D_0$. The jump relations for $P^{(0)}$ on $\Gamma_S\cap D_0$ then follow directly from the definition of $P^{(0)}$, the jump relations of $M$ in \eqref{jumps:M}, and the analyticity of $\mathcal E_0$.

Finally, it follows from \eqref{def:P-0} and \eqref{eq:asy:M} that
\begin{align}
	&P^{(0)}(z)N(z)^{-1}=I+\frac{1}{r^{1/2}z}\hat{\mathcal E}_0(z)\hat{M}_1\hat{\mathcal E}_0(z)^{-1}+\frac{1}{rz}\hat{\mathcal E}_0(z)\hat{M}_2\hat{\mathcal E}_0(z)^{-1}+\Boh(r^{-\frac32}), \label{eq: calculation matching condition 0} \\
	&\hat{M}_1=\left(\begin{array}{cccc}
	0 & 0 & (M_1)_{13} & (M_1)_{14} \\
	0 & 0 & (M_1)_{23} & (M_1)_{24} \\
	0 & 0 & 0 & 0 \\
	0 & 0 & 0 & 0
	\end{array}
	\right), \quad
	\hat{M}_2=\left(\begin{array}{cccc}
	(M_1)_{11} & (M_1)_{12} & 0  & 0 \\
	(M_1)_{21} & (M_1)_{22} & 0 & 0 \\
	0 & 0 & (M_1)_{33} & (M_1)_{34} \\
	0 & 0 & (M_1)_{43} & (M_1)_{44}
	\end{array}
	\right), \label{eq: hatM1, hatM2}
\end{align}
as $r\to +\infty$, uniformly for $z\in \partial D_0$, with
\begin{equation}\label{eq:hatE0}
	\hat{\mathcal E}_0(z)=\mathcal E_0(z)\diag\left(r^{-\frac14}, r^{-\frac14}, r^{\frac14}, r^{\frac14}\right),
\end{equation}
and $M_1$ as in \eqref{eq:asy:M}. This proves \eqref{eq:matching-P-0}.

This completes the proof.
\end{proof}

\subsection{Final transformation}
Define
\begin{equation}\label{def:R}
	R(z)=\begin{cases}
		S(z)N(z)^{-1}, & z\in \mathbb{C}\setminus \cup_{j=1}^m (D_{-x_j}\cup D_{x_j})\cup D_0,\\
	S(z)P^{(p)}(z)^{-1}, & z\in D_{p}, \ p\in\bigcup_{j=1}^{m} \{-x_j, x_j\}\cup \{0\}.
	\end{cases}
\end{equation}
Recalling the RH problems for $S$, $N$ and $P^{(p)}$, we can verify that $R$ satisfies the following RH problem.
\begin{paragraph}{RH problem for $R$}
	\begin{itemize}
		\item [(a)] $R(z)$ is analytic in $\mathbb{C}\setminus \Gamma_R$, where 
		\begin{equation}
			\Gamma_{R}:=\bigcup_{j=1}^m \left(\partial D_{-x_j}\cup \partial D_{x_j}\right)\cup \partial D_0\cup \Gamma_S\setminus \left(\bigcup_{j=1}^m (D_{-x_j}\cup D_{x_j})\cup D_0\cup\mathbb{R}\right);
		\end{equation}
		see Figure \ref{fig:R} for an illustration.
		\item [(b)] For $z\in \Gamma_R$, we have $R_+(z)=R_-(z)J_R(z)$, where
		\begin{equation}
			J_R(z)=\begin{cases}
				P^{(p)}(z)N(z)^{-1}, & z\in \partial D_{p}, \ p\in\bigcup_{j=1}^{m} \{-x_j, x_j\}\cup \{0\},\\
				N(z)J_S(z)N(z)^{-1}, & z\in \Gamma_R\setminus \left(\bigcup_{j=1}^m (\partial D_{-x_j}\cup \partial D_{x_j})\cup \partial D_0\right).
			\end{cases}
		\end{equation}
		\item [(c)] As $z\to \infty$, we have
		\begin{equation}\label{eq: R at z=infinity}
			R(z)=I+\frac{R_1}{z}+\mathcal{O}\left(\frac{1}{z^2}\right),
		\end{equation}
		where $R_1$ is independent of $z$.
	\end{itemize}
\end{paragraph}

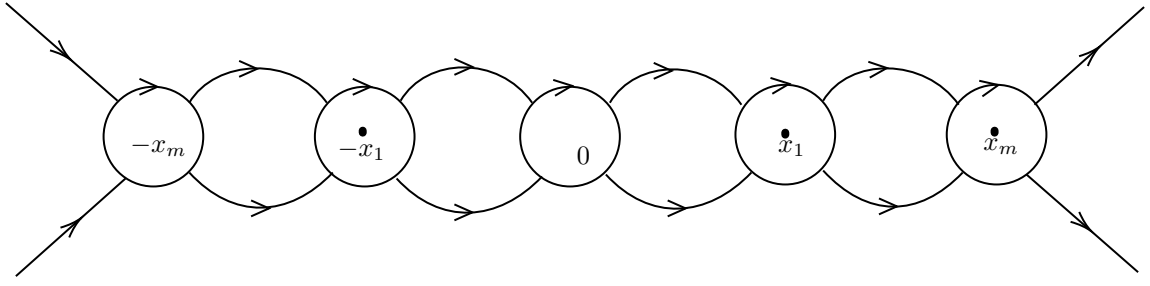
\begin{figure}[t]
\centering
\tikzset{every picture/.style={line width=0.75pt}} 
\begin{tikzpicture}[x=0.75pt,y=0.75pt,yscale=-1,xscale=1]
\draw    (46,67) -- (102,116) ;
\draw [shift={(78.52,95.45)}, rotate = 221.19] [color={rgb, 255:red, 0; green, 0; blue, 0 }  ][line width=0.75]    (10.93,-3.29) .. controls (6.95,-1.4) and (3.31,-0.3) .. (0,0) .. controls (3.31,0.3) and (6.95,1.4) .. (10.93,3.29)   ;
\draw   (95,134) .. controls (95,120.19) and (106.19,109) .. (120,109) .. controls (133.81,109) and (145,120.19) .. (145,134) .. controls (145,147.81) and (133.81,159) .. (120,159) .. controls (106.19,159) and (95,147.81) .. (95,134) -- cycle ;
\draw    (106,155) -- (51,204) ;
\draw [shift={(83.73,174.84)}, rotate = 138.3] [color={rgb, 255:red, 0; green, 0; blue, 0 }  ][line width=0.75]    (10.93,-3.29) .. controls (6.95,-1.4) and (3.31,-0.3) .. (0,0) .. controls (3.31,0.3) and (6.95,1.4) .. (10.93,3.29)   ;
\draw   (201,134) .. controls (201,120.19) and (212.19,109) .. (226,109) .. controls (239.81,109) and (251,120.19) .. (251,134) .. controls (251,147.81) and (239.81,159) .. (226,159) .. controls (212.19,159) and (201,147.81) .. (201,134) -- cycle ;
\draw   (304,134) .. controls (304,120.19) and (315.19,109) .. (329,109) .. controls (342.81,109) and (354,120.19) .. (354,134) .. controls (354,147.81) and (342.81,159) .. (329,159) .. controls (315.19,159) and (304,147.81) .. (304,134) -- cycle ;
\draw    (244,116) .. controls (257,94) and (293,93) .. (310,117) ;
\draw    (315,154) .. controls (293,177) and (264,179) .. (242,155) ;
\draw    (349,117) .. controls (362,95) and (398,94) .. (415,118) ;
\draw    (420,152) .. controls (398,175) and (369,177) .. (347,153) ;
\draw   (412,133) .. controls (412,119.19) and (423.19,108) .. (437,108) .. controls (450.81,108) and (462,119.19) .. (462,133) .. controls (462,146.81) and (450.81,158) .. (437,158) .. controls (423.19,158) and (412,146.81) .. (412,133) -- cycle ;
\draw   (518,133) .. controls (518,119.19) and (529.19,108) .. (543,108) .. controls (556.81,108) and (568,119.19) .. (568,133) .. controls (568,146.81) and (556.81,158) .. (543,158) .. controls (529.19,158) and (518,146.81) .. (518,133) -- cycle ;
\draw    (558,153) -- (614,202) ;
\draw [shift={(590.52,181.45)}, rotate = 221.19] [color={rgb, 255:red, 0; green, 0; blue, 0 }  ][line width=0.75]    (10.93,-3.29) .. controls (6.95,-1.4) and (3.31,-0.3) .. (0,0) .. controls (3.31,0.3) and (6.95,1.4) .. (10.93,3.29)   ;
\draw    (617,69) -- (562,118) ;
\draw [shift={(594.73,88.84)}, rotate = 138.3] [color={rgb, 255:red, 0; green, 0; blue, 0 }  ][line width=0.75]    (10.93,-3.29) .. controls (6.95,-1.4) and (3.31,-0.3) .. (0,0) .. controls (3.31,0.3) and (6.95,1.4) .. (10.93,3.29)   ;
\draw  [color={rgb, 255:red, 0; green, 0; blue, 0 }  ,draw opacity=1 ] (113,105) -- (122,109) -- (113,113) ;
\draw  [color={rgb, 255:red, 0; green, 0; blue, 0 }  ,draw opacity=1 ] (220,105) -- (229,109) -- (220,113) ;
\draw  [color={rgb, 255:red, 0; green, 0; blue, 0 }  ,draw opacity=1 ] (271,168) -- (280,172) -- (271,176) ;
\draw  [color={rgb, 255:red, 0; green, 0; blue, 0 }  ,draw opacity=1 ] (271,95) -- (280,99) -- (271,103) ;
\draw  [color={rgb, 255:red, 0; green, 0; blue, 0 }  ,draw opacity=1 ] (320.57,105.54) -- (329.97,108.49) -- (321.48,113.49) ;
\draw  [color={rgb, 255:red, 0; green, 0; blue, 0 }  ,draw opacity=1 ] (371,97) -- (380,101) -- (371,105) ;
\draw  [color={rgb, 255:red, 0; green, 0; blue, 0 }  ,draw opacity=1 ] (378,166) -- (387,170) -- (378,174) ;
\draw  [color={rgb, 255:red, 0; green, 0; blue, 0 }  ,draw opacity=1 ] (429.57,104.54) -- (438.97,107.49) -- (430.49,112.49) ;
\draw  [color={rgb, 255:red, 0; green, 0; blue, 0 }  ,draw opacity=1 ] (534.3,105.97) -- (543.91,108.11) -- (535.88,113.81) ;
\draw  [line width=3] [line join = round][line cap = round] (120,131) .. controls (120,131) and (120,131) .. (120,131) ;
\draw  [line width=3] [line join = round][line cap = round] (225,131) .. controls (225,131.33) and (225,131.67) .. (225,132) ;
\draw  [line width=3] [line join = round][line cap = round] (328,133) .. controls (328,133) and (328,133) .. (328,133) ;
\draw  [line width=3] [line join = round][line cap = round] (437,132) .. controls (437,132.33) and (437,132.67) .. (437,133) ;
\draw  [line width=3] [line join = round][line cap = round] (542,131) .. controls (542,131.33) and (542,131.67) .. (542,132) ;
\draw    (138,117) .. controls (151,95) and (190,93) .. (207,117) ;
\draw    (210,152) .. controls (188,175) and (160,176) .. (138,152) ;
\draw  [color={rgb, 255:red, 0; green, 0; blue, 0 }  ,draw opacity=1 ] (163,96) -- (172,100) -- (163,104) ;
\draw  [color={rgb, 255:red, 0; green, 0; blue, 0 }  ,draw opacity=1 ] (169,166) -- (178,170) -- (169,174) ;
\draw    (456,116) .. controls (469,94) and (507,94) .. (524,118) ;
\draw  [color={rgb, 255:red, 0; green, 0; blue, 0 }  ,draw opacity=1 ] (478.65,96.44) -- (487.98,99.58) -- (479.39,104.4) ;
\draw    (526,152) .. controls (504,175) and (478,175) .. (456,151) ;
\draw  [color={rgb, 255:red, 0; green, 0; blue, 0 }  ,draw opacity=1 ] (484,165) -- (493,169) -- (484,173) ;
\draw (107,132.4) node [anchor=north west][inner sep=0.75pt]  [font=\small]  {$-x_{m}$};
\draw (432,134.4) node [anchor=north west][inner sep=0.75pt]  [font=\small]  {$x_{1}$};
\draw (331,137.4) node [anchor=north west][inner sep=0.75pt]  [font=\small]  {$0$};
\draw (535,133.4) node [anchor=north west][inner sep=0.75pt]  [font=\small]  {$x_{m}$};
\draw (211,134.4) node [anchor=north west][inner sep=0.75pt]  [font=\small]  {$-x_{1}$};
\end{tikzpicture}
   \caption{The jump contours of the RH problem for $R$ with $m=2$.}
   \label{fig:R}
\end{figure}

The jump matrix $J_R$ satisfies the following large $r$ estimates.
\begin{proposition}
	As $r\to +\infty$, we have
	\begin{align}\label{eq:estimate-JR}
		&J_R(z)=I+\Boh(e^{-cr^{3/2}}), && \text{uniformly for } z\in \Gamma_R\setminus \left(\bigcup_{j=1}^m (\partial D_{-x_j}\cup \partial D_{x_j})\cup \partial D_0\right), \nonumber\\
		&J_R(z)=I+\Boh(r^{-\frac32}), && \text{uniformly for } z\in \bigcup_{j=1}^m (\partial D_{-x_j}\cup \partial D_{x_j}),\\
		&J_R(z)=I+\frac{J^{(1)}(z)}{r^{1/2}}+\frac{J^{(2)}(z)}{r}+\Boh(r^{-\frac32}), && \text{uniformly for } z\in \partial D_0, \nonumber
	\end{align}
	where 
	\begin{equation}\label{eq: def J^(1) and J^(2)}
		J^{(1)}(z)=\frac{1}{z}\hat{\mathcal E}_0(z)\hat{M}_1\hat{\mathcal E}_0(z)^{-1}, \quad J^{(2)}(z)=\frac{1}{z}\hat{\mathcal E}_0(z)\hat{M}_2\hat{\mathcal E}_0(z)^{-1},
	\end{equation}
	with $\hat{\mathcal E}_0(z)$, $\hat{M}_1$ and $\hat{M}_2$ defined in \eqref{eq: hatM1, hatM2}--\eqref{eq:hatE0}
\end{proposition}

Therefore $R$ satisfies a small-norm RH problem \cites{DZ1993, Deift1999}, and we have
\begin{equation}\label{eq: R asymptotics as r to +infinity}
	R(z)=I+\frac{R^{(1)}(z)}{r^{1/2}}+\frac{R^{(2)}(z)}{r}+\Boh(r^{-3/2}), \quad r\to +\infty,
\end{equation}
uniformly for $z\in \mathbb{C}\setminus \Gamma_R$. Together with \eqref{eq:estimate-JR} and integral equation
\begin{equation}
	R(z)=I+\frac{1}{2\pi \ii}\int_{\Gamma_R} \frac{R_-(\xi)(J_R(\xi)-I)}{\xi-z}\dif \xi, \quad z\in \mathbb{C}\setminus \Gamma_R,
\end{equation}
it is deduced that
\begin{equation}\label{eq: R(1)_definition} 
R^{(1)}(z)=\frac{1}{2\pi i}\oint_{\partial D_0} \frac{J^{(1)}(\xi)}{\xi-z}\dif \xi
=\left\{
\begin{array}{ll}
	\dfrac{1}{z}\underset{\xi=0}{\res} J^{(1)}(\xi)-J^{(1)}(z),  & z\in D_0,\\
	\dfrac{1}{z}\underset{\xi=0}{\res} J^{(1)}(\xi),  &\text{elsewhere}.
\end{array}
\right.
\end{equation}
and
\begin{multline}\label{eq: R(2)_definition}
R^{(2)}(z)=\frac{1}{2\pi i}\oint_{\partial D_0} \frac{R^{(1)}_{-}(\xi)J^{(1)}(\xi)+J^{(2)}(\xi)}{\xi-z}\dif \xi \\
=\left\{
\begin{array}{ll}
	\frac{1}{z}\underset{\xi=0}{\res} \left(R^{(1)}_{-}(\xi)J^{(1)}(\xi)+J^{(2)}(\xi)\right)-R^{(1)}_{-}(z)J^{(1)}(z)-J^{(2)}(z),  & z\in D_0,\\
	\frac{1}{z}\underset{\xi=0}{\res} \left(R^{(1)}_{-}(\xi)J^{(1)}(\xi)+J^{(2)}(\xi)\right),  &\text{elsewhere}.
\end{array}
\right.
\end{multline}
Moreover, the asymptotic expansion \eqref{eq: R asymptotics as r to +infinity} also gives us the following asymptotics for
the derivatives of $R$ with respect to the parameter $\beta_j$.
\begin{proposition}
For any $k_1,\ldots,k_m \in \mathbb{N}^{+}\cup\{0\}$ with $k=k_1+\ldots+k_m \geq 1$, we have
\begin{equation}\label{eq: R-beta-derivative}
	\partial_{\beta}^{k}R(z)=\partial^{k}_{\beta} R^{(1)}(z)r^{-1/2}+\partial^{k}_{\beta} R^{(2)}(z)r^{-1}+\Boh\left(\frac{(\log r)^k}{r^{3/2}}\right), \quad \text{uniformly for } z\in \mathbb{C}\setminus \Gamma_R,
\end{equation}
where $k=k_1+\cdots+k_m$, $\partial_{\beta}^{k}=\partial_{\beta_1}^{k_1}\ldots\partial_{\beta_m}^{k_m}$.
\end{proposition}
The proof follows from the same analysis as in \cite[Section 3.5]{CL23}, and we therefore omit the details.
Recalling the definition of $\beta_j$ in \eqref{def: first def beta_j}, we immediately have the following corollary.
\begin{corollary}
For any $k_1,\ldots,k_m \in \mathbb{N}^{+}\cup\{0\}$ with $k_1+\cdots+k_m\geq 1$, 
we have
\begin{multline}\label{eq:derivative-R}
	\partial^{k_1}_{u_1}\ldots \partial^{k_m}_{u_m} R(z)\\
	=\partial^{k_1}_{u_1}\ldots \partial^{k_m}_{u_m} R^{(1)}(z)r^{-\frac12}+\partial^{k_1}_{u_1}\ldots \partial^{k_m}_{u_m} R^{(2)}(z)r^{-1}
	+\Boh((\log r)^{k_1+\cdots+k_m}r^{-\frac32}), \quad r\to +\infty,
\end{multline}
uniformly for $z\in \mathbb{C}\setminus \Gamma_R$.
\end{corollary}

\section{Asymptotic analysis of the RH problem for $X$ as $r\to 0^{+}$}\label{sec: asymptotic analysis of RH small r}
In this section, we perform the asymptotic analysis of the RH problem for $X$ as $r\to 0^{+}$, 
which is much simpler than the case of $r\to +\infty$. 
As $r\to 0^{+}$, the interval $(-rx_m, rx_m)$ shrinks to the origin. 
Comparing the RH problem for $M$ in Section \ref{sec:tacnode-RH} and the RH problem for $X$ in Section \ref{sec:differential identities}, 
we see that $X$ could be approximated by $M$ for $z$ bounded away from the $(-rx_m, rx_m)$ as $r\to 0^{+}$.
That is, let $\varrho>2rx_m$ be a fixed constant, then $X$ is approximated by the following global parametrix
for $|z|>\varrho$:
\begin{equation}\label{eq:breveN}
\breve{N}(z) = M(z) \begin{cases}
J_1(z), \qquad & \textrm{$\arg z < \varphi$ and $\arg (z-rx_m) > \varphi$,} \\
J_5(z)^{-1}, \qquad & \textrm{$\arg z >- \varphi$ and $\arg (z-rx_m) <- \varphi$,} \\
J_2(z), \qquad & \textrm{$\arg z > \pi - \varphi$ and $\arg (z+rx_m) < \pi -\varphi$,}\\
J_4(z)^{-1}, \qquad & \textrm{$\arg z <  \varphi-\pi$ and $\arg (z+rx_m) >\varphi-\pi$,}\\
I, \qquad & \textrm{elsewhere,}
\end{cases}
\end{equation}
where $M$ is the solution of the tacnode RH problem, $\varphi$ is given in \eqref{phi} and $J_k$, $k=0,\ldots, 5$, denotes the jump matrix of $M$ on the ray $\Gamma_k$ as shown in Figure \ref{fig:tacnode}.

As $z$ lies in a neighborhood of $(-rx_m, rx_m)$, i.e., $|z|<\varrho$, 
we approximate $X$ by the following local parametrix
\begin{multline}\label{eq:breveP0}
\breve P^{(0)}(z) = \widehat M(z) 
\left(I-\sum_{k=1}^{m}\mathfrak{s}_k\left(\log(z-rx_k)-\log(z+rx_k)\right)
\begin{pmatrix}
0 & 0 & 1 & 1\\
0 & 0 & 1 & 1\\
0 & 0 & 0 & 0\\
0 & 0 & 0 & 0
\end{pmatrix}\right)\\
\times
\begin{cases}
J_1(z)^{-1}, \quad & z \in \Omega_1^{(r)},\\
I, \quad & z \in \Omega_2^{(r)},\\
J_2(z)^{-1}, \quad & z \in \Omega_3^{(r)},\\
J_1(z)^{-1}J_0(z)^{-1}J_5(z)^{-1}J_4(z), \quad & z \in \Omega_4^{(r)},\\
J_1(z)^{-1}J_0(z)^{-1}J_5(z)^{-1}, \quad & z \in \Omega_5^{(r)},\\
J_1^{-1}(z)J_0^{-1}(z), \quad & z \in \Omega_6^{(r)}.
\end{cases}
\end{multline}
Here $\widehat{M}$ denotes the analytic continuation of the restriction of $M$ to the sector bounded by the rays $\Gamma_1$ and $\Gamma_2$; see Figure \ref{fig:tacnode}. The regions $\Omega_k^{(r)}$, $k=1,\ldots,6$, are shown in Figure \ref{fig:X}. We take the principal branches of the logarithm and recall that $\mathfrak{s}_j:=(s_{j+1}-s_j)/2\pi \ii$ and $s_{m+1}=1$.

It is straightforward to verify that $\breve P^{(0)}$ satisfies the 
same jump conditions as $X$ on $\Gamma_X \cap \{z: |z|<\varrho\}$, where $\Gamma_X$ is defined in \eqref{def:gammaX}.
The only nontrivial jump condition is on $(-rx_m, rx_m)$ due to 
the logarithmic terms introducing a discontinuity across this interval.
Indeed, since $\log(z-rx_j)_{+} - \log(z-rx_j)_{-} = 2\pi \ii$ for $z\in (-rx_j, rx_j)$, 
it follows that
\begin{multline}\label{eq: calculation jump condition breveP0}
\breve P^{(0)}_{-}(z)^{-1}\breve P^{(0)}_{+}(z) 
= J_5J_0J_1
\begin{pmatrix}
	1 & 0 & s_j-1 & s_j-1\\
	0 & 1 & s_j-1 & s_j-1\\
	0 & 0 & 1 & 0\\
	0 & 0 & 0 & 1
\end{pmatrix}\\
=\begin{pmatrix}
	1 & 0 & s_j & s_j\\
	0 & 1 & s_j & s_j\\
	0 & 0 & 1 & 0\\
	0 & 0 & 0 & 1
\end{pmatrix}
=J_X(z), \quad  z\in (-rx_j,-rx_{j-1})\cup(rx_{j-1}, rx_j), \quad j=1,\ldots,m.
\end{multline} 
The local behaviors can also be shown to 
agree with that of $X$ given in \eqref{eq:X-near-rx_j}
and \eqref{eq:X-near--rx_j} directly. Therefore
$\breve P^{(0)}$ is indeed a local parametrix for $X$ in the neighborhood of the origin.
Furthermore, as $r\to 0^{+}$, the logarithmic terms in \eqref{eq:breveP0} are of order $\Boh(r)$, and 
thus we have the matching condition
\begin{equation}\label{eq: matching condition breveP0}
	\breve{P}^{(0)}(z)\breve{N}(z)^{-1} = I + \Boh(r), \quad \text{uniformly for } |z|=\varrho.
\end{equation}
Define the final transformation
\begin{equation}\label{def:breveR}
\breve R(z) = 
\begin{cases}
X(z) \breve{N}(z)^{-1}, \qquad & |z|>\varrho,\\
X(z) \breve P^{(0)}(z)^{-1}, \qquad & |z|<\varrho.
\end{cases}
\end{equation}
Then $\breve R$ satisfies the following RH problem.
\begin{paragraph}{RH problem for $\breve R$}
	\begin{itemize}
	\item[(a)] $\breve R(z)$ is analytic in $\mathbb{C}\setminus \{z: |z|=\varrho\}$.
	\item[(b)] For $z\in\{z: |z|=\varrho\}$, $\breve R_+(z)=\breve R_{-}(z)J_{\breve R}(z)$,
		where
		\begin{equation}
			J_{\breve R}(z) = \breve P^{(0)}(z)\breve{N}(z)^{-1}.
		\end{equation}
	\item[(c)] As $z\to \infty$, $\breve R(z) = I + \Boh(z^{-1})$.
	\end{itemize}
\end{paragraph}
Indeed, it follows from \eqref{eq: calculation jump condition breveP0} that 
$\breve{R}$ is analytic in $[-rx_m, rx_m]\setminus\cup_{j=1}^{m} \{ \pm rx_j \}$.
Furthermore, using \eqref{eq:X-near-rx_j}, \eqref{eq:X-near--rx_j}
as well as \eqref{eq:breveP0}, we can verify that those discontinuities are removable singularities.
Thus $\breve{R}$ is analytic in $\mathbb{C}\setminus \{z: |z|=\varrho\}$.
By \eqref{eq: matching condition breveP0}, we have that as $r\to 0^{+}$,
\begin{equation}
	J_{\breve R}(z) = I + \Boh(r), \quad \text{uniformly for } |z|=\varrho.
\end{equation}
Therefore, the RH problem for $\breve R$ is a small-norm RH problem, and we have
\begin{equation}\label{eq: breveR and derivative breveR rto0}
	\breve R(z) = I + \Boh(r), \quad \frac{\dif}{\dif z}\breve R(z) = \Boh(r), \quad r\to 0^{+},
\end{equation}
uniformly for $z\in \mathbb{C}\setminus \{z: |z|=\varrho\}$.

\section{Proofs of the main results}\label{sec: proofs of main results}
In this section, we present the proofs of Proposition \ref{prop:coupled system-asy}, Theorem \ref{thm: Hamiltonian and its asymptotics} and Theorem \ref{thm: large_gap} 
by using the asymptotic results obtained in Sections \ref{sec: asymptotic analysis of RH large r} and \ref{sec: asymptotic analysis of RH small r}.

\subsection{Proof of Proposition \ref{prop:coupled system-asy}}
The aim of this section is to derive the large and small $r$ asymptotics of the 
solutions to the system \eqref{eq:coupled system} and \eqref{eq:extra-condition-1} appearing in Proposition \ref{prop:coupled system-asy}. 
These asymptotics are crucial for establishing the large gap asymptotics of $F(r\vec x, \vec u)$.

\paragraph{Large $r$ asymptotics of $p_5$, $q_5$, $p_6$, $q_6$, $p_{j,k}$ and $q_{j,k}$.}
We first derive the large $r$ asymptotics of $p_5$, $p_6$, $q_5$, $q_6$ by using 
\eqref{def:p56} and \eqref{def:q56}. 
By tracing back a series of transformations from $X\to T\to S\to R$ in \eqref{def:X to T}, \eqref{def:T to S} and \eqref{def:R},
we obtain that for $z$ outside the union of disks $\bigcup_{j=1}^{m} (D_{-x_j} \cup D_{x_j}) \cup D_0$,
\begin{multline}
	X(rz)=\diag\left(r^{-\frac{1}{4}}, r^{-\frac{1}{4}}, r^{\frac{1}{4}}, r^{\frac{1}{4}}\right)R(z)N(z)\\
	\times \diag\left(e^{-\theta_1(rz)+\tau rz}, e^{-\theta_2(rz)-\tau rz}, e^{\theta_1(rz)+\tau rz}, e^{\theta_2(rz)-\tau rz}\right).
\end{multline}
As $z\to +\infty$, comparing the coefficient of $\Boh(1/z)$ on both sides by using \eqref{eq:asyX}, \eqref{eq:asy N at infinity-2} and \eqref{eq: R at z=infinity},
it is deduced that
\begin{equation}\label{eq: X1 in terms of R1 and N1}
	X_1=r\diag\left(r^{-\frac{1}{4}}, r^{-\frac{1}{4}}, r^{\frac{1}{4}}, r^{\frac{1}{4}}\right)\left(R_1+N_1\right)\diag\left(r^{\frac{1}{4}}, r^{\frac{1}{4}}, r^{-\frac{1}{4}}, r^{-\frac{1}{4}}\right).
\end{equation}
By \eqref{def:p56}, \eqref{def:q56} and the symmetries \eqref{eq:symmX11}--\eqref{eq:symmX12}, 
we see that
\begin{align}\label{eq: q5p6 using symmetry}
q_{5}(r)= -(\dot X_1)_{11}-(X_1)_{11}, \quad \text{ and } \quad  p_{6}(r)= \ii \vr_1 (\dot X_1)_{12} +\ii \vr_2 (\widetilde X_1)_{12}.
\end{align}
Thus, it is sufficient to evaluate the first row of $X_1$. We first focus on 
the first row of $R_1$. Using \eqref{eq: R at z=infinity}, \eqref{eq: R asymptotics as r to +infinity}, \eqref{eq: R(1)_definition} and \eqref{eq: R(2)_definition}, we have
\begin{equation}
	R_1=r^{-1/2}\underset{\xi=0}{\res} J^{(1)}(\xi)+r^{-1}\underset{\xi=0}{\res} J^{(2)}(\xi)+r^{-1}\underset{\xi=0}{\res}(R^{(1)}_{-}(\xi)J^{(1)}(\xi))+\Boh(r^{-3/2}), \quad r\to +\infty.
\end{equation}
Straightforward calculations by using 
\eqref{eq: def J^(1) and J^(2)}, \eqref{eq:hatE0}, \eqref{eq: mathcalE_0(z) asy near z=0}, \eqref{eq: mathcalE_0(0)} and \eqref{eq: hatM1, hatM2} yield
\begin{equation}
	\begin{aligned}
		&\left(\res_{\xi=0} J^{(1)}(\xi)\right)_{11}=-2\sum\limits_{j=1}^m \beta_jx_j^{-\frac12}(M_1)_{13},\quad \left(\res_{\xi=0} J^{(1)}(\xi)\right)_{12}=2\sum\limits_{j=1}^m \beta_jx_j^{-\frac12}(M_1)_{14},\\
		&\left(\res_{\xi=0} J^{(1)}(\xi)\right)_{13}=(M_1)_{13}, \quad \left(\res_{\xi=0} J^{(1)}(\xi)\right)_{14}=(M_1)_{14},
	\end{aligned}
\end{equation}
and
\begin{equation}
	\begin{aligned}
		\left(\res_{\xi=0} J^{(2)}(\xi)\right)_{11} &= \left(1-4\left(\sum\limits_{j=1}^m \beta_jx_j^{-\frac12}\right)^2\right)(M_1)_{11}+4\left(\sum\limits_{j=1}^m \beta_jx_j^{-\frac12}\right)^2(M_1)_{33}, \\
		\left(\res_{\xi=0} J^{(2)}(\xi)\right)_{12} &= \left(1-4\left(\sum\limits_{j=1}^m \beta_jx_j^{-\frac12}\right)^2\right)(M_1)_{12}-4\left(\sum\limits_{j=1}^m \beta_jx_j^{-\frac12}\right)^2(M_1)_{34}, \\
		\left(\res_{\xi=0} J^{(2)}(\xi)\right)_{13} &= 2\sum\limits_{j=1}^m \beta_jx_j^{-\frac12}((M_1)_{11}-(M_1)_{33}), \\
		\left(\res_{\xi=0} J^{(2)}(\xi)\right)_{14} &= -2\sum\limits_{j=1}^m \beta_jx_j^{-\frac12}((M_1)_{12}+(M_1)_{34}),
	\end{aligned}
\end{equation}
and
\begin{equation}
	\begin{aligned}
		\left(\res_{\xi=0}(R^{(1)}_{-}(\xi)J^{(1)}(\xi))\right)_{11} &= 4\left(\sum\limits_{j=1}^m \beta_jx_j^{-\frac12}\right)^2 \big(\left((M_1)_{13}\right)^2 + (M_1)_{14}(M_1)_{23}\big), \\
		\left(\res_{\xi=0}(R^{(1)}_{-}(\xi)J^{(1)}(\xi))\right)_{12} &= -4\left(\sum\limits_{j=1}^m \beta_jx_j^{-\frac12}\right)^2 \big((M_1)_{13}(M_1)_{14} + (M_1)_{14}(M_1)_{24}\big), \\
		\left(\res_{\xi=0}(R^{(1)}_{-}(\xi)J^{(1)}(\xi))\right)_{13} &= -2\sum\limits_{j=1}^m \beta_jx_j^{-\frac12} \big(\left((M_1)_{13}\right)^2 + (M_1)_{14}(M_1)_{23}\big), \\
		\left(\res_{\xi=0}(R^{(1)}_{-}(\xi)J^{(1)}(\xi))\right)_{14} &= -2\sum\limits_{j=1}^m \beta_jx_j^{-\frac12} \big((M_1)_{13}(M_1)_{14} + (M_1)_{14}(M_1)_{24}\big).
	\end{aligned}
\end{equation}
Thus, we have
\begin{align}
	(R_1)_{11} &= -2r^{-\frac{1}{2}}\sum\limits_{j=1}^m \beta_jx_j^{-\frac{1}{2}}(M_1)_{13} \nonumber \\
	&\quad + r^{-1}\Bigg[ \left(1-4\left(\sum\limits_{j=1}^m \beta_jx_j^{-\frac{1}{2}}\right)^2\right)(M_1)_{11} \nonumber \\
	&\qquad \quad +4\left(\sum\limits_{j=1}^m \beta_jx_j^{-\frac{1}{2}}\right)^2\Big((M_1)_{33} + \left((M_1)_{13}\right)^2 + (M_1)_{14}(M_1)_{23}\Big) \Bigg] + \Boh(r^{-3/2}), \\
	(R_1)_{12} &= 2r^{-\frac{1}{2}}\sum\limits_{j=1}^m \beta_jx_j^{-\frac{1}{2}}(M_1)_{14} \nonumber \\
	&\quad + r^{-1}\Bigg[ \left(1-4\left(\sum\limits_{j=1}^m \beta_jx_j^{-\frac{1}{2}}\right)^2\right)(M_1)_{12} \nonumber \\
	&\qquad \quad -4\left(\sum\limits_{j=1}^m \beta_jx_j^{-\frac{1}{2}}\right)^2\Big((M_1)_{34} + (M_1)_{13}(M_1)_{14} + (M_1)_{14}(M_1)_{24}\Big) \Bigg] + \Boh(r^{-3/2}), \\
	(R_1)_{13} &= r^{-\frac{1}{2}}(M_1)_{13} + 2r^{-1}\sum\limits_{j=1}^m \beta_jx_j^{-\frac{1}{2}} \Big((M_1)_{11}-(M_1)_{33} - \left((M_1)_{13}\right)^2 - (M_1)_{14}(M_1)_{23}\Big) + \Boh(r^{-3/2}), \\
	(R_1)_{14} &= r^{-\frac{1}{2}}(M_1)_{14} - 2r^{-1}\sum\limits_{j=1}^m \beta_jx_j^{-\frac{1}{2}} \Big((M_1)_{12}+(M_1)_{34} + (M_1)_{13}(M_1)_{14} + (M_1)_{14}(M_1)_{24}\Big) + \Boh(r^{-3/2}).
\end{align}
Together with \eqref{eq: X1 in terms of R1 and N1} and the first row of $N_1$ explicitly given by \eqref{eq: N_1}, we substitute the expressions of $R_1$ into $X_1$ to obtain the first row of $X_1$:
\begin{align}
	(X_1)_{11} &= 2r\left(\sum_{j=1}^{m}\beta_j x_{j}^{\frac12}\right)^2  -2r^{\frac{1}{2}}\sum\limits_{j=1}^m \beta_jx_j^{-\frac{1}{2}}(M_1)_{13} \nonumber \\
	&\quad + \Bigg[ \left(1-4\left(\sum\limits_{j=1}^m \beta_jx_j^{-\frac{1}{2}}\right)^2\right)(M_1)_{11} \nonumber \\
	&\qquad \quad +4\left(\sum\limits_{j=1}^m \beta_jx_j^{-\frac{1}{2}}\right)^2\Big((M_1)_{33} + \left((M_1)_{13}\right)^2 + (M_1)_{14}(M_1)_{23}\Big) \Bigg] + \Boh(r^{-1 / 2}), \label{eq:X1_11-large r}\\
	(X_1)_{12} &= 2r^{\frac{1}{2}}\sum\limits_{j=1}^m \beta_jx_j^{-\frac{1}{2}}(M_1)_{14} \nonumber \\
	&\quad + \Bigg[ \left(1-4\left(\sum\limits_{j=1}^m \beta_jx_j^{-\frac{1}{2}}\right)^2\right)(M_1)_{12} \nonumber \\
	&\qquad \quad -4\left(\sum\limits_{j=1}^m \beta_jx_j^{-\frac{1}{2}}\right)^2\Big((M_1)_{34} + (M_1)_{13}(M_1)_{14} + (M_1)_{14}(M_1)_{24}\Big) \Bigg] + \Boh(r^{-1 / 2}), \label{eq:X1_12-large r}\\
	(X_1)_{13} &= -2r^{\frac{1}{2}}\sum_{j=1}^{m}\beta_j x_{j}^{\frac12} + (M_1)_{13} \nonumber \\
	&\quad + 2r^{-\frac{1}{2}}\sum\limits_{j=1}^m \beta_jx_j^{-\frac{1}{2}} \Big((M_1)_{11}-(M_1)_{33} - \left((M_1)_{13}\right)^2 - (M_1)_{14}(M_1)_{23}\Big) + \Boh(r^{-1}), \label{eq:X1_13-large r}\\
	(X_1)_{14} &= (M_1)_{14} - 2r^{-\frac{1}{2}}\sum\limits_{j=1}^m \beta_jx_j^{-\frac{1}{2}} \Big((M_1)_{12}+(M_1)_{34} + (M_1)_{13}(M_1)_{14} + (M_1)_{14}(M_1)_{24}\Big) + \Boh(r^{-1}). \label{eq:X1_14-large r}
\end{align}
Substituting the above expressions into \eqref{def:p56}, \eqref{def:q56} 
and \eqref{eq: q5p6 using symmetry}, we obtain the large $r$ asymptotics of $q_5$, $q_6$, $p_5$ and $p_6$
stated in \eqref{q5-large r}, \eqref{q6-large r}, \eqref{p5-large r}, \eqref{p6-large r} respectively.

Now we turn to the large $r$ asymptotics of $p_{j,k}$ and $q_{j,k}$, $j=1,\ldots,m$, $k=1,2,3,4$ by using 
\eqref{def:qkpk}. It follows from \eqref{eq:X-near-rx_j} and \eqref{eq: Phi-expand-rx_j} that for $z\in\Omega_2^{(1)}$,
\begin{equation}\label{eq: X_Rj0(r) limit representation}
	X_{R,j}^{(0)}(r)=\lim_{\substack{z\to x_j \\ z\in \Omega_{2}^{(1)}}}X(rz)
	\begin{pmatrix}
	1 & 0 & \mathfrak{s}_j \log(rz-rx_j) & \mathfrak{s}_j \log(rz-rx_j)
   \\
	0 & 1 & \mathfrak{s}_j \log(rz-rx_j)& \mathfrak{s}_j \log(rz-rx_j)
    \\
	0 & 0 & 1 & 0
    \\
    0 & 0 & 0 & 1	
\end{pmatrix}.
\end{equation}
For $z\in D_{x_j}$ but outside the lens-shaped region, we obtain from \eqref{def:X to T}, \eqref{def:T to S}
and \eqref{def:R} that
\begin{multline}\label{eq: X(rz) after tracing transformations locally but outside lenses}
	X(rz)=\diag\left(r^{-\frac{1}{4}}, r^{-\frac{1}{4}}, r^{\frac{1}{4}}, r^{\frac{1}{4}}\right)R(z)P^{(x_j)}(z)\\
	\times \diag\left(e^{-\theta_1(rz)+\tau rz}, e^{-\theta_2(rz)-\tau rz}, e^{\theta_1(rz)+\tau rz}, e^{\theta_2(rz)-\tau rz}\right).
\end{multline}
This, together with the explicit expressions of $P^{(x_j)}$ given in \eqref{def:P+x_j}--\eqref{def:G+x_j}, 
yield that 
\begin{align}\label{eq: X_Rj0 in terms of R and mathcal Exj}
	X_{R,j}^{(0)}(r)&=\diag\left(r^{-\frac{1}{4}}, r^{-\frac{1}{4}}, r^{\frac{1}{4}}, r^{\frac{1}{4}}\right)R(x_j)\mathcal E_{x_j}(x_j) \nonumber \\
	&\quad \times\lim_{\substack{z \to x_j \\ z \in \Omega_2^{(1)}} }
	\left[\begin{pmatrix}
        \left(\Phi_{\CH}\right)_{11}\left(r^{\frac32}f_{x_j}(z);\beta_j\right) & 0 & \left(\Phi_{\CH}\right)_{12}\left(r^{\frac32}f_{x_j}(z);\beta_j\right) & 0 \\
        0 & 1 & 0 & 0 \\
        \left(\Phi_{\CH}\right)_{21}\left(r^{\frac32}f_{x_j}(z);\beta_j\right) & 0 & \left(\Phi_{\CH}\right)_{22}\left(r^{\frac32}f_{x_j}(z);\beta_j\right) & 0 \\
        0 & 0 & 0 & 1
    \end{pmatrix}\nonumber\right.\\
	&\quad \times (s_js_{j+1})^{-\frac14(E_{11}-E_{33})}e^{\theta_1(rz)(E_{11}-E_{33})}\begin{pmatrix}
	1 & 0 & 0 & 0\\
	-e^{\theta_1(rz)-\theta_2(rz)-2 \tau rz} & 1 & 0 & 0\\
	0 & 0 & 1 & e^{\theta_1(rz)-\theta_2(rz)+2 \tau rz}\\
	0 & 0 & 0 & 1
	\end{pmatrix}\nonumber\\
	&\quad \times  \diag \left(e^{-\theta_1(rz)+\tau rz}, e^{-\theta_2(rz)-\tau rz}, e^{\theta_1(rz)+\tau rz}, e^{\theta_2(rz)-\tau rz} \right)\nonumber\\
	&\quad \times \left.\begin{pmatrix}
	1 & 0 & \mathfrak{s}_j \log(rz-rx_j) & \mathfrak{s}_j \log(rz-rx_j)\\
	0 & 1 & \mathfrak{s}_j \log(rz-rx_j) & \mathfrak{s}_j \log(rz-rx_j)\\
	0 & 0 & 1 & 0\\
	0 & 0 & 0 & 1
	\end{pmatrix}\right],
\end{align}
where $\mathcal E_{x_j}$ and $f_{x_j}$ are defined in \eqref{def:mathcal E+xj} and \eqref{def: f-xj}, respectively.
Recalling the definition of $\beta_j$ in \eqref{def: second def beta_j}, we note that
\begin{equation}\label{eq: beta_j and s_j relation}
	\frac{1-e^{2\pi \ii \beta_j}}{2\pi\ii}=\frac{\mathfrak{s}_j}{s_{j+1}}.
\end{equation}
Therefore, using \eqref{eq:H-expand-2}, we can rewrite \eqref{eq: X_Rj0 in terms of R and mathcal Exj} as 
\begin{multline}\label{eq: X_Rj0 after limit to x_j}
	X_{R,j}^{(0)}(r)=\diag\left(r^{-\frac{1}{4}}, r^{-\frac{1}{4}}, r^{\frac{1}{4}}, r^{\frac{1}{4}}\right)R(x_j)\mathcal E_{x_j}(x_j)\Phi_{\CH}^{(0)\#}(\beta_j)\\
	\times \begin{pmatrix}
		s_{j+1}^{-\frac12}e^{\tau r x_j} & 0 & \substack{e^{\tau r x_j} s_{j+1}^{-\frac12}\mathfrak{s}_j \\ \times \big(\log r - \log (e^{-\frac{\pi \ii}{2}} r^{\frac 32} f_{x_j}'(x_j))\big) }& \substack{s_{j+1}^{-\frac12}e^{\tau r x_j} \mathfrak{s}_j \\ \times\big(\log r - \log (e^{-\frac{\pi \ii}{2}} r^{\frac 32} f_{x_j}'(x_j))\big)} \\
		-e^{-\theta_2(rx_j)-\tau rx_j} & e^{-\theta_2(rx_j)-\tau rx_j} & 0 & 0 \\
		0 & 0 & s_{j+1}^{\frac12}e^{\tau r x_j} & s_{j+1}^{\frac12}e^{\tau r x_j} \\
		0 & 0 & 0 & e^{\theta_2(rx_j)-\tau rx_j}
	\end{pmatrix},
\end{multline}
where
\begin{align}\label{def: Phi_CH^(0)sharp}
	\Phi_{\CH}^{(0)\#}(\beta_j)=
	\begin{pmatrix}
		\left(\Phi^{(0)}_{\CH}(\beta_j)\right)_{11} & 0 & \left(\Phi_{\CH}^{(0)}(\beta_j)\right)_{12}& 0 \\
        0 & 1 & 0 & 0 \\
        \left(\Phi^{(0)}_{\CH}(\beta_j)\right)_{21} & 0 & \left(\Phi_{\CH}^{(0)}(\beta_j)\right)_{22} & 0 \\
        0 & 0 & 0 & 1
	\end{pmatrix}
\end{align}
with $\Phi_{\CH}^{(0)}(\beta_j)$ defined in \eqref{eq:H-expand-coeff-0}.

Substituting \eqref{eq: X_Rj0 after limit to x_j} into \eqref{def:qkpk}, we obtain
\begin{align}
\begin{pmatrix}
q_{j,1}(r)\\q_{j,2}(r)\\q_{j,3}(r)\\q_{j,4}(r)
\end{pmatrix}&=s_{j+1}^{-\frac12}e^{\tau rx_j} \diag\left(r^{-\frac{1}{4}}, r^{-\frac{1}{4}}, r^{\frac{1}{4}}, r^{\frac{1}{4}}\right)R(x_j)\mathcal E_{x_j}(x_j)\Phi_{\CH}^{(0)\#}(\beta_j)
\begin{pmatrix}
1\\0\\0\\0
\end{pmatrix},\label{eq: q_j_vector final cal}\\
\begin{pmatrix}
p_{j,1}(r)\\p_{j,2}(r)\\p_{j,3}(r)\\p_{j,4}(r)
\end{pmatrix}&=-\mathfrak{s}_js_{j+1}^{-\frac12}e^{-\tau rx_j} \diag\left(r^{\frac{1}{4}}, r^{\frac{1}{4}}, r^{-\frac{1}{4}}, r^{-\frac{1}{4}}\right)R(x_j)^{-\rm T}\mathcal E_{x_j}(x_j)^{-\rm T}\Phi_{\CH}^{(0)\#}(\beta_j)^{-\rm T}
\begin{pmatrix}
0\\0\\1\\0
\end{pmatrix}.\label{eq: p_j_vector final cal}
\end{align}
It is also noted from \eqref{eq: R asymptotics as r to +infinity} and \eqref{eq: R(1)_definition} that as $r\to +\infty$,
\begin{multline}\label{eq: R(x_j) as r to +infinity}
R(x_j)=
\\
I+\frac{r^{-\frac{1}{2}}}{x_j}\begin{pmatrix}
-2\sum\limits_{k=1}^m \beta_kx_k^{-\frac12}(M_1)_{13} & 2\sum\limits_{k=1}^m \beta_kx_k^{-\frac12}(M_1)_{14} & (M_1)_{13} & (M_1)_{14} \\[2mm]
-2\sum\limits_{k=1}^m \beta_kx_k^{-\frac12}(M_1)_{23} & 2\sum\limits_{k=1}^m \beta_kx_k^{-\frac12}(M_1)_{24} & (M_1)_{23} & (M_1)_{24} \\[2mm]
-4\big(\sum\limits_{k=1}^m \beta_kx_k^{-\frac12}\big)^2(M_1)_{13} & 4\big(\sum\limits_{k=1}^m \beta_kx_k^{-\frac12}\big)^2(M_1)_{14} & 2\sum\limits_{k=1}^m \beta_kx_k^{-\frac12}(M_1)_{13} & 2\sum\limits_{k=1}^m \beta_kx_k^{-\frac12}(M_1)_{14} \\[2mm]
4\big(\sum\limits_{k=1}^m \beta_kx_k^{-\frac12}\big)^2(M_1)_{23} & -4\big(\sum\limits_{k=1}^m \beta_kx_k^{-\frac12}\big)^2(M_1)_{24} & -2\sum\limits_{k=1}^m \beta_kx_k^{-\frac12}(M_1)_{23} & -2\sum\limits_{k=1}^m \beta_kx_k^{-\frac12}(M_1)_{24}
\end{pmatrix}\\
+\Boh(r^{-1}).
\end{multline}
We illustrate the computation with $q_{j,1}(r)$. Combining \eqref{eq: q_j_vector final cal}, \eqref{eq: R(x_j) as r to +infinity} and \eqref{eq:E_{x_j}(x_j)}, we obtain
\begin{equation}
q_{j,1}(r)=e^{\tau r x_j}
	\widehat{Q}_{j,1}(r)r^{-\frac14}\left(1+\Boh\left(r^{-\frac12}\right)\right),
\end{equation}
where 
\begin{align}
	\widehat{Q}_{j,1}(r)
	&:=\frac{e^{-\frac12\sum\limits_{\ell=j+1}^{m}u_{\ell}-\frac{u_j}{4}+\frac{\pi\ii}{4}}}{\sqrt2\,x_j^{\frac14}} \nonumber\\
	&\times\Bigg[
        e^{-\frac{2\ii}{3}\vr_1r^{\frac32}x_j^{\frac32}+2\ii\vs_1r^{\frac12}x_j^{\frac12}}
        \left(8r^{\frac12}x_j^{\frac12}(r\vr_1x_j-\vs_1)\right)^{\frac{u_j}{2\pi\ii}}
        \prod_{\substack{k=1\\k\neq j}}^{m}
        \left(\frac{{x}_{j}^{\frac12}+{x}_{k}^{\frac12}}{|{x}_{j}^{\frac12}-{x}_{k}^{\frac12}|}\right)^{\frac{u_k}{2\pi\ii}}
        \Gamma\left(1-\frac{u_j}{2\pi\ii}\right)\nonumber\\
        &\hspace{-1em}
        -\ii e^{\frac{2\ii}{3}\vr_1r^{\frac32}x_j^{\frac32}-2\ii\vs_1r^{\frac12}x_j^{\frac12}}
        \left(8r^{\frac12}x_j^{\frac12}(r\vr_1x_j-\vs_1)\right)^{-\frac{u_j}{2\pi\ii}}
        \prod_{\substack{k=1\\k\neq j}}^{m}
        \left(\frac{{x}_{j}^{\frac12}+{x}_{k}^{\frac12}}{|{x}_{j}^{\frac12}-{x}_{k}^{\frac12}|}\right)^{-\frac{u_k}{2\pi\ii}}
        \Gamma\left(1+\frac{u_j}{2\pi\ii}\right)
        \Bigg]. \label{eq: Q_j_1 intermi}
\end{align}
To simplify \eqref{eq: Q_j_1 intermi}, set
\begin{align}\label{def: Z_j(r)}
Z_j(r):&=e^{-\frac{2\ii}{3}\vr_1r^{\frac32}x_j^{\frac32}+2\ii\vs_1r^{\frac12}x_j^{\frac12}}
        \left(8r^{\frac12}x_j^{\frac12}(r\vr_1x_j-\vs_1)\right)^{\frac{u_j}{2\pi\ii}}
        \prod_{\substack{k=1\\k\neq j}}^{m}
        \left(\frac{{x}_{j}^{\frac12}+{x}_{k}^{\frac12}}{|{x}_{j}^{\frac12}-{x}_{k}^{\frac12}|}\right)^{\frac{u_k}{2\pi\ii}}
        \Gamma\left(1-\frac{u_j}{2\pi\ii}\right).
\end{align}
Since $u_j\in\mathbb R$, the quantity $u_j/(2\pi\ii)$ is purely imaginary, and hence
\begin{equation}
\overline{\Gamma\left(1-\frac{u_j}{2\pi\ii}\right)}=\Gamma\left(1+\frac{u_j}{2\pi\ii}\right).
\end{equation}
For sufficiently large $r$, one has $8r^{\frac12}x_j^{\frac12}(r\vr_1x_j-\vs_1)>0$. These two facts imply that the second term in the square bracket in \eqref{eq: Q_j_1 intermi} is exactly $-\ii\overline{Z_j(r)}$. Therefore,
\begin{align*}
    \widehat{Q}_{j,1}(r)= \frac{e^{-\frac12\sum\limits_{\ell=j+1}^{m}u_{\ell}-\frac{u_j}{4}}}{\sqrt2\,x_j^{\frac14}} \left( e^{\frac{\pi\ii}{4}} Z_j(r) + e^{-\frac{\pi\ii}{4}} \overline{Z_j(r)} \right)
    = \frac{e^{-\frac12\sum\limits_{\ell=j+1}^{m}u_{\ell}-\frac{u_j}{4}}}{\sqrt2\,x_j^{\frac14}} 2 \re \left( e^{-\frac{\pi\ii}{4}} \overline{Z_j(r)} \right).
\end{align*}
Writing $\overline{Z_j(r)}$ in polar form, we obtain
\begin{equation*}
\overline{Z_j(r)} = \left|\Gamma\left(1+\frac{u_j}{2\pi\ii}\right)\right| e^{\ii \vartheta_j(r)},
\end{equation*}
where $\vartheta_j(r)$ is defined in \eqref{eq:def-theta-j-general}. Thus,
\begin{align}
    \widehat{Q}_{j,1}(r) &= \frac{e^{-\frac12\sum\limits_{\ell=j+1}^{m}u_{\ell}-\frac{u_j}{4}}}{\sqrt2\,x_j^{\frac14}} 2 \left|\Gamma\left(1+\frac{u_j}{2\pi\ii}\right)\right| \cos\left(\vartheta_j(r)-\frac{\pi}{4}\right),
\end{align}
which leads to \eqref{qj1-large r}.

The large $r$ asymptotics of the remaining $q_{j,k}(r)$ and $p_{j,k}(r)$ in \eqref{qj2-large r}--\eqref{qj4-large r} and \eqref{pj1-large r}--\eqref{pj4-large r} are obtained analogously from \eqref{eq: q_j_vector final cal}, \eqref{eq: p_j_vector final cal}, \eqref{eq: R(x_j) as r to +infinity} and \eqref{eq:E_{x_j}(x_j)}. We omit the straightforward calculations.

\paragraph{Small $r$ asymptotics of $p_5$, $p_6$, $q_5$, $q_6$, $p_{j,k}$ and $q_{j,k}$.}
The small $r$ asymptotics of $p_5$, $p_6$, $q_5$, $q_6$, $p_{j,k}$, and $q_{j,k}$ follow from the asymptotic analysis in Section \ref{sec: asymptotic analysis of RH small r}.

Using \eqref{def:breveR} and \eqref{eq: breveR and derivative breveR rto0}, we obtain, for $|z|>\varrho$,
\begin{equation}
X(z)=\left(I+\Boh(r)\right)\breve{N}(z).
\end{equation}
Together with \eqref{eq:asy:M} and \eqref{eq:asyX}, this implies that
\begin{equation}
	X_1=M_1\left(I+\Boh(r)\right).
\end{equation}
A direct calculation based on \eqref{def:p56}, \eqref{def:q56} and \eqref{eq: q5p6 using symmetry} then yields \eqref{q5-small r}--\eqref{q6-small r} and \eqref{p5-small r}--\eqref{p6-small r} as $r\to 0^+$.

We now compute the small $r$ asymptotics of $p_{j,k}$ and $q_{j,k}$, $j=1,\ldots,m$, $k=1,2,3,4$.
By \eqref{eq:breveP0} and \eqref{def:breveR}, for $z\in\Omega_2^{(1)}\cap\{z: r|z|<\varrho\}$,
\begin{equation}
	X(rz)=\breve{R}(rz)\widehat M(rz)\left(I-\sum_{j=1}^{m}\mathfrak{s}_j\left(\log(rz-rx_j)-\log(rz+rx_j)\right)
\begin{pmatrix}
0 & 0 & 1 & 1\\
0 & 0 & 1 & 1\\
0 & 0 & 0 & 0\\
0 & 0 & 0 & 0
\end{pmatrix}\right).
\end{equation}
Hence, by \eqref{eq: X_Rj0(r) limit representation},
\begin{equation}
	X_{R,j}^{(0)}(r)=\breve{R}(rx_j)\widehat M(rx_j)\left(I+\sum_{j=1}^{m}\mathfrak{s}_j\log(2rx_j)
\begin{pmatrix}
0 & 0 & 1 & 1\\
0 & 0 & 1 & 1\\
0 & 0 & 0 & 0\\
0 & 0 & 0 & 0
\end{pmatrix}\right),
\end{equation}
and \eqref{eq: breveR and derivative breveR rto0} gives, as $r\to 0^+$,
\begin{equation}
	X_{R,j}^{(0)}(r)=\left(\widehat M(0)+\Boh(r)\right)\left(I+\sum_{j=1}^{m}\mathfrak{s}_j\log(2rx_j)
\begin{pmatrix}
0 & 0 & 1 & 1\\
0 & 0 & 1 & 1\\
0 & 0 & 0 & 0\\
0 & 0 & 0 & 0
\end{pmatrix}\right).
\end{equation}
Therefore, using \eqref{def:qkpk}, we obtain, as $r\to 0^+$,
\begin{align}
	&\vec{q}_j(r)=\left(\widehat M(0)+\Boh(r)\right)\begin{pmatrix}1 & 1 & 0 & 0\end{pmatrix}^{\rm T}=\Boh(1), \label{vecq small r}\\
	&\vec{p}_j(r)=-\mathfrak{s}_j\left(\widehat M(0)^{-\rm T}+\Boh(r)\right)\begin{pmatrix}0 & 0 & 1 & 1\end{pmatrix}^{\rm T}=\Boh(1), \label{vecp small r}
\end{align}
which lead to \eqref{qj1234-small r} and \eqref{pj1234-small r}.

Now we have completed the proof of Proposition \ref{prop:coupled system-asy}. \qed

\subsection{Proof of Theorem \ref{thm: Hamiltonian and its asymptotics}}
Recalling \eqref{eq:Hamiltonian-via-RH} and \eqref{eq:diff_identity}, it is readily seen that
\begin{equation*}
\frac{\dif}{\dif r}\log F(r\vec{x},\vec{u}) = H(r).
\end{equation*}
Since $F(0\vec{x},\vec u)=1$ by \eqref{eq:generating_function}, integrating the above identity from $0$ to $r$ yields \eqref{eq:integral representation via Hamiltonian}.
It remains to establish the small and large $r$ asymptotics of $H(r)$.

\paragraph{Small $r$ asymptotics of $H$.}
By Proposition \ref{prop:coupled system-asy} and \eqref{def:tildeX}, the quantities $\tilde{q}_{j,k}$, $\tilde{p}_{j,k}$, $\tilde{q}_5$, $\tilde{q}_6$, $\tilde{p}_5$, and $\tilde{p}_6$ satisfy the same estimates as $q_{j,k}$, $p_{j,k}$, $q_5$, $q_6$, $p_5$, and $p_6$, respectively, as $r\to 0^+$. 
We now estimate the terms in \eqref{eq:Hamiltonian} one by one.

The contribution multiplied by $rx_j^2$ is immediately $\Boh(r)$, since all $q_{j,k}(r)$, $p_{j,k}(r)$,
their tilded counterparts are $\Boh(1)$ as $r\to 0^+$. Next, every term in the $x_j$-bracket is a product of three factors chosen from
$q_{j,k}(r)$, $p_{j,k}(r)$, $\widetilde q_{j,k}(r)$, $\widetilde p_{j,k}(r)$, $q_5(r)$, $p_5(r)$, $q_6(r)$, $p_6(r)$,
possibly plus the constants $\tau,\vs_1,\vs_2$. Since all these quantities are $\Boh(1)$ as $r\to0^+$,
the whole $x_j$-bracket is also $\Boh(1)$. 

It remains to consider the terms with prefactor $1/r$. By \eqref{def:S_jk} and \eqref{def:Upsilon_jk},
we rewrite
\begin{align}
S_{jk}(r)&=\vec q_j(r)^{\rm T}\vec p_k(r),\\
\Upsilon_{jk}(r)&=\vec q_j(r)^{\rm T}
\begin{pmatrix}
0&1&0&0\\
1&0&0&0\\
0&0&0&-1\\
0&0&-1&0
\end{pmatrix}
\widetilde{\vec p}_k(r),
\end{align}
where $\vec{p}_j$ and $\vec{q}_j$ are defined in \eqref{def:qkpk}, 
$\widetilde {(\cdot)}$ is defined through \eqref{def:tildeX}.
Using \eqref{vecq small r} and \eqref{vecp small r}, we obtain
\begin{align}
S_{jk}(r)
=-\mathfrak{s}_k
\begin{pmatrix}1&1&0&0\end{pmatrix}
\begin{pmatrix}0\\0\\1\\1\end{pmatrix}
+\Boh(r)
=\Boh(r),
\end{align}
since the coefficient of the $\Boh(1)$ term vanishes. We next consider $\Upsilon_{jk}(r)$. By \eqref{eq:symmM1} at $z=0$,
\begin{equation}\label{eq: symmetry of M at 0}
\widehat M(0)=
\begin{pmatrix}
	\sigma_1 & 0\\
	0 & -\sigma_1
\end{pmatrix}
	\widetilde{\widehat M}(0)
\begin{pmatrix}
	\sigma_1 & 0\\
	0 & -\sigma_1
\end{pmatrix},
\qquad
\widetilde{\widehat M}(0)^{-\rm T}=\begin{pmatrix}
	\sigma_1 & 0\\
	0 & -\sigma_1
\end{pmatrix}\widehat M(0)^{-\rm T}
\begin{pmatrix}
	\sigma_1 & 0\\
	0 & -\sigma_1
\end{pmatrix},
\end{equation}
where $\sigma_1$ is defined in \eqref{def:Pauli}. Therefore, using \eqref{vecq small r}--\eqref{vecp small r},
we have
\begin{align}
\Upsilon_{jk}(r)
&=-\mathfrak{s}_k
\begin{pmatrix}1&1&0&0\end{pmatrix}
\widehat M(0)^{\rm T}
\begin{pmatrix}
	\sigma_1 & 0\\
	0 & -\sigma_1
\end{pmatrix}\widetilde{\widehat M}(0)^{-\rm T}
\begin{pmatrix}0\\0\\1\\1\end{pmatrix}
+\Boh(r)\nonumber\\
&=-\mathfrak{s}_k
\begin{pmatrix}1&1&0&0\end{pmatrix}
\begin{pmatrix}
	\sigma_1 & 0\\
	0 & -\sigma_1
\end{pmatrix}
\begin{pmatrix}0\\0\\1\\1\end{pmatrix}
+\Boh(r)
=\Boh(r).
\end{align}
Hence,
\begin{align*}
&\frac{S_{jk}(r)S_{kj}(r)+\widetilde S_{jk}(r)\widetilde S_{kj}(r)}{r(x_j-x_k)}=\Boh(r),\\
&\frac{\Upsilon_{jk}(r)\widetilde\Upsilon_{kj}(r)+\Upsilon_{kj}(r)\widetilde\Upsilon_{jk}(r)}{r(x_j+x_k)}=\Boh(r),\\
&\frac{1}{r}\Upsilon_{jj}(r)\widetilde\Upsilon_{jj}(r)=\Boh(r).
\end{align*}
Multiplying the above three estimates by the prefactor $1/r$, 
and combining them with the previous estimates, we finally get \eqref{eq:H-small-r}. 
In particular, the integral in \eqref{eq:integral representation via Hamiltonian} is well-defined at the lower endpoint $0$.

\paragraph{Large $r$ asymptotics of $H$.} 
We rewrite the Hamiltonian $H(r)$ in \eqref{eq:Hamiltonian-via-RH} as 
\begin{equation}\label{eq:Hamiltonian-via-RH-2}
	H(r)=-\sum_{j=1}^m x_j\mathfrak{s}_j
	\begin{pmatrix}0 & 0 & 1 & 1\end{pmatrix}
	\left(X_{R,j}^{(1)}(r)-X_{L,j}^{(1)}(r)\right)
	\begin{pmatrix}1 \\ 1 \\ 0 \\ 0\end{pmatrix},
\end{equation}
which is convenient for computing the large $r$ asymptotics of $H(r)$.

It follows from \eqref{eq:X-near-rx_j}--\eqref{eq: Phi-expand-rx_j} that
\begin{equation}\label{eq: for X_Rj1(r)-cal-1}
	X_{R,j}^{(1)}(r)=\frac{1}{r}X_{R,j}^{(0)}(r)^{-1}
	\lim_{\substack{z\to x_j \\ z\in \Omega_2^{(1)}}}\left[
		X(rz)\begin{pmatrix}
			1 & 0 & \mathfrak{s}_j \log(rz-rx_j) & \mathfrak{s}_j \log(rz-rx_j)\\
			0 & 1 & \mathfrak{s}_j \log(rz-rx_j) & \mathfrak{s}_j \log(rz-rx_j)\\
			0 & 0 & 1 & 0\\
			0 & 0 & 0 & 1
		\end{pmatrix}
	\right]',
\end{equation}
where $(')$ denotes the derivative with respect to $z$. 
It follows from \eqref{eq: X_Rj0 after limit to x_j} that
\begin{multline}\label{eq: for X_Rj1(r)-cal-2}
	\begin{pmatrix}
		0 & 0 & 1 & 1
	\end{pmatrix}
	X_{R,j}^{(0)}(r)^{-1}=\\
	\begin{pmatrix}
		0 & 0 & s_{j+1}^{-\frac12}e^{-\tau r x_j} & 0
	\end{pmatrix}
	\Phi_{\CH}^{(0)\#}(\beta_j)^{-1}\mathcal E_{x_j}(x_j)^{-1}R(x_j)^{-1}\diag\left(r^{\frac{1}{4}}, r^{\frac{1}{4}}, r^{-\frac{1}{4}}, r^{-\frac{1}{4}}\right).
\end{multline}
Furthermore, by \eqref{eq: X(rz) after tracing transformations locally but outside lenses} and \eqref{def:P+x_j},
\begin{align}
	&\lim_{\substack{z\to x_j \\ z\in \Omega_2^{(1)}}}
	\left[
		X(rz)\begin{pmatrix}
			1 & 0 & \mathfrak{s}_j \log(rz-rx_j) & \mathfrak{s}_j \log(rz-rx_j)\\
			0 & 1 & \mathfrak{s}_j \log(rz-rx_j) & \mathfrak{s}_j \log(rz-rx_j)\\
			0 & 0 & 1 & 0\\
			0 & 0 & 0 & 1
		\end{pmatrix}
	\right]'
	\begin{pmatrix}
		1 \\ 1 \\ 0 \\ 0
	\end{pmatrix}\nonumber\\
	&=\diag(r^{-\frac{1}{4}}, r^{-\frac{1}{4}}, r^{\frac{1}{4}}, r^{\frac{1}{4}})\nonumber\\
	&\times\lim_{\substack{z\to x_j \\ z\in \Omega_2^{(1)}}}
	\left[
		R(z)\mathcal E_{x_j}(z)
		\begin{pmatrix}
        \left(\Phi_{\CH}\right)_{11}\left(r^{\frac32}f_{x_j}(z);\beta_j\right) & 0 & \left(\Phi_{\CH}\right)_{12}\left(r^{\frac32}f_{x_j}(z);\beta_j\right) & 0 \\
        0 & 1 & 0 & 0 \\
        \left(\Phi_{\CH}\right)_{21}\left(r^{\frac32}f_{x_j}(z);\beta_j\right) & 0 & \left(\Phi_{\CH}\right)_{22}\left(r^{\frac32}f_{x_j}(z);\beta_j\right) & 0 \\
        0 & 0 & 0 & 1
	    \end{pmatrix} \nonumber\right.\\
	&\hspace*{3em}\left.
	\times (s_js_{j+1})^{-\frac14(E_{11}-E_{33})}e^{\theta_1(rz)(E_{11}-E_{33})}
	\begin{pmatrix}
	1 & 0 & 0 & 0\\
	-e^{\theta_1(rz)-\theta_2(rz)-2 \tau rz} & 1 & 0 & 0\\
	0 & 0 & 1 & e^{\theta_1(rz)-\theta_2(rz)+2 \tau rz}\\
	0 & 0 & 0 & 1
	\end{pmatrix}
	\right.\nonumber\\
	&\hspace*{3em} \left.
	\times \diag \left(e^{-\theta_1(rz)+\tau rz}, e^{-\theta_2(rz)-\tau rz}, e^{\theta_1(rz)+\tau rz}, e^{\theta_2(rz)-\tau rz} \right)
	\right]'
	\begin{pmatrix}
		1 \\ 1 \\ 0 \\ 0
	\end{pmatrix}.
\end{align}
A direct computation to the limit on the r.h.s., using also \eqref{eq: E_{x_j}(z) at x_j} and \eqref{eq:H-expand-2}, yields
\begin{align}\label{eq: for X_Rj1(r)-cal-3}
	&\lim_{\substack{z\to x_j \\ z\in \Omega_2^{(1)}}}
	\left[
		X(rz)\begin{pmatrix}
			1 & 0 & \mathfrak{s}_j \log(rz-rx_j) & \mathfrak{s}_j \log(rz-rx_j)\\
			0 & 1 & \mathfrak{s}_j \log(rz-rx_j) & \mathfrak{s}_j \log(rz-rx_j)\\
			0 & 0 & 1 & 0\\
			0 & 0 & 0 & 1
		\end{pmatrix}
	\right]'
	\begin{pmatrix}
		1 \\ 1 \\ 0 \\ 0
	\end{pmatrix}\nonumber\\
	&=\diag(r^{-\frac{1}{4}}, r^{-\frac{1}{4}}, r^{\frac{1}{4}}, r^{\frac{1}{4}})
	\left(s_{j+1}^{-\frac12}e^{\tau rx_j}R'(x_j)\mathcal E_{x_j}(x_j)\Phi_{\CH}^{(0)\#}(\beta_j)
	+s_{j+1}^{-\frac12}e^{\tau rx_j}R(x_j)\mathcal{E}_{x_j}'(x_j)\Phi_{\CH}^{(0)\#}(\beta_j)\nonumber\right.\\
	&\left.
	+s_{j+1}^{-\frac12}e^{\tau rx_j}r^{\frac32}f_{x_j}'(x_j)R(x_j)\mathcal E_{x_j}(x_j)\Phi_{\CH}^{(0)\#}(\beta_j)\Phi_{\CH}^{(1)\#}(\beta_j)
	+s_{j+1}^{-\frac12}\tau re^{\tau rx_j}R(x_j)\mathcal E_{x_j}(x_j)\Phi_{\CH}^{(0)\#}(\beta_j)
	\right)
	\begin{pmatrix}
	1 \\ 0 \\ 0 \\ 0
	\end{pmatrix},
\end{align}
where $\Phi_{\CH}^{(0)\#}(\beta_j)$ is defined in \eqref{def: Phi_CH^(0)sharp}, and $\Phi_{\CH}^{(1)\#}(\beta_j)$ is 
also a $4\times 4$ matrix defined through $\Phi_{\CH}^{(1)}(\beta_j)$ given in \eqref{eq:H-expand-2}
with the same block structure as $\Phi_{\CH}^{(0)\#}(\beta_j)$. We also remind the readers that 
the right-most vector becoming $\begin{pmatrix}1 & 0 & 0 & 0\end{pmatrix}^{\rm T}$ is due to the 
fact that we encounter the matrix structure as in \eqref{eq: X_Rj0 after limit to x_j}.
For later use, we note from \eqref{eq:H-expand-coeff-1} that
\begin{equation}\label{eq: 31entry of Phi_CH^(1)sharp}
	\left(\Phi_{\CH}^{(1)\#}(\beta_j)\right)_{31}=\frac{\beta_j\pi\ii e^{-\beta_j\pi\ii}}{\sin(\beta_j\pi)}.
\end{equation}
Combining \eqref{eq: for X_Rj1(r)-cal-1}, \eqref{eq: for X_Rj1(r)-cal-2} and \eqref{eq: for X_Rj1(r)-cal-3}, we obtain
\begin{align}\label{eq: rightpart-computate 31 entry-pre}
	&-\sum_{j=1}^m x_j\mathfrak{s}_j
	\begin{pmatrix}0 & 0 & 1 & 1\end{pmatrix}
	X_{R,j}^{(1)}(r)
	\begin{pmatrix}1 \\ 1 \\ 0 \\ 0\end{pmatrix}\nonumber\\ 
	&\quad=-\frac{1}{r}\sum_{j=1}^m x_j\mathfrak{s}_j s_{j+1}^{-1}\left[\Phi_{\CH}^{(0)\#}(\beta_j)^{-1}\mathcal{E}_{x_j}(x_j)^{-1}R(x_j)^{-1}R'(x_j)\mathcal{E}_{x_j}(x_j)\Phi_{\CH}^{(0)\#}(\beta_j) \right.\nonumber\\
	& \left. \hspace*{5em} +\Phi_{\CH}^{(0)\#}(\beta_j)^{-1}\mathcal{E}_{x_j}(x_j)^{-1}\mathcal{E}_{x_j}'(x_j)\Phi_{\CH}^{(0)\#}(\beta_j)+r^{\frac32}f_{x_j}'(x_j)\Phi_{\CH}^{(1)\#}(\beta_j)+\tau rI\right]_{31}.
\end{align}
Using \eqref{eq:E_{x_j}(x_j)} and \eqref{eq: R(x_j) as r to +infinity}, we see that the first term in the bracket decays algebraically as $r\to +\infty$. More precisely,
\begin{equation}\label{eq: rightpart-computate 31 entry-1}
	\left[\Phi_{\CH}^{(0)\#}(\beta_j)^{-1}\mathcal{E}_{x_j}(x_j)^{-1}R(x_j)^{-1}R'(x_j)\mathcal{E}_{x_j}(x_j)\Phi_{\CH}^{(0)\#}(\beta_j)\right]_{31}=\Boh(r^{-\frac12}).
\end{equation}
For the second term, it follows from \eqref{def: Phi_CH^(0)sharp} and \eqref{eq:H-expand-coeff-0} that
\begin{align}
&\left[\Phi_{\CH}^{(0)\#}(\beta_j)^{-1}\mathcal{E}_{x_j}(x_j)^{-1}\mathcal{E}_{x_j}'(x_j)\Phi_{\CH}^{(0)\#}(\beta_j)\right]_{31}\nonumber\\
&=\Gamma(1-\beta_j)\Gamma(1+\beta_j)e^{-\pi\ii\beta_j}\left[\left(\mathcal{E}_{x_j}(x_j)^{-1}\mathcal{E}_{x_j}'(x_j)\right)_{33}-\left(\mathcal{E}_{x_j}(x_j)^{-1}\mathcal{E}_{x_j}'(x_j)\right)_{11}\right]\nonumber\\
&\quad+\Gamma(1-\beta_j)^2e^{-2\pi\ii\beta_j}\left(\mathcal{E}_{x_j}(x_j)^{-1}\mathcal{E}_{x_j}'(x_j)\right)_{31}-\Gamma(1+\beta_j)^2\left(\mathcal{E}_{x_j}(x_j)^{-1}\mathcal{E}_{x_j}'(x_j)\right)_{13}.
\end{align}
To simplify this expression, recalling \eqref{eq:E_{x_j}(x_j)E_{x_j}'(x_j)}--\eqref{def: c_xj}, \eqref{def:d1-xj0}--\eqref{def:d2-xj1} and \eqref{eq: fx_j and ' and ''}, we obtain
\begin{align}
	&\Gamma(1-\beta_j)\Gamma(1+\beta_j)e^{-\pi\ii\beta_j}\left[\left(\mathcal{E}_{x_j}(x_j)^{-1}\mathcal{E}_{x_j}'(x_j)\right)_{33}-\left(\mathcal{E}_{x_j}(x_j)^{-1}\mathcal{E}_{x_j}'(x_j)\right)_{11}\right] \nonumber\\
	&=|\Gamma(1+\beta_j)|^2e^{-\pi \ii\beta_j}\left(d_{2,-x_j}^{(1)}-d_{1,-x_j}^{(1)}-\frac{\beta_jf_{x_j}''(x_j)}{f_{x_j}'(x_j)}\right)\nonumber\\
	&=|\Gamma(1+\beta_j)|^2e^{-\pi \ii\beta_j}\Bigg(-\frac{3\beta_j}{2x_j}
	+\sum_{\substack{k=1\\k\neq j}}^{m}\frac{2\beta_kx_k^{\frac12}}{x_j^{\frac12}(x_j-x_k)}\Bigg)+\Boh(r^{-1}),
\end{align}
and
\begin{align}
	&\Gamma(1-\beta_j)^2e^{-2\pi\ii\beta_j}\left(\mathcal{E}_{x_j}(x_j)^{-1}\mathcal{E}_{x_j}'(x_j)\right)_{31}-\Gamma(1+\beta_j)^2\left(\mathcal{E}_{x_j}(x_j)^{-1}\mathcal{E}_{x_j}'(x_j)\right)_{13}\nonumber\\
	&=-\frac{\ii e^{-\pi\ii\beta_j}}{4x_j}\left(Z_j(r)^2+\overline{Z_j(r)}^2\right)=-\frac{\ii e^{-\pi\ii\beta_j}}{2x_j}\left|\Gamma\left(1+\beta_j\right)\right|^2\cos\left(2\vartheta_j(r)\right),
\end{align}
where $Z_j(r)$ is defined in \eqref{def: Z_j(r)} and $\vartheta_j(r)$ is explicitly given in \eqref{eq:def-theta-j-general}.
Using $|\Gamma(1+\beta_j)|^2 = \pi\beta_j/\sin(\pi\beta_j)$ additionally, we obtain
\begin{multline}\label{eq: rightpart-computate 31 entry-2}
\left[\Phi_{\CH}^{(0)\#}(\beta_j)^{-1}\mathcal{E}_{x_j}(x_j)^{-1}\mathcal{E}_{x_j}'(x_j)\Phi_{\CH}^{(0)\#}(\beta_j)\right]_{31}\\
=-\frac{3\pi e^{-\pi\ii\beta_j}\beta_j^2}{2x_j\sin(\pi\beta_j)}+\frac{2\pi \beta_j e^{-\pi\ii\beta_j }}{x_j^\frac12\sin(\pi\beta_j)}\sum_{\substack{k=1\\k\neq j}}^{m}\frac{\beta_kx_k^{\frac12}}{x_j-x_k}-\frac{\ii\pi\beta_j e^{-\pi\ii\beta_j}}{2x_j\sin(\pi\beta_j)}\cos\left(2\vartheta_j(r)\right)+\Boh(r^{-1}).
\end{multline}
For the third term in the bracket of \eqref{eq: rightpart-computate 31 entry-pre}, by \eqref{eq: fx_j and ' and ''} and \eqref{eq: 31entry of Phi_CH^(1)sharp}, we have
\begin{align}\label{eq: rightpart-computate 31 entry-3}
	\left[r^{\frac32}f_{x_j}'(x_j)\Phi_{\CH}^{(1)\#}(\beta_j)\right]_{31}=\frac{2\pi\ii\beta_j e^{-\pi\ii\beta_j}}{\sin(\pi\beta_j)}\left(\vr_1x_j^\frac12 r^{\frac32}-\vs_1 x_j^{-\frac12}r^{\frac12}\right).
\end{align}

Substituting \eqref{eq: rightpart-computate 31 entry-1}, \eqref{eq: rightpart-computate 31 entry-2}, \eqref{eq: rightpart-computate 31 entry-3} into \eqref{eq: rightpart-computate 31 entry-pre},
and noticing $(\tau rI)_{31}=0$, we obtain
\begin{multline}\label{eq: rightpart-computate 31 entry-pre-2}
-\sum_{j=1}^m x_j\mathfrak{s}_j\begin{pmatrix}0 & 0 & 1 & 1\end{pmatrix}X_{R,j}^{(1)}(r)\begin{pmatrix}1 \\ 1 \\ 0 \\ 0\end{pmatrix}
=-2\pi\ii\vr_1r^{\frac12}\sum_{j=1}^m\frac{\beta_j\mathfrak{s}_js_{j+1}^{-1}e^{-\pi\ii\beta_j}}{\sin(\pi\beta_j)}x_j^{\frac32}
+2\pi\ii\vs_1r^{-\frac12}\sum_{j=1}^m\frac{\beta_j\mathfrak{s}_js_{j+1}^{-1}e^{-\pi\ii\beta_j}}{\sin(\pi\beta_j)}x_j^{\frac12} \\
\quad +\pi r^{-1}\sum_{j=1}^m\frac{\mathfrak{s}_js_{j+1}^{-1}e^{-\pi\ii\beta_j}}{\sin(\pi\beta_j)}
\left(\frac{3\beta_j^2}{2}-2\beta_jx_j^{\frac12}\sum_{\substack{k=1\\k\neq j}}^{m}\frac{\beta_kx_k^{\frac12}}{x_j-x_k}+\frac{\ii\beta_j}{2}\cos\left(2\vartheta_j(r)\right)\right)
+\Boh\left(r^{-\frac32}\right).
\end{multline}
Equivalently, using \eqref{eq: beta_j and s_j relation} which indicates that
\begin{equation}
\mathfrak{s}_js_{j+1}^{-1}\frac{e^{-\pi\ii\beta_j}}{\sin(\pi\beta_j)}=-\frac{1}{\pi},
\end{equation}
we can rewrite \eqref{eq: rightpart-computate 31 entry-pre-2} as
\begin{multline*}
-\sum_{j=1}^m x_j\mathfrak{s}_j\begin{pmatrix}0 & 0 & 1 & 1\end{pmatrix}X_{R,j}^{(1)}(r)\begin{pmatrix}1 \\ 1 \\ 0 \\ 0\end{pmatrix}
=2\ii\vr_1r^{\frac12}\sum_{j=1}^m\beta_j x_j^{\frac32}
-2\ii\vs_1r^{-\frac12}\sum_{j=1}^m\beta_j x_j^{\frac12} \\
\quad -r^{-1}\sum_{j=1}^m
\left(\frac{3\beta_j^2}{2}-2\beta_jx_j^{\frac12}\sum_{\substack{k=1\\k\neq j}}^{m}\frac{\beta_kx_k^{\frac12}}{x_j-x_k}+\frac{\ii\beta_j}{2}\cos\left(2\vartheta_j(r)\right)\right)
+\Boh\left(r^{-\frac32}\right).
\end{multline*}
It is also noted that the cross term vanishes after summing over $j$, that is,
\begin{equation*}
\sum_{j=1}^{m}\beta_jx_j^{\frac12}\sum_{\substack{k=1\\k\neq j}}^{m}\frac{\beta_kx_k^{\frac12}}{x_j-x_k}
=\sum_{\substack{j,k=1\\j\neq k}}^{m}\frac{\beta_j\beta_kx_j^{\frac12}x_k^{\frac12}}{x_j-x_k}=0,
\end{equation*}
where the last equality follows from the anti-symmetry of the summand under $(j,k)\leftrightarrow (k,j)$.
Thus
\begin{multline}\label{eq: rightpart-computate 31 entry-final}
	-\sum_{j=1}^m x_j\mathfrak{s}_j\begin{pmatrix}0 & 0 & 1 & 1\end{pmatrix}X_{R,j}^{(1)}(r)\begin{pmatrix}1 \\ 1 \\ 0 \\ 0\end{pmatrix}
=2\ii\vr_1r^{\frac12}\sum_{j=1}^m\beta_j x_j^{\frac32}
-2\ii\vs_1r^{-\frac12}\sum_{j=1}^m\beta_j x_j^{\frac12} \\
\quad -r^{-1}\sum_{j=1}^m
\left(\frac{3\beta_j^2}{2}+\frac{\ii\beta_j}{2}\cos\left(2\vartheta_j(r)\right)\right)
+\Boh\left(r^{-\frac32}\right).
\end{multline}
It remains to compute the contribution from the left limit $X_{L,j}^{(1)}(r)$ in \eqref{eq:Hamiltonian-via-RH-2}.
By the symmetry relation \eqref{eq:symm}, as $z\to-rx_j$, we have
\begin{align*}
X(z)
&=\begin{pmatrix}
\sigma_1 & 0\\
0 & -\sigma_1
\end{pmatrix}\widetilde X(-z)\begin{pmatrix}
\sigma_1 & 0\\
0 & -\sigma_1
\end{pmatrix}\nonumber\\
&=\begin{pmatrix}
\sigma_1 & 0\\
0 & -\sigma_1
\end{pmatrix}\widetilde X_{R,j}^{(0)}(r)\left(I+\widetilde X_{R,j}^{(1)}(r)(-z-rx_j)+\Boh((z+rx_j)^2)\right)\begin{pmatrix}
\sigma_1 & 0\\
0 & -\sigma_1
\end{pmatrix},
\end{align*}
where $\sigma_1$ is defined in \eqref{def:Pauli}. Comparing this expansion with \eqref{eq: Phi-expand--rx_j} and using \eqref{eq:symmXLR}, 
we obtain
\begin{equation}\label{eq:symmXLR1}
X_{L,j}^{(1)}(r)=-\widetilde X_{R,j}^{(1)}(r)
\begin{pmatrix}
	\sigma_1 & 0\\
	0 & -\sigma_1
\end{pmatrix},
\end{equation}
and hence
\begin{align}
\sum_{j=1}^m x_j\mathfrak{s}_j\begin{pmatrix}0 & 0 & 1 & 1\end{pmatrix}X_{L,j}^{(1)}(r)\begin{pmatrix}1 \\ 1 \\ 0 \\ 0\end{pmatrix}=-\sum_{j=1}^m x_j\mathfrak{s}_j\begin{pmatrix}0 & 0 & 1 & 1\end{pmatrix}\widetilde X_{R,j}^{(1)}(r)\begin{pmatrix}1 \\ 1 \\ 0 \\ 0\end{pmatrix}.
\end{align}
After swapping $\vr_1\leftrightarrow\vr_2$ and $\vs_1\leftrightarrow\vs_2$ in \eqref{eq: rightpart-computate 31 entry-final}, it is deduced that
\begin{multline}\label{eq: left contribution via XL1}
\sum_{j=1}^m x_j\mathfrak{s}_j\begin{pmatrix}0 & 0 & 1 & 1\end{pmatrix}X_{L,j}^{(1)}(r)\begin{pmatrix}1 \\ 1 \\ 0 \\ 0\end{pmatrix}
=2\ii\vr_2r^{\frac12}\sum_{j=1}^m\beta_jx_j^{\frac32}
-2\ii\vs_2r^{-\frac12}\sum_{j=1}^m\beta_j x_j^{\frac12} \\
\quad -r^{-1}\sum_{j=1}^m
\left(\frac{3\beta_j^2}{2}+\frac{\ii\beta_j}{2}\cos\left(2\widetilde\vartheta_j(r)\right)\right)
+\Boh\left(r^{-\frac32}\right),
\end{multline}
where $\widetilde\vartheta_j(r)$ is defined through \eqref{eq:def-theta-j-general} after swapping $\vr_1\leftrightarrow\vr_2$ and $\vs_1\leftrightarrow\vs_2$.
Substituting \eqref{eq: rightpart-computate 31 entry-final} and \eqref{eq: left contribution via XL1} into \eqref{eq:Hamiltonian-via-RH-2}, we arrive at
\begin{multline}
	H(r)=\sum_{j=1}^{m}\Bigg[2\ii(\vr_1+\vr_2)\beta_jx_j^{\frac32}r^{\frac12}
	-2\ii(\vs_1+\vs_2)\beta_jx_j^{\frac12}r^{-\frac12} \\
	\qquad\qquad\qquad\qquad -\Bigg(3\beta_j^2+\frac{\ii\beta_j}{2}\cos\left(2\vartheta_j(r)\right)+\frac{\ii\beta_j}{2}\cos\left(2\widetilde\vartheta_j(r)\right)\Bigg)r^{-1}\Bigg]
	+\Boh\left(r^{-\frac32}\right),
\end{multline}
which is also equivalent to \eqref{eq:H-large-r} by using \eqref{def: first def beta_j}.

This completes the proof of Theorem \ref{thm: Hamiltonian and its asymptotics}. \qed

\subsection{Proof of Theorem \ref{thm: large_gap}}
First, substituting \eqref{eq:X1_11-large r} into \eqref{eq:derivativeFtau}, we see that 
\begin{align*}
&\frac{\partial}{\partial \tau}\log F(r\vec{x}, \vec{u}; \vr_1, \vr_2, \vs_1, \vs_2, \tau)=\Boh(r^{-\frac12}),
\end{align*}
which implies that the first terms up to and including the nontrivial constant term in the large $r$ expansion of $F(r\vec{x}, \vec{u})$ are independent of $\tau$.
We therefore set $\tau=0$ and integrate \eqref{eq:Hamiltonian-diff-identity-tau0} from $0$ to an arbitrary $r>0$: 
\begin{align}\label{eq: huge integral-0}
&\int_{0}^{r}H(\xi)\dif\xi=\int_{0}^{r}\left(\sum_{k=5}^{6}\left(p_k(\xi)q_k'(\xi)+\widetilde p_k(\xi)\widetilde q_k'(\xi)\right)+\sum_{j=1}^{m}\sum_{k=1}^{4}\left(p_{j,k}(\xi)q_{j,k}'(\xi)+\widetilde p_{j,k}(\xi)\widetilde q_{j,k}'(\xi)\right)-H(\xi)\right)\dif\xi\nonumber\\
&\quad+\frac{1}{3}\Bigg(
2\xi H(\xi)
+\sum_{j=1}^{m}\Big(
p_{j,1}(\xi)q_{j,1}(\xi)+p_{j,2}(\xi)q_{j,2}(\xi)
+\widetilde p_{j,1}(\xi)\widetilde q_{j,1}(\xi)
+\widetilde p_{j,2}(\xi)\widetilde q_{j,2}(\xi)
\Big)\nonumber\\
&\quad\quad\quad-2p_5(\xi)q_5(\xi)-2\widetilde p_5(\xi)\widetilde q_5(\xi)
-p_6(\xi)q_6(\xi)-\widetilde p_6(\xi)\widetilde q_6(\xi)
+2\frac{\vs_1}{\vr_1}p_5(\xi)+2\frac{\vs_2}{\vr_2}\widetilde p_5(\xi)
\Bigg)\Bigg{|}_{\xi=0}^{\xi=r}.
\end{align}
To proceed, let us write $\vec 0=(0,\ldots,0)\in\mathbb{R}^{m}$ and
\begin{equation*}
	\vec{\beta}:=(\beta_1,\ldots,\beta_m), \quad \quad \vec{\beta}_{\ell}=(\beta_1,\ldots,\beta_{\ell},0,\ldots,0),  \quad \quad \vec{\beta}_{\ell}'=(\beta_1,\ldots,\beta_{\ell-1},\beta_{\ell}',0,\ldots,0).
\end{equation*}
We also make explicit dependence of the functions in \eqref{def: unknown functions} and of $H$
by writing $p_{j,k}(r;\vec\beta)$, $q_{j,k}(r;\vec\beta)$, $p_5(r;\vec\beta)$, $p_6(r;\vec\beta)$, $q_5(r;\vec\beta)$, $q_6(r;\vec\beta)$, and $H(r;\vec\beta)$.
It follows from \eqref{def: first def beta_j} that the parameter $\gamma$ 
in \eqref{eq:Hamiltonian-diff-identity-parameter} can be replaced by any $\beta_\ell$ with $\ell=1,\ldots,m$. 

Applying \eqref{eq:Hamiltonian-diff-identity-parameter} with $\vec{\beta}=\vec{\beta}_{\ell-1}$ and $\gamma=\beta_{\ell}$, 
we integrate with respect to $\beta_\ell$ from $0$ to $\beta_\ell$ for any $\ell=1,\ldots,m$,
and with respect to $r$ from $0$ to $r$, to obtain
\begin{align}\label{eq: huge integral-1}
&\int_{0}^{r}\Bigg(\sum_{k=5}^{6}\left(p_k(\xi;\vec{\beta}_{\ell})q_k'(\xi;\vec{\beta}_{\ell})+\widetilde p_k(\xi;\vec{\beta}_{\ell})\widetilde q_k'(\xi;\vec{\beta}_{\ell})\right) \nonumber\\
&\hspace*{7em}+\sum_{j=1}^{m}\sum_{k=1}^{4}\left(p_{j,k}(\xi;\vec{\beta}_{\ell})q_{j,k}'(\xi;\vec{\beta}_{\ell})+\widetilde p_{j,k}(\xi;\vec{\beta}_{\ell})\widetilde q_{j,k}'(\xi;\vec{\beta}_{\ell})\right)-H(\xi;\vec{\beta}_{\ell})\Bigg)\dif\xi \nonumber\\
&\quad-\int_{0}^{r}\Bigg(\sum_{k=5}^{6}\left(p_k(\xi;\vec{\beta}_{\ell-1})q_k'(\xi;\vec{\beta}_{\ell-1})+\widetilde p_k(\xi;\vec{\beta}_{\ell-1})\widetilde q_k'(\xi;\vec{\beta}_{\ell-1})\right) \nonumber\\
&\hspace*{7em}+\sum_{j=1}^{m}\sum_{k=1}^{4}\left(p_{j,k}(\xi;\vec{\beta}_{\ell-1})q_{j,k}'(\xi;\vec{\beta}_{\ell-1})+\widetilde p_{j,k}(\xi;\vec{\beta}_{\ell-1})\widetilde q_{j,k}'(\xi;\vec{\beta}_{\ell-1})\right)-H(\xi;\vec{\beta}_{\ell-1})\Bigg)\dif\xi \nonumber\\
&=\int_{0}^{\beta_\ell}\Bigg(\sum_{k=5}^{6}\left(p_k(r;\vec{\beta}_{\ell}')\partial_{\beta_\ell'}q_k(r;\vec{\beta}_{\ell}')+\widetilde p_k(r;\vec{\beta}_{\ell}')\partial_{\beta_\ell'}\widetilde q_k(r;\vec{\beta}_{\ell}')\right)\nonumber\\
&\hspace*{9em}+\sum_{j=1}^{m}\sum_{k=1}^{4}\left(p_{j,k}(r;\vec{\beta}_{\ell}')\partial_{\beta_\ell'}q_{j,k}(r;\vec{\beta}_{\ell}')+\widetilde p_{j,k}(r;\vec{\beta}_{\ell}')\partial_{\beta_\ell'}\widetilde q_{j,k}(r;\vec{\beta}_{\ell}')\right)\Bigg)\dif\beta_\ell' \nonumber\\
&\quad\quad-\cancel{\int_{0}^{\beta_\ell}\Bigg(\sum_{k=5}^{6}\left(p_k(0;\vec{\beta}_{\ell}')\partial_{\beta_\ell'}q_k(0;\vec{\beta}_{\ell}')+\widetilde p_k(0;\vec{\beta}_{\ell}')\partial_{\beta_\ell'}\widetilde q_k(0;\vec{\beta}_{\ell}')\right)}\nonumber\\
&\hspace*{8em}\cancel{+\sum_{j=1}^{m}\sum_{k=1}^{4}\left(p_{j,k}(0;\vec{\beta}_{\ell}')\partial_{\beta_\ell'}q_{j,k}(0;\vec{\beta}_{\ell}')+\widetilde p_{j,k}(0;\vec{\beta}_{\ell}')\partial_{\beta_\ell'}\widetilde q_{j,k}(0;\vec{\beta}_{\ell}')\right)\Bigg)\dif\beta_\ell'},
\end{align}
where the second integral of the r.h.s. vanishes by the asymptotics of $p_{j,k}$, $q_{j,k}$, $p_5$, $q_5$, $p_6$ and $q_6$ established in Proposition \ref{prop:coupled system-asy}.

Summing \eqref{eq: huge integral-1} over $\ell=1,\ldots,m$, we obtain
\begin{align}\label{eq: huge integral-2}
&\int_{0}^{r}\Bigg(\sum_{k=5}^{6}\left(p_k(\xi;\vec{\beta})q_k'(\xi;\vec{\beta})+\widetilde p_k(\xi;\vec{\beta})\widetilde q_k'(\xi;\vec{\beta})\right) \nonumber\\
&\hspace*{7em}+\sum_{j=1}^{m}\sum_{k=1}^{4}\left(p_{j,k}(\xi;\vec{\beta})q_{j,k}'(\xi;\vec{\beta})+\widetilde p_{j,k}(\xi;\vec{\beta})\widetilde q_{j,k}'(\xi;\vec{\beta})\right)-H(\xi;\vec{\beta})\Bigg)\dif\xi \nonumber\\
&\quad-\int_{0}^{r}\Bigg(\sum_{k=5}^{6}\left(p_k(\xi;\vec{0})q_k'(\xi;\vec{0})+\widetilde p_k(\xi;\vec{0})\widetilde q_k'(\xi;\vec{0})\right) \nonumber\\
&\hspace*{7em}+\sum_{j=1}^{m}\sum_{k=1}^{4}\left(p_{j,k}(\xi;\vec{0})q_{j,k}'(\xi;\vec{0})+\widetilde p_{j,k}(\xi;\vec{0})\widetilde q_{j,k}'(\xi;\vec{0})\right)-H(\xi;\vec{0})\Bigg)\dif\xi \nonumber\\
&=\sum_{\ell=1}^{m}\int_{0}^{\beta_\ell}\Bigg(\sum_{k=5}^{6}\left(p_k(r;\vec{\beta}_{\ell}')\partial_{\beta_\ell'}q_k(r;\vec{\beta}_{\ell}')+\widetilde p_k(r;\vec{\beta}_{\ell}')\partial_{\beta_\ell'}\widetilde q_k(r;\vec{\beta}_{\ell}')\right)\nonumber\\
&\hspace*{9em}+\sum_{j=1}^{m}\sum_{k=1}^{4}\left(p_{j,k}(r;\vec{\beta}_{\ell}')\partial_{\beta_\ell'}q_{j,k}(r;\vec{\beta}_{\ell}')+\widetilde p_{j,k}(r;\vec{\beta}_{\ell}')\partial_{\beta_\ell'}\widetilde q_{j,k}(r;\vec{\beta}_{\ell}')\right)\Bigg)\dif\beta_\ell'.
\end{align}
If $\beta_\ell=0$, or equivalently $s_{\ell+1}=s_\ell$, then \eqref{def:qkpk} implies that 
$p_{\ell,k}(r)=0$ for $k=1,2,3,4$ and all $r>0$. 
Using \eqref{eq:Hamiltonian-via-RH}, it follows that $H(r;\vec{0})=0$.
Moreover, \eqref{def:vecf and vech} and \eqref{eq:Y-jump} lead to
$Y|_{\vec\beta=\vec 0}=I$, which together with \eqref{def:X1} and \eqref{def:p56}--\eqref{def:q56} imply 
that
\begin{align}
p_{5}(r;\vec 0)&= \ii \vr_1 (M_1)_{13}, & p_{6}(r;\vec 0) &= -\ii \vr_1 (M_1)_{43} -\ii \vr_2 (M_1)_{21}, \label{eq: p56_zero}\\
q_{5}(r;\vec 0)&= (M_1)_{33}-(M_1)_{11}, & q_{6}(r;\vec 0)&= (M_1)_{14},\label{eq: q56_zero}
\end{align}
are all independent of $r$. Consequently, the second integral on the l.h.s. of \eqref{eq: huge integral-2} vanishes, namely,
\begin{multline*}
\int_{0}^{r}\Bigg(\sum_{k=5}^{6}\left(p_k(\xi;\vec{0})q_k'(\xi;\vec{0})+\widetilde p_k(\xi;\vec{0})\widetilde q_k'(\xi;\vec{0})\right) \\
+\sum_{j=1}^{m}\sum_{k=1}^{4}\left(p_{j,k}(\xi;\vec{0})q_{j,k}'(\xi;\vec{0})+\widetilde p_{j,k}(\xi;\vec{0})\widetilde q_{j,k}'(\xi;\vec{0})\right)-H(\xi;\vec{0})\Bigg)\dif\xi=0.
\end{multline*}
Thus \eqref{eq: huge integral-2} simplifies to 
\begin{align}\label{eq: huge integral-3}
&\int_{0}^{r}\Bigg(\sum_{k=5}^{6}\left(p_k(\xi;\vec{\beta})q_k'(\xi;\vec{\beta})+\widetilde p_k(\xi;\vec{\beta})\widetilde q_k'(\xi;\vec{\beta})\right) \nonumber\\
&\hspace*{7em}+\sum_{j=1}^{m}\sum_{k=1}^{4}\left(p_{j,k}(\xi;\vec{\beta})q_{j,k}'(\xi;\vec{\beta})+\widetilde p_{j,k}(\xi;\vec{\beta})\widetilde q_{j,k}'(\xi;\vec{\beta})\right)-H(\xi;\vec{\beta})\Bigg)\dif\xi \nonumber\\
&=\sum_{\ell=1}^{m}\int_{0}^{\beta_\ell}\Bigg(\sum_{k=5}^{6}\left(p_k(r;\vec{\beta}_{\ell}')\partial_{\beta_\ell'}q_k(r;\vec{\beta}_{\ell}')+\widetilde p_k(r;\vec{\beta}_{\ell}')\partial_{\beta_\ell'}\widetilde q_k(r;\vec{\beta}_{\ell}')\right)\nonumber\\
&\hspace*{9em}+\sum_{j=1}^{\ell}\sum_{k=1}^{4}\left(p_{j,k}(r;\vec{\beta}_{\ell}')\partial_{\beta_\ell'}q_{j,k}(r;\vec{\beta}_{\ell}')+\widetilde p_{j,k}(r;\vec{\beta}_{\ell}')\partial_{\beta_\ell'}\widetilde q_{j,k}(r;\vec{\beta}_{\ell}')\right)\Bigg)\dif\beta_\ell',
\end{align}
where the summation over $j$ in the last integral runs only up to $\ell$ rather than $m$, because $p_{j,k}(r;\vec{\beta}_{\ell}')$ and $\widetilde p_{j,k}(r;\vec{\beta}_{\ell}')$ vanish for $j>\ell$.
Substituting \eqref{eq: huge integral-3} into \eqref{eq: huge integral-0}, we obtain
\begin{align}\label{eq: huge integral-4-simplified}
&\int_{0}^{r}H(\xi; \vec\beta)\dif\xi=\sum_{\ell=1}^{m}\int_{0}^{\beta_\ell}\Bigg(\sum_{k=5}^{6}\left(p_k(r;\vec{\beta}_{\ell}')\partial_{\beta_\ell'}q_k(r;\vec{\beta}_{\ell}')+\widetilde p_k(r;\vec{\beta}_{\ell}')\partial_{\beta_\ell'}\widetilde q_k(r;\vec{\beta}_{\ell}')\right)\nonumber\\
&\hspace*{9em}+\sum_{j=1}^{\ell}\sum_{k=1}^{4}\left(p_{j,k}(r;\vec{\beta}_{\ell}')\partial_{\beta_\ell'}q_{j,k}(r;\vec{\beta}_{\ell}')+\widetilde p_{j,k}(r;\vec{\beta}_{\ell}')\partial_{\beta_\ell'}\widetilde q_{j,k}(r;\vec{\beta}_{\ell}')\right)\Bigg)\dif\beta_\ell' \nonumber\\
&+\frac{1}{3}\Bigg(
2r H(r;\vec{\beta})
+\sum_{j=1}^{m}\Big(
p_{j,1}(r;\vec{\beta})q_{j,1}(r;\vec{\beta})+p_{j,2}(r;\vec{\beta})q_{j,2}(r;\vec{\beta})
+\widetilde p_{j,1}(r;\vec{\beta})\widetilde q_{j,1}(r;\vec{\beta})
+\widetilde p_{j,2}(r;\vec{\beta})\widetilde q_{j,2}(r;\vec{\beta})
\Big)\nonumber\\
&\quad-2p_5(r;\vec{\beta})q_5(r;\vec{\beta})-2\widetilde p_5(r;\vec{\beta})\widetilde q_5(r;\vec{\beta})
-p_6(r;\vec{\beta})q_6(r;\vec{\beta})-\widetilde p_6(r;\vec{\beta})\widetilde q_6(r;\vec{\beta})
+2\frac{\vs_1}{\vr_1}p_5(r;\vec{\beta})+2\frac{\vs_2}{\vr_2}\widetilde p_5(r;\vec{\beta})\nonumber\\
&\quad+2p_5(0;\vec{\beta})q_5(0;\vec{\beta})+2\widetilde p_5(0;\vec{\beta})\widetilde q_5(0;\vec{\beta})
+p_6(0;\vec{\beta})q_6(0;\vec{\beta})+\widetilde p_6(0;\vec{\beta})\widetilde q_6(0;\vec{\beta})
-2\frac{\vs_1}{\vr_1}p_5(0;\vec{\beta})-2\frac{\vs_2}{\vr_2}\widetilde p_5(0;\vec{\beta})
\Bigg),
\end{align}
where we also used 
\begin{equation}\label{eq: huge integral-B0-4}
\sum_{j=1}^{m}\Big(
p_{j,1}(\xi)q_{j,1}(\xi)+p_{j,2}(\xi)q_{j,2}(\xi)
+\widetilde p_{j,1}(\xi)\widetilde q_{j,1}(\xi)
+\widetilde p_{j,2}(\xi)\widetilde q_{j,2}(\xi)
\Big)=\Boh(\xi), \qquad \xi\to 0^+,
\end{equation}
following directly from \eqref{vecq small r}, \eqref{vecp small r} and the symmetry relation of $\widehat M(0)$ in \eqref{eq: symmetry of M at 0}.
Also, we denote
\begin{multline}\label{def: mathcalB}
\mathcal{B}(0;\vec{\beta}):=2p_5(0;\vec{\beta})q_5(0;\vec{\beta})+2\widetilde p_5(0;\vec{\beta})\widetilde q_5(0;\vec{\beta})
\\+p_6(0;\vec{\beta})q_6(0;\vec{\beta})+\widetilde p_6(0;\vec{\beta})\widetilde q_6(0;\vec{\beta})
-2\frac{\vs_1}{\vr_1}p_5(0;\vec{\beta})-2\frac{\vs_2}{\vr_2}\widetilde p_5(0;\vec{\beta}).
\end{multline}
Using \eqref{q5-small r}--\eqref{q6-small r} and \eqref{p5-small r}--\eqref{p6-small r}, and noting that $\tau=0$ (thus $f=\dot{f}$),
we find that 
\begin{align}\label{eq: mathcalB(0)}
\mathcal B(0;\vec{\beta})
=&-4\ii \vr_1 (M_1)_{11}(M_1)_{13}
-4\ii \vr_2 (\widetilde M_1)_{11}(\widetilde M_1)_{13}
\nonumber\\
&+\left(\ii \vr_1 (M_1)_{12}+\ii \vr_2 (\widetilde M_1)_{12}\right)(M_1)_{14}
+\left(\ii \vr_1 (M_1)_{12}+\ii \vr_2 (\widetilde M_1)_{12}\right)(\widetilde M_1)_{14}
\nonumber\\
&-2\ii \vs_1 (M_1)_{13}-2\ii \vs_2 (\widetilde M_1)_{13}.
\end{align}
We now turn to evaluate the terms $p_{j,k}(r;\vec{\beta}_\ell)\partial_{\beta_\ell}q_{j,k}(r;\vec{\beta}_\ell)$ for $k=1,2,3,4$
and $p_k(r;\vec{\beta}_\ell)\partial_{\beta_\ell}q_k(r;\vec{\beta}_\ell)$ for $k=5,6$.

Using \eqref{qj1-large r} and \eqref{pj1-large r}, we have, as $r\to+\infty$,
\begin{align}\label{eq: evaluation-pparitalq-a}
&p_{j,1}(r;\vec{\beta}_\ell)\partial_{\beta_\ell}q_{j,1}(r;\vec{\beta}_\ell)=p_{j,1}(r;\vec{\beta}_\ell)q_{j,1}(r;\vec{\beta}_\ell)\partial_{\beta_\ell}\log q_{j,1}(r;\vec{\beta}_\ell)\nonumber\\
&\quad=\Bigg(\ii\beta_j\cos\left(2\vartheta_j(r;\vec{\beta}_\ell)\right)
-2\beta_j x_j^{-\frac12}\sum_{k=1}^{\ell} \beta_k x_k^{\frac12}
\Big(1+\sin\left(2\vartheta_j(r;\vec{\beta}_\ell)\right)\Big)\Bigg)\nonumber\\
&\quad\quad\quad\times\Bigg(\partial_{\beta_\ell}\log\mathcal{A}_j(r;\vec{\beta}_\ell)
-\tan\left(\vartheta_j(r;\vec{\beta}_\ell)-\frac{\pi}{4}\right)\partial_{\beta_\ell}\vartheta_j(r;\vec{\beta}_\ell)\Bigg)
+\Boh\left(r^{-\frac12}\log r\right).
\end{align}
With the aid of \eqref{qj3-large r} and \eqref{pj3-large r}, we obtain, as $r\to+\infty$,
\begin{align}\label{eq: evaluation-pparitalq-b}
&p_{j,3}(r;\vec{\beta}_\ell)\partial_{\beta_\ell}q_{j,3}(r;\vec{\beta}_\ell)=p_{j,3}(r;\vec{\beta}_\ell)q_{j,3}(r;\vec{\beta}_\ell)\partial_{\beta_\ell}\log q_{j,3}(r;\vec{\beta}_\ell)\nonumber\\
&\quad=\Bigg(-\ii\beta_j\cos\left(2\vartheta_j(r;\vec{\beta}_\ell)\right)
+2\beta_j x_j^{-\frac12}\sum_{k=1}^{\ell}\beta_k x_k^{\frac12}
\Big(1+\sin\left(2\vartheta_j(r;\vec{\beta}_\ell)\right)\Big)\Bigg)\partial_{\beta_\ell}\log\mathcal{A}_j(r;\vec{\beta}_\ell)\nonumber\\
&\quad\quad+2\beta_j x_{\ell}^{\frac12}x_j^{-\frac12}
\Big(1+\sin\left(2\vartheta_j(r;\vec{\beta}_\ell)\right)\Big)\nonumber\\
&\quad\quad+\Bigg(2\beta_j x_j^{-\frac12}\sum_{k=1}^{\ell}\beta_k x_k^{\frac12}
\cos\left(2\vartheta_j(r;\vec{\beta}_\ell)\right)
+\ii\beta_j\Big(1+\sin\left(2\vartheta_j(r;\vec{\beta}_\ell)\right)\Big)\Bigg)\partial_{\beta_\ell}\vartheta_j(r;\vec{\beta}_\ell)
+\Boh\left(r^{-\frac12}\log r\right).
\end{align}
Using \eqref{qj2-large r}, \eqref{pj2-large r}, \eqref{qj4-large r} and \eqref{pj4-large r}, we have, as $r\to+\infty$,
\begin{align}\label{eq: evaluation-pparitalq-c}
	&p_{j,2}(r;\vec{\beta}_\ell)\partial_{\beta_\ell}q_{j,2}(r;\vec{\beta}_\ell)=\Boh\left(r^{-1}\log r\right),
	&p_{j,4}(r;\vec{\beta}_\ell)\partial_{\beta_\ell}q_{j,4}(r;\vec{\beta}_\ell)=\Boh\left(r^{-1}\log r\right).
\end{align}
Furthermore, by \eqref{def: first def beta_j}--\eqref{def:amplitude Aj}, it is seen that
\begin{align}\label{eq: partiallogAj}
\partial_{\beta_\ell}\log\mathcal A_j(r;\vec\beta_\ell)=
\begin{cases}
-\pi \ii, & j<\ell,\\
-\frac{\pi \ii}{2}+\partial_{\beta_\ell}\log|\Gamma(1-\beta_\ell)|, & j=\ell,\\
0, & j>\ell.
\end{cases}
\end{align}
Combining the formulae \eqref{eq: evaluation-pparitalq-a}--\eqref{eq: partiallogAj}, we obtain, as $r\to+\infty$,
\begin{align}\label{eq: evaluation-sum-pparitalq}
&\sum_{j=1}^{\ell}\sum_{k=1}^{4}p_{j,k}(r;\vec{\beta}_\ell)\partial_{\beta_\ell}q_{j,k}(r;\vec{\beta}_\ell)\nonumber\\
&\quad=2\ii\sum_{j=1}^{\ell}\beta_j\partial_{\beta_\ell}\vartheta_j(r;\vec{\beta}_\ell)
+2x_{\ell}^{\frac12}\sum_{j=1}^{\ell}\beta_jx_j^{-\frac12}
\Bigl(1+\sin\bigl(2\vartheta_j(r;\vec{\beta}_\ell)\bigr)\Bigr)
+\Boh\left(r^{-\frac12}\log r\right).
\end{align}
By \eqref{def:tildeX}, the same argument yields, as $r\to+\infty$,
\begin{align}\label{eq: evaluation-sum-tildepparitalq}
&\sum_{j=1}^{\ell}\sum_{k=1}^{4}\widetilde p_{j,k}(r;\vec{\beta}_\ell)\partial_{\beta_\ell}\widetilde q_{j,k}(r;\vec{\beta}_\ell)\nonumber\\
&\quad=2\ii\sum_{j=1}^{\ell}\beta_j\partial_{\beta_\ell}\widetilde\vartheta_j(r;\vec{\beta}_\ell)
+2x_{\ell}^{\frac12}\sum_{j=1}^{\ell}\beta_jx_j^{-\frac12}
\Bigl(1+\sin\bigl(2\widetilde\vartheta_j(r;\vec{\beta}_\ell)\bigr)\Bigr)
+\Boh\left(r^{-\frac12}\log r\right).
\end{align}
To evaluate $p_5(r;\vec{\beta}_{\ell})\partial_{\beta_\ell}q_5(r;\vec{\beta}_{\ell})$, it follows from the condition 
\eqref{eq:extra-condition-2} that
\begin{align}\label{eq: evaluation-p5paritalq5-a}
	&p_5(r;\vec{\beta}_{\ell})\partial_{\beta_\ell}q_5(r;\vec{\beta}_{\ell})\nonumber\\
	&=-\frac{p_5(r;\vec{\beta}_{\ell})}{\ii\vr_1}\partial_{\beta_\ell}\left(\sum_{j=1}^{\ell}\left(q_{j,1}(r;\vec{\beta}_\ell)p_{j,3}(r;\vec{\beta}_\ell)-\widetilde q_{j,2}(r;\vec{\beta}_\ell)\widetilde p_{j,4}(r;\vec{\beta}_\ell)\right)+\frac{p_5(r;\vec{\beta}_{\ell})^2}{\ii\vr_1}-\ii\vr_2q_6(r;\vec{\beta}_{\ell})\widetilde q_6(r;\vec{\beta}_{\ell})+\ii\vs_1\right)\nonumber\\
	&=-\frac{p_5(r;\vec{\beta}_{\ell})}{\ii\vr_1}\partial_{\beta_\ell}\left(\sum_{j=1}^{\ell}\left(q_{j,1}(r;\vec{\beta}_\ell)p_{j,3}(r;\vec{\beta}_\ell)\right)\right)
	+\frac{p_5(r;\vec{\beta}_{\ell})}{\ii\vr_1}\partial_{\beta_\ell}\left(\sum_{j=1}^{\ell}\left(\widetilde q_{j,2}(r;\vec{\beta}_\ell)\widetilde p_{j,4}(r;\vec{\beta}_\ell)\right)\right)\nonumber\\
	&\quad\quad +\frac{2}{3\vr_1^2}\partial_{\beta_\ell}p_5(r;\vec{\beta}_{\ell})^3
	+\frac{\vr_2}{\vr_1}p_5(r;\vec{\beta}_{\ell})\partial_{\beta_\ell}\left(q_6(r;\vec{\beta}_{\ell})\widetilde q_6(r;\vec{\beta}_{\ell})\right),
\end{align}
where we have used the fact that $p_{j,k}(r)=0$ for $k=1,2,3,4$ and $j>\ell$.

We now integrate each term on the r.h.s of \eqref{eq: evaluation-p5paritalq5-a} with respect to $\beta_\ell$ from $0$ to $\beta_\ell$, and estimate the resulting expressions as $r\to+\infty$.
For the first term, integrating by parts and using \eqref{qj1-large r}, \eqref{pj3-large r} and \eqref{p5-large r}, we get that, as $r\to+\infty$,
\begin{multline}\label{eq: evaluation-p5paritalq5-b}
-\frac{1}{\ii\vr_1}\int_{0}^{\beta_\ell}p_5(r;\vec{\beta}_{\ell}') \partial_{\beta_\ell}'\left(\sum_{j=1}^{\ell}\left(q_{j,1}(r;\vec{\beta}_\ell')p_{j,3}(r;\vec{\beta}_\ell')\right)\right) \dif\beta_\ell'\\
=-\frac{p_5(r;\vec{\beta}_{\ell})}{\ii\vr_1}\sum_{j=1}^{\ell}\left(q_{j,1}(r;\vec{\beta}_\ell)p_{j,3}(r;\vec{\beta}_\ell)\right)+\frac{p_5(r;\vec{\beta}_{\ell-1})}{\ii\vr_1}\sum_{j=1}^{\ell-1}\left(q_{j,1}(r;\vec{\beta}_{\ell-1})p_{j,3}(r;\vec{\beta}_{\ell-1})\right)\\
\quad-2x_{\ell}^{\frac12}\int_{0}^{\beta_\ell}\sum_{j=1}^{\ell}\beta_j'x_j^{-\frac12}\Bigl(1+\sin\bigl(2\vartheta_j(r;\vec{\beta}_{\ell}')\bigr)\Bigr)\dif\beta_\ell'+\Boh\left(r^{-\frac12}\log r\right).
\end{multline}
Here $\beta_j'$ denotes the $j$-th component of $\vec{\beta}_\ell'$.
Next, with the aid of \eqref{qj2-large r}, \eqref{pj4-large r} and \eqref{p5-large r}, we have, as $r\to+\infty$,
\begin{align}\label{eq: evaluation-p5paritalq5-c}
	\frac{1}{\ii\vr_1}\int_{0}^{\beta_\ell}p_5(r;\vec{\beta}_{\ell}') \partial_{\beta_\ell}'\left(\sum_{j=1}^{\ell}\left(\widetilde q_{j,2}(r;\vec{\beta}_\ell')\widetilde p_{j,4}(r;\vec{\beta}_\ell')\right)\right) \dif\beta_\ell'=\Boh\left(r^{-1}\log r\right).
\end{align}
Also, by using \eqref{q6-large r} and \eqref{p5-large r}, it follows that
\begin{align}\label{eq: evaluation-p5paritalq5-d}
	\frac{2}{3\vr_1^2}\int_{0}^{\beta_\ell}\partial_{\beta_\ell}'p_5(r;\vec{\beta}_{\ell}')^3\dif\beta_\ell'=\frac{2}{3\vr_1^2}p_5(r;\vec{\beta}_{\ell})^3-\frac{2}{3\vr_1^2}p_5(r;\vec{\beta}_{\ell-1})^3,
\end{align}
and, as $r\to+\infty$,
\begin{multline}\label{eq: evaluation-p5paritalq5-e}
	\frac{\vr_2}{\vr_1}\int_{0}^{\beta_\ell}p_5(r;\vec{\beta}_{\ell}')\partial_{\beta_\ell}'\left(q_6(r;\vec{\beta}_{\ell}')\widetilde q_6(r;\vec{\beta}_{\ell}')\right)\dif\beta_\ell'\\
	=\frac{\vr_2}{\vr_1}p_5(r;\vec{\beta}_{\ell})q_6(r;\vec{\beta}_{\ell})\widetilde q_6(r;\vec{\beta}_{\ell})-\frac{\vr_2}{\vr_1}p_5(r;\vec{\beta}_{\ell-1})q_6(r;\vec{\beta}_{\ell-1})\widetilde q_6(r;\vec{\beta}_{\ell-1})+\Boh\left(r^{-\frac{1}{2}}\log r\right).
\end{multline}
Combining the estimates \eqref{eq: evaluation-p5paritalq5-a}--\eqref{eq: evaluation-p5paritalq5-e} and then summing over $\ell=1,\ldots,m$, we get
\begin{align}\label{eq: evaluation-p5paritalq5-final}
&\sum_{\ell=1}^{m}\int_{0}^{\beta_\ell}p_5(r;\vec{\beta}_{\ell}')\partial_{\beta_\ell'}q_5(r;\vec{\beta}_{\ell}')\dif\beta_\ell'\nonumber\\
&=-\frac{p_5(r;\vec{\beta})}{\ii\vr_1}\sum_{j=1}^{m}\left(q_{j,1}(r;\vec{\beta})p_{j,3}(r;\vec{\beta})\right)
+\frac{2}{3\vr_1^2}p_5(r;\vec{\beta})^3-\frac{2}{3\vr_1^2}p_5(r;\vec{0})^3\nonumber\\
	&\quad +\frac{\vr_2}{\vr_1}p_5(r;\vec{\beta})q_6(r;\vec{\beta})\widetilde q_6(r;\vec{\beta})
-\frac{\vr_2}{\vr_1}p_5(r;\vec{0})q_6(r;\vec{0})\widetilde q_6(r;\vec{0})\nonumber\\
	&\quad -2\sum_{\ell=1}^{m}x_{\ell}^{\frac12}\int_{0}^{\beta_\ell}\sum_{j=1}^{\ell}\beta_j'x_j^{-\frac12}\Bigl(1+\sin\bigl(2\vartheta_j(r;\vec{\beta}_{\ell}')\bigr)\Bigr)\dif\beta_\ell'
+\Boh\left(r^{-\frac12}\log r\right).
\end{align}
By \eqref{def:tildeX}, the same argument also gives, as $r\to+\infty$,
\begin{align}\label{eq: evaluation-tildep5paritalq5-final}
&\sum_{\ell=1}^{m}\int_{0}^{\beta_\ell}\widetilde p_5(r;\vec{\beta}_{\ell}')\partial_{\beta_\ell'}\widetilde q_5(r;\vec{\beta}_{\ell}')\dif\beta_\ell'\nonumber\\
&=-\frac{\widetilde p_5(r;\vec{\beta})}{\ii\vr_2}
\sum_{j=1}^{m}\left(\widetilde q_{j,1}(r;\vec{\beta})\widetilde p_{j,3}(r;\vec{\beta})\right)
+\frac{2}{3\vr_2^2}\widetilde p_5(r;\vec{\beta})^3
-\frac{2}{3\vr_2^2}\widetilde p_5(r;\vec{0})^3\nonumber\\
&\quad +\frac{\vr_1}{\vr_2}\widetilde p_5(r;\vec{\beta})q_6(r;\vec{\beta})\widetilde q_6(r;\vec{\beta})
-\frac{\vr_1}{\vr_2}\widetilde p_5(r;\vec{0})q_6(r;\vec{0})\widetilde q_6(r;\vec{0})\nonumber\\
&\quad -2\sum_{\ell=1}^{m}x_{\ell}^{\frac12}\int_{0}^{\beta_\ell}\sum_{j=1}^{\ell}\beta_j'x_j^{-\frac12}
\Bigl(1+\sin\bigl(2\widetilde\vartheta_j(r;\vec{\beta}_{\ell}')\bigr)\Bigr)\dif\beta_\ell'
+\Boh\left(r^{-\frac12}\log r\right).
\end{align}
Finally, it follows from \eqref{q6-large r}, \eqref{p6-large r} and \eqref{def:tildeX} that, as $r\to+\infty$,
\begin{align}\label{eq: evaluation-p6paritalq6-final}
	\sum_{\ell=1}^{m}\int_{0}^{\beta_\ell}\Bigl(p_6(r;\vec{\beta}_{\ell}')\partial_{\beta_\ell'}q_6(r;\vec{\beta}_{\ell}')
	+\widetilde p_6(r;\vec{\beta}_{\ell}')\partial_{\beta_\ell'}\widetilde q_6(r;\vec{\beta}_{\ell}')\Bigr)\dif\beta_\ell'
	=p_6(r;\vec{\beta})q_6(r;\vec{\beta})+\widetilde p_6(r;\vec{\beta})\widetilde q_6(r;\vec{\beta})\nonumber\\
	\quad-p_6(r;\vec{0})q_6(r;\vec{0})-\widetilde p_6(r;\vec{0})\widetilde q_6(r;\vec{0})
	+\Boh\left(r^{-\frac12}\log r\right).
\end{align}

Substituting \eqref{eq: evaluation-sum-pparitalq}, \eqref{eq: evaluation-sum-tildepparitalq}, \eqref{eq: evaluation-p5paritalq5-final}, \eqref{eq: evaluation-tildep5paritalq5-final} and \eqref{eq: evaluation-p6paritalq6-final} into \eqref{eq: huge integral-4-simplified}, we obtain that
\begin{align}
	&\int_{0}^{r}H(\xi;\vec{\beta})\dif\xi
	=2\ii\sum_{\ell=1}^{m}\int_{0}^{\beta_\ell}\sum_{j=1}^{\ell}\beta_j'\partial_{\beta_\ell'}\vartheta_j(r;\vec{\beta}_{\ell}')\dif\beta_\ell'
	+2\ii\sum_{\ell=1}^{m}\int_{0}^{\beta_\ell}\sum_{j=1}^{\ell}\beta_j'\partial_{\beta_\ell'}\widetilde\vartheta_j(r;\vec{\beta}_{\ell}')\dif\beta_\ell'\nonumber\\
	&\quad -\frac{p_5(r;\vec{\beta})}{\ii\vr_1}\sum_{j=1}^{m}q_{j,1}(r;\vec{\beta})p_{j,3}(r;\vec{\beta})
	-\frac{\widetilde p_5(r;\vec{\beta})}{\ii\vr_2}
	\sum_{j=1}^{m}\widetilde q_{j,1}(r;\vec{\beta})\widetilde p_{j,3}(r;\vec{\beta})\nonumber\\
	&\quad +\frac{2}{3\vr_1^2}p_5(r;\vec{\beta})^3-\frac{2}{3\vr_1^2}p_5(r;\vec{0})^3
+\frac{2}{3\vr_2^2}\widetilde p_5(r;\vec{\beta})^3
-\frac{2}{3\vr_2^2}\widetilde p_5(r;\vec{0})^3\nonumber\\
	&\quad +\frac{\vr_2}{\vr_1}p_5(r;\vec{\beta})q_6(r;\vec{\beta})\widetilde q_6(r;\vec{\beta})
+\frac{\vr_1}{\vr_2}\widetilde p_5(r;\vec{\beta})q_6(r;\vec{\beta})\widetilde q_6(r;\vec{\beta})\nonumber\\
	&\quad -\frac{\vr_2}{\vr_1}p_5(r;\vec{0})q_6(r;\vec{0})\widetilde q_6(r;\vec{0})
-\frac{\vr_1}{\vr_2}\widetilde p_5(r;\vec{0})q_6(r;\vec{0})\widetilde q_6(r;\vec{0})\nonumber\\
	&\quad +\frac{2}{3}r H(r;\vec{\beta})
	+\frac{1}{3}\sum_{j=1}^{m}\Big(
	p_{j,1}(r;\vec{\beta})q_{j,1}(r;\vec{\beta})+p_{j,2}(r;\vec{\beta})q_{j,2}(r;\vec{\beta})
	+\widetilde p_{j,1}(r;\vec{\beta})\widetilde q_{j,1}(r;\vec{\beta})
	+\widetilde p_{j,2}(r;\vec{\beta})\widetilde q_{j,2}(r;\vec{\beta})
	\Big)\nonumber\\
	&\quad -\frac{2}{3}p_5(r;\vec{\beta})q_5(r;\vec{\beta})
	-\frac{2}{3}\widetilde p_5(r;\vec{\beta})\widetilde q_5(r;\vec{\beta})
	+\frac{2}{3}p_6(r;\vec{\beta})q_6(r;\vec{\beta})
	+\frac{2}{3}\widetilde p_6(r;\vec{\beta})\widetilde q_6(r;\vec{\beta})\nonumber\\
	&\quad +\frac{2\vs_1}{3\vr_1}p_5(r;\vec{\beta})
	+\frac{2\vs_2}{3\vr_2}\widetilde p_5(r;\vec{\beta})
	+\frac{1}{3}\mathcal{B}(0;\vec{\beta})\nonumber\\
	&\quad -\frac{2}{3}p_6(r;\vec{0})q_6(r;\vec{0})
	-\frac{2}{3}\widetilde p_6(r;\vec{0})\widetilde q_6(r;\vec{0})
	+\Boh\left(r^{-\frac12}\log r\right),
\end{align}
where $\mathcal{B}(0;\vec{\beta})$ is given in \eqref{def: mathcalB}--\eqref{eq: mathcalB(0)}.

Using Proposition \ref{prop:coupled system-asy}, \eqref{eq: p56_zero}--\eqref{eq: q56_zero}, \eqref{M11}--\eqref{M14}, and noting that $\tau=0$, 
a long but straightforward calculation shows that the above formula simplifies to
\begin{align}\label{eq: integralH-large-r-pre}
	\int_{0}^{r}H(\xi;\vec{\beta})\dif\xi
	&=2\ii\sum_{\ell=1}^{m}\int_{0}^{\beta_\ell}\sum_{j=1}^{\ell}
	\beta_j'\partial_{\beta_\ell'}\vartheta_j(r;\vec{\beta}_{\ell}')\dif\beta_\ell'
	+2\ii\sum_{\ell=1}^{m}\int_{0}^{\beta_\ell}\sum_{j=1}^{\ell}
	\beta_j'\partial_{\beta_\ell'}\widetilde\vartheta_j(r;\vec{\beta}_{\ell}')\dif\beta_\ell'\nonumber\\
	&\quad -2\sum_{j=1}^{m}\beta_j^2
	+\sum_{j=1}^{m}\Bigg[
	\frac{4\ii}{3}(\vr_1+\vr_2)\beta_jx_j^{\frac32}r^{\frac32}
	-4\ii(\vs_1+\vs_2)\beta_jx_j^{\frac12}r^{\frac12}\Bigg]
	+\Boh\left(r^{-\frac12}\log r\right).
\end{align}
By \eqref{eq:def-theta-j-general}, we have
\begin{align}
\partial_{\beta_\ell'}\vartheta_j(r;\vec\beta_\ell')
=\begin{cases}
\ii\log\left(\dfrac{x_j^{\frac12}+x_\ell^{\frac12}}{\left|x_j^{\frac12}-x_\ell^{\frac12}\right|}\right), & j<\ell,\\[2mm]
\ii\log\left(8r^{\frac12}x_\ell^{\frac12}(r\vr_1x_\ell-\vs_1)\right)
+\partial_{\beta_\ell'}\arg\Gamma(1+\beta_\ell'), & j=\ell,\\[2mm]
\end{cases}
\end{align}
which implies that
\begin{align}\label{eq: evaluation-theta-integral}
&2\ii\sum_{\ell=1}^{m}\int_{0}^{\beta_\ell}\sum_{j=1}^{\ell}
\beta_j'\partial_{\beta_\ell'}\vartheta_j(r;\vec{\beta}_{\ell}')\dif\beta_\ell'\nonumber\\
&\quad=2\ii\sum_{j=1}^{m}\int_{0}^{\beta_j}t\partial_t\arg\Gamma(1+t)\dif t-\sum_{j=1}^{m}\beta_j^2\log\left(8r^{\frac12}x_j^{\frac12}(r\vr_1x_j-\vs_1)\right)
-2\sum_{1\le j<\ell\le m}\beta_j\beta_\ell
\log\left(\dfrac{x_j^{\frac12}+x_\ell^{\frac12}}{\left|x_j^{\frac12}-x_\ell^{\frac12}\right|}\right) \nonumber\\
&\quad=\sum_{j=1}^{m}\log\left(G(1+\beta_j)G(1-\beta_j)\right)+\sum_{j=1}^{m}\beta_j^2-\sum_{j=1}^{m}\frac{3\beta_j^2}{2}\log r\nonumber\\
&\hspace*{7em}-\sum_{j=1}^{m}\beta_j^2\log(8\vr_1x_j^{\frac32})-2\sum_{1\le j<\ell\le m}\beta_j\beta_\ell\log\left(\dfrac{x_j^{\frac12}+x_\ell^{\frac12}}{\left|x_j^{\frac12}-x_\ell^{\frac12}\right|}\right)+\Boh(r^{-1}).
\end{align}
Similarly, by \eqref{def:tildeX} and \eqref{eq:def-theta-j-general},
\begin{align}\label{eq: evaluation-tildetheta-integral}
&2\ii\sum_{\ell=1}^{m}\int_{0}^{\beta_\ell}\sum_{j=1}^{\ell}
\beta_j'\partial_{\beta_\ell'}\widetilde\vartheta_j(r;\vec{\beta}_{\ell}')\dif\beta_\ell'\nonumber\\
&\quad=\sum_{j=1}^{m}\log\left(G(1+\beta_j)G(1-\beta_j)\right)+\sum_{j=1}^{m}\beta_j^2-\sum_{j=1}^{m}\frac{3\beta_j^2}{2}\log r\nonumber\\
&\hspace*{7em}-\sum_{j=1}^{m}\beta_j^2\log(8\vr_2x_j^{\frac32})-2\sum_{1\le j<\ell\le m}\beta_j\beta_\ell\log\left(\dfrac{x_j^{\frac12}+x_\ell^{\frac12}}{\left|x_j^{\frac12}-x_\ell^{\frac12}\right|}\right)+\Boh(r^{-1}).
\end{align}
Substituting \eqref{eq: evaluation-theta-integral} and \eqref{eq: evaluation-tildetheta-integral} into \eqref{eq: integralH-large-r-pre}, we arrive at
\begin{multline}
	\int_{0}^{r}H(\xi;\vec{\beta})\dif\xi
	=\sum_{j=1}^{m}\frac{4\ii}{3}(\vr_1+\vr_2)\beta_jx_j^{\frac32}r^{\frac32}
	-\sum_{j=1}^{m}4\ii(\vs_1+\vs_2)\beta_jx_j^{\frac12}r^{\frac12}
	-3\sum_{j=1}^{m}\beta_j^2\log r \\
	+2\sum_{j=1}^{m}\log\left(G(1+\beta_j)G(1-\beta_j)\right)
	-\sum_{j=1}^{m}\beta_j^2\log\left(64\vr_1\vr_2x_j^3\right)
	-4\sum_{1\le j<\ell\le m}\beta_j\beta_\ell\log\left(\dfrac{x_j^{\frac12}+x_\ell^{\frac12}}{\left|x_j^{\frac12}-x_\ell^{\frac12}\right|}\right)
	+\Boh\left(r^{-\frac12}\log r\right).
\end{multline}
Inserting this into \eqref{eq:integral representation via Hamiltonian} and using \eqref{def: first def beta_j}, we 
finally obtain the large $r$ asymptotics of $F(r\vec{x},\vec{u})$ as stated in \eqref{eq: large_gap_asymptotics_generating_function}.
The fact that the error term in \eqref{eq: large_gap_asymptotics_generating_function} is $\Boh(r^{-\frac12})$ instead of $\Boh(r^{-\frac12}\log r)$ 
can be justified by integrating \eqref{eq:H-large-r} on both sides with respect to $r$ from a sufficiently large positive constant to $r>0$, and 
\eqref{eq:derivative of error term} follows from \eqref{eq:derivative-R}.

This completes the proof of Theorem \ref{thm: large_gap}. \qed

\section*{Acknowledgements}
This work is partially supported by China Postdoctoral Science Foundation under grant
number 2024M760480 and Shanghai Post-Doctoral Excellence Program under grant number 2024100. 
The author is deeply grateful to Lun Zhang for generously sharing the history of the tacnode process and related works, 
for carefully reading the manuscript, and for providing many valuable and insightful comments. 
The author also thanks Luming Yao for her helpful comments, particularly on clarifying 
the typo that is pointed out in Remark \ref{remark:typo}.
\vspace{2mm}

\appendix

\section{Proof of the Hamiltonian formulation}\label{app:verification-Hamiltonian}
In this appendix, we prove that the Hamiltonian \eqref{eq:Hamiltonian} yields the equations
\eqref{eq:Hamiltonian-derivative-1} and \eqref{eq:Hamiltonian-derivative-2}, and hence the system
\eqref{eq:coupled system}.

We first consider the variables $p_5, q_5, p_6, q_6$. By differentiating \eqref{eq:Hamiltonian}
term by term, we obtain
\begin{align*}
\frac{\partial H}{\partial p_5}
&=\sum_{j=1}^{m}x_j\Big(
q_{j,1}(r)p_{j,1}(r)-q_{j,3}(r)p_{j,3}(r)
-\widetilde q_{j,2}(r)\widetilde p_{j,2}(r)+\widetilde q_{j,4}(r)\widetilde p_{j,4}(r)
\Big)
=q_5'(r),\\
-\frac{\partial H}{\partial q_5}
&=-\ii \vr_1\sum_{j=1}^{m}x_j\Big(
q_{j,1}(r)p_{j,3}(r)+\widetilde q_{j,2}(r)\widetilde p_{j,4}(r)
\Big)
=p_5'(r),\\
\frac{\partial H}{\partial p_6}
&=-\sum_{j=1}^{m}x_j\Big(
q_{j,1}(r)p_{j,4}(r)+\widetilde q_{j,2}(r)\widetilde p_{j,3}(r)
\Big)
=q_6'(r),\\
-\frac{\partial H}{\partial q_6}
&=\ii\sum_{j=1}^{m}x_j\Big(
\vr_2 q_{j,2}(r)p_{j,1}(r)+\vr_1 q_{j,4}(r)p_{j,3}(r)
-\vr_2 \widetilde q_{j,1}(r)\widetilde p_{j,2}(r)-\vr_1 \widetilde q_{j,3}(r)\widetilde p_{j,4}(r)
\Big)
=p_6'(r).
\end{align*}
This proves \eqref{eq:Hamiltonian-derivative-1}.

We next turn to the variables $q_{j,k}(r)$ and $p_{j,k}(r)$. Write the interaction part of
\eqref{eq:Hamiltonian} as
\begin{align*}
H_{\rm int}
:=&\sum_{a=1}^{m}\frac{x_a}{r}\sum_{\substack{b=1\\b\neq a}}^{m}\left(
\frac{S_{ab}(r)S_{ba}(r)+\widetilde S_{ab}(r)\widetilde S_{ba}(r)}{x_a-x_b}
+\frac{\Upsilon_{ab}(r)\widetilde\Upsilon_{ba}(r)+\Upsilon_{ba}(r)\widetilde\Upsilon_{ab}(r)}{x_a+x_b}
\right)\\
&\quad+\frac{1}{r}\sum_{a=1}^{m}\Upsilon_{aa}(r)\widetilde\Upsilon_{aa}(r).
\end{align*}
By the definitions \eqref{def:S_jk} and \eqref{def:Upsilon_jk}, together with \eqref{def:tildeX}, we obtain
\begin{align*}
\frac{\partial S_{ab}}{\partial p_{j,\ell}}&=\delta_{bj}q_{a,\ell},
&\frac{\partial S_{ab}}{\partial q_{j,\ell}}&=\delta_{aj}p_{b,\ell},\\
\frac{\partial \widetilde\Upsilon_{ab}}{\partial p_{j,\ell}}&=\delta_{bj}\times
\begin{cases}
\widetilde q_{a,2}, & \ell=1,\\
\widetilde q_{a,1}, & \ell=2,\\
-\widetilde q_{a,4}, & \ell=3,\\
-\widetilde q_{a,3}, & \ell=4,
\end{cases}
&\frac{\partial \Upsilon_{ab}}{\partial q_{j,\ell}}&=\delta_{aj}\times
\begin{cases}
\widetilde p_{b,2}, & \ell=1,\\
\widetilde p_{b,1}, & \ell=2,\\
-\widetilde p_{b,4}, & \ell=3,\\
-\widetilde p_{b,3}, & \ell=4.
\end{cases}
\end{align*}
For example, differentiating $H_{\rm int}$ with respect to $p_{j,1}$ gives
\begin{align*}
\frac{\partial H_{\rm int}}{\partial p_{j,1}}
&=\frac{1}{r}\sum_{\substack{k=1\\k\neq j}}^{m}
\left(
\frac{x_jS_{jk}(r)q_{k,1}(r)}{x_j-x_k}
+\frac{x_kq_{k,1}(r)S_{jk}(r)}{x_k-x_j}
\right)\\
&\quad+\frac{1}{r}\sum_{\substack{k=1\\k\neq j}}^{m}
\left(
\frac{x_j\Upsilon_{jk}(r)\widetilde q_{k,2}(r)}{x_j+x_k}
+\frac{x_k\Upsilon_{jk}(r)\widetilde q_{k,2}(r)}{x_j+x_k}
\right)
+\frac{1}{r}\Upsilon_{jj}(r)\widetilde q_{j,2}(r)\\
&=\frac{1}{r}\sum_{k=1}^{m}S_{jk}(r)q_{k,1}(r)
+\frac{1}{r}\sum_{k=1}^{m}\Upsilon_{jk}(r)\widetilde q_{k,2}(r).
\end{align*}
Hence
\begin{align*}
\frac{\partial H}{\partial p_{j,1}}
&=x_j\Big((p_5(r)+\tau)q_{j,1}(r)-\ii\vr_2 q_6(r)q_{j,2}(r)+\ii\vr_1 q_{j,3}(r)\Big)\\
&\quad+\frac{1}{r}\sum_{k=1}^{m}S_{jk}(r)q_{k,1}(r)
+\frac{1}{r}\sum_{k=1}^{m}\Upsilon_{jk}(r)\widetilde q_{k,2}(r)
=q_{j,1}'(r).
\end{align*}
Similarly,
\begin{align*}
-\frac{\partial H_{\rm int}}{\partial q_{j,1}}
&=-\frac{1}{r}\sum_{\substack{k=1\\k\neq j}}^{m}
\left(
\frac{x_jp_{k,1}(r)S_{kj}(r)}{x_j-x_k}
+\frac{x_kS_{kj}(r)p_{k,1}(r)}{x_k-x_j}
\right)\\
&\quad-\frac{1}{r}\sum_{\substack{k=1\\k\neq j}}^{m}
\left(
\frac{x_j\widetilde p_{k,2}(r)\widetilde\Upsilon_{kj}(r)}{x_j+x_k}
+\frac{x_k\widetilde p_{k,2}(r)\widetilde\Upsilon_{kj}(r)}{x_j+x_k}
\right)
-\frac{1}{r}\widetilde p_{j,2}(r)\widetilde\Upsilon_{jj}(r)\\
&=-\frac{1}{r}\sum_{k=1}^{m}S_{kj}(r)p_{k,1}(r)
-\frac{1}{r}\sum_{k=1}^{m}\widetilde\Upsilon_{kj}(r)\widetilde p_{k,2}(r),
\end{align*}
and therefore
\begin{align*}
-\frac{\partial H}{\partial q_{j,1}}
&=-\ii\vr_1 r x_j^2p_{j,3}(r)
-x_j\Big((p_5(r)+\tau)p_{j,1}(r)+\ii\vr_1\widetilde q_6(r)p_{j,2}(r)\\
&\qquad\qquad\qquad+(\ii\vr_1q_5(r)-\ii\vs_1)p_{j,3}(r)-p_6(r)p_{j,4}(r)\Big)\\
&\quad-\frac{1}{r}\sum_{k=1}^{m}S_{kj}(r)p_{k,1}(r)
-\frac{1}{r}\sum_{k=1}^{m}\widetilde\Upsilon_{kj}(r)\widetilde p_{k,2}(r)
=p_{j,1}'(r).
\end{align*}
The computations for $q_{j,2}, q_{j,3}, q_{j,4}$ and $p_{j,2}, p_{j,3}, p_{j,4}$ are identical,
the only difference being the signs coming from the definition of $\Upsilon_{jk}(r)$. This proves
\eqref{eq:Hamiltonian-derivative-2}.

\section{The confluent hypergeometric parametrix}\label{app:CH-parametrix}
The confluent hypergeometric parametrix $\Phi_{\CH}(z)=\Phi_{\CH}(z;\beta)$ with $\beta$ being a parameter is a solution of the following RH problem.

\subsection*{RH problem for $\Phi_{\CH}$}
\begin{description}
  \item(a)   $\Phi_{\CH}(z)$ is analytic in $\mathbb{C}\setminus \{\cup^6_{j=1}\widehat\Sigma_j\cup\{0\}\}$, where the contours $\widehat\Sigma_j$, $j=1,\ldots,6,$ are shown in Figure \ref{fig:jumps-Phi-C}.

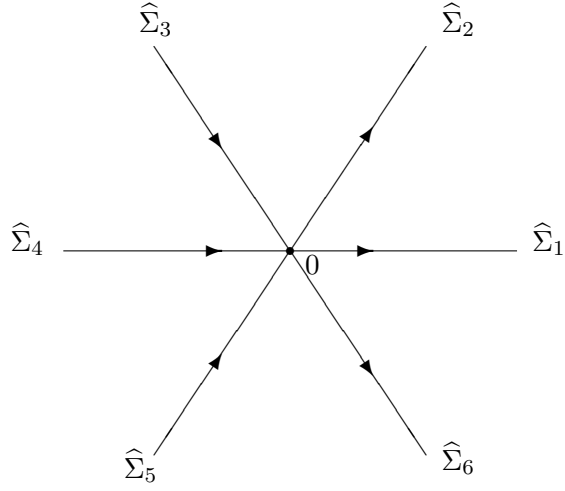
\begin{figure}[h]
\begin{center}
   \setlength{\unitlength}{1truemm}
   \begin{picture}(100,70)(-5,2)
       \put(40,40){\line(-2,-3){18}}
       \put(40,40){\line(-2,3){18}}
       \put(40,40){\line(-1,0){30}}
       \put(40,40){\line(1,0){30}}
      \put(40,40){\line(2,-3){18}}
      \put(40,40){\line(2,3){18}}

       \put(30,55){\thicklines\vector(2,-3){1}}
       \put(30,40){\thicklines\vector(1,0){1}}
       \put(50,40){\thicklines\vector(1,0){1}}
       \put(30,25){\thicklines\vector(2,3){1}}
      \put(50,25){\thicklines\vector(2,-3){1}}
       \put(50,55){\thicklines\vector(2,3){1}}


       \put(42,36.9){$0$}
         \put(72,40){$\widehat \Sigma_1$}
           \put(60,69){$\widehat \Sigma_2$}
             \put(20,69){$\widehat \Sigma_3$}
              \put(3,40){$\widehat \Sigma_4$}
       \put(18,10){$\widehat \Sigma_5$}
          \put(60,11){$ \widehat \Sigma_6$}

%

       \put(40,40){\thicklines\circle*{1}}
\end{picture}
   \caption{The jump contours for the RH problem for $\Phi_{\CH}$.}
   \label{fig:jumps-Phi-C}
\end{center}
\end{figure}

  \item(b) $\Phi_{\CH}$ satisfies the jump condition
  \begin{equation}\label{HJumps}
  \Phi_{\CH,+}(z)=\Phi_{\CH,-}(z) \widehat J_j(z), \quad z \in \widehat\Sigma_j,\quad j=1,\ldots,6,
  \end{equation}
  where
  \begin{equation*}
 \widehat J_1(z) = \begin{pmatrix}
    0 &   e^{-\beta \pi \ii} \\
    -  e^{\beta \pi \ii} &  0
    \end{pmatrix}, \qquad \widehat J_2(z) = \begin{pmatrix}
    1 & 0 \\
    e^{ \beta \pi \ii } & 1
    \end{pmatrix}, \qquad
    \widehat J_3(z) = \begin{pmatrix}
    1 & 0 \\
    e^{ -\beta\pi \ii} & 1
    \end{pmatrix},           
  \end{equation*}
  \begin{equation*}
  \widehat J_4(z) = \begin{pmatrix}
    0 &   e^{\beta\pi \ii} \\
     -  e^{-\beta\pi \ii} &  0
     \end{pmatrix}, \qquad
      \widehat J_5(z) = \begin{pmatrix}
     1 & 0 \\
     e^{- \beta\pi \ii} & 1
     \end{pmatrix},\qquad
     \widehat J_6(z) = \begin{pmatrix}
   1 & 0 \\
   e^{\beta\pi \ii} & 1
   \end{pmatrix}.
  \end{equation*}

  \item(c) As $z\to \infty$, $\Phi_{\CH}(z)$ admits the following asymptotic behavior:
  \begin{align}\label{H at infinity}
 \Phi_{\CH}(z)=\left(I + \frac{\Phi_{\CH, 1}(\beta)}{z}+\Boh(z^{-2})\right) z^{-\beta \sigma_3}e^{-\frac{\ii z}{2}\sigma_3}
  \left\{\begin{array}{ll}
                         I, & ~0< \arg z <\pi,
                         \\
                        \begin{pmatrix}
                                                             0 &   -e^{\beta\pi \ii} \\
                                                            e^{-\beta\pi \ii } &  0
                        \end{pmatrix}, &~ \pi< \arg z<\frac{3\pi}{2},
                        \\
                        \begin{pmatrix}
                        0 &   -e^{-\beta\pi \ii} \\
                        e^{\beta\pi \ii} &  0
                        \end{pmatrix}, & -\frac{\pi}{2}<\arg z<0.
 \end{array}\right.
\end{align}
\item(d) As $z\to 0$, we have $\Phi_{\CH}(z)=\Boh(\log |z|)$.
\end{description}

By \cite{IK}, the above RH problem admits an explicit solution in terms of confluent hypergeometric functions, although the explicit formula for $\Phi_{\CH}(z)$ will not be needed here.
Moreover, as $z \to 0$, we have
\begin{equation}\label{eq:H-expand-2}
\Phi_{\CH}(z) e^{-\frac{\beta \pi \ii}{2} \sigma_3} = \Phi_{\CH}^{(0)}(\beta)\left( I+ \Phi_{\CH}^{(1)}(\beta)z+\Boh(z^2) \right) \begin{pmatrix} 1 & -\frac{\kappa}{2\pi \ii} \log (e^{-\frac{\pi \ii}{2}}z) \\
0 & 1  \end{pmatrix},
\end{equation}
for $z$ belonging to the region bounded by the rays $\widehat \Sigma_2$ and $\widehat \Sigma_3$, where $\kappa=1-e^{2\beta \pi \ii}$,
\begin{align}\label{eq:H-expand-coeff-0}
\Phi_{\CH}^{(0)}(\beta)
=\begin{pmatrix} \Gamma\left(1-\beta\right) e^{-\beta \pi \ii} &\frac{1}{\Gamma(\beta)} \left( \frac{\Gamma'\left(1-\beta\right)}{\Gamma\left(1-\beta\right)} +2\gamma_{\textrm{E}} \right) \vspace{5pt} \\
\Gamma\left(1+\beta\right) & -\frac{e^{\beta \pi \ii}}{\Gamma(-\beta)} \left( \frac{\Gamma'\left(-\beta\right)}{\Gamma\left(-\beta\right) } +2\gamma_{\textrm{E}}\right) \end{pmatrix}
\end{align}
with $\gamma_{\textrm{E}}\approx 0.57721$ being the Euler's constant,
and
 \begin{equation}\label{eq:H-expand-coeff-1}
(\Phi_{\CH}^{(1)})_{21}=\frac{ \beta \pi \ii \, e^{-\beta \pi \ii} }{\sin(\beta \pi )}.
\end{equation}


\end{document}